\documentclass{new_tlp}

\usepackage{amsmath}
\usepackage{mathptmx}
\usepackage{amssymb}

\usepackage{amsopn} % required for \DeclareMathOperator
\usepackage{IEEEtrantools}

\usepackage[cmtip,arrow]{xy}
\usepackage[curve]{xypic}

\usepackage{chngcntr}
\usepackage{titlesec}

\DeclareMathAlphabet\mathbfcal{OMS}{cmsy}{b}{n}
\makeatletter
\renewcommand\appendix{\inappendixtrue\par
  \@addtoreset{equation}{section}%
  \@addtoreset{figure}{section}%
  \@addtoreset{table}{section}%
  \setcounter{section}\z@
  \renewcommand\thesection{\@Alph\c@section}%
  \renewcommand\thesubsection{\@Alph\c@section.\@arabic\c@subsection}
  \renewcommand\theequation{\@Alph\c@section\@arabic\c@equation}%
  \renewcommand\thefigure{\@Alph\c@section\,\@arabic\c@figure}%
  \renewcommand\thetable{\@Alph\c@section\,\@arabic\c@table}%
  \renewcommand{\@seccntformat}[1]{Appendix \csname the##1\endcsname.\quad}%
}
\makeatother

\newtheorem{definition}{Definition}
\newtheorem{theorem}{Theorem}
\newtheorem{example}{Example}
\newtheorem{lemma}{Lemma}[section]

\newtheorem{proposition}{Proposition}

\newenvironment{proofof}[1]{\vspace{5pt}\setlength{\parindent}{0cm}\setlength{\parskip}{0.2cm} {\bf Proof of #1}.}{}
\renewenvironment{proof}{\setlength{\parindent}{0cm}\setlength{\parskip}{0.2cm} {\emph{Proof}}.}{\vspace{0.5cm}}

\usepackage[usenames,dvipsnames]{xcolor}

\newcommand{\review}[2]{#2}

\newcommand{\sidecomment}[1]%
                    {\marginpar{\footnotesize\emph{\color{blue} #1}}}

\newcommand\qed{~\hfill$\Box$}
\newcommand{\tuple}[1]{\ensuremath{\langle #1 \rangle}}

\newcommand{\set}[1]{\ensuremath{\{#1\}}}

\newcommand{\setm}[2]{\ensuremath{\{\ #1\ \big|\ #2\ \}}}

\newcommand{\eqdef}{%
  \mathrel{\vbox{\offinterlineskip\ialign{%
    \hfil##\hfil\cr%
    $\scriptscriptstyle\mathrm{def}$\cr%
    \noalign{\kern1pt}%
    $=$\cr%
    \noalign{\kern-0.1pt}%
}}}}

\DeclareMathOperator{\Not}{\it not}
\DeclareMathOperator{\Undef}{\it undef}
\DeclareMathOperator{\sneg}{\sim\!}
\DeclareMathOperator{\cdotl}{\!\cdot\!}

\newcommand{\at}{At}
\newcommand{\lb}{Lb}
\newcommand{\signature}{\tuple{\at,\lb}}

\newcommand{\evalues}{\ensuremath{\mathbf{V}_{\lb}}}
\newcommand{\graphs}{\ensuremath{\mathbf{G}_{\lb}}}

\newcommand{\cgvalues}{\ensuremath{\mathbf{I}_{\lb}^{CG}}}
\newcommand{\cgvaluesterms}{\ensuremath{\mathbf{V}_{\lb}^{CG}}}
\newcommand{\values}{\cgvaluesterms}

\newcommand{\rR}{\ensuremath{R}}
\newcommand{\rA}{\ensuremath{A}}
\newcommand{\rH}{\ensuremath{H}}
\newcommand{\rB}{\ensuremath{B}}
\newcommand{\rC}{\ensuremath{C}}
\newcommand{\rL}{\ensuremath{L}}

\newcommand{\cP}{{\ensuremath{P}}}
\newcommand{\cQ}{{\ensuremath{Q}}}
\newcommand{\wP}{{\ensuremath{\mathfrak{P}}}}

\newcommand{\wwP}{\mathcal{P}}

\newcommand{\cI}{{\ensuremath{\tilde{I}}}}

\newcommand{\wI}{{\ensuremath{\mathfrak{I}}}}
\newcommand{\wJ}{{\ensuremath{\mathfrak{J}}}}

\newcommand{\eI}{{\ensuremath{I}}}
\newcommand{\eJ}{{\ensuremath{J}}}

\newcommand{\eU}{{\ensuremath{U}}}

\newcommand{\sI}{{\ensuremath{\hat{I}}}}

\newcommand{\botI}{\ensuremath{\mathbf{0}}}
\newcommand{\topI}{\ensuremath{\mathbf{1}}}

\newcommand{\ereduct}[2]{{\ensuremath{#1^{#2}}}}

\newcommand{\eWpP}[1]{\ensuremath{\Gamma_{#1}}}
\newcommand{\eWp}{\eWpP{\cP}}

\newcommand{\eWWpP}[1]{\ensuremath{\Gamma^2_{#1}}}

\newcommand{\eWprPI}[3]{{\ensuremath{{\eWpP{#1}^2}\!\uparrow^{#2}({#3})}}}
\newcommand{\eWprI}[2]{\eWprPI{\cP}{#1}{#2}}
\newcommand{\eWprP}[2]{\eWprPI{#1}{#2}{\botI}}
\newcommand{\eWpr}{\eWprP{\cP}}

\newcommand{\eWWprPI}[3]{{\ensuremath{\eWWpP{#1}\!\uparrow^{#2}({#3})}}}

\newcommand{\eWWprP}[2]{\eWWprPI{#1}{#2}{\botI}}
\newcommand{\eWWpr}{\eWWprP{\cP}}

\newcommand{\wlfpP}[1]{{\ensuremath{\mathfrak{T}_{#1}}}}
\newcommand{\wgfpP}[1]{{\ensuremath{\mathfrak{TU}_{#1}}}}

\newcommand{\elfpP}[1]{{\ensuremath{\mathbb{L}_{#1}}}}
\newcommand{\egfpP}[1]{{\ensuremath{\mathbb{U}_{#1}}}}

\newcommand{\elfp}{{\elfpP{\cP}}}
\newcommand{\egfp}{{\egfpP{\cP}}}

\newcommand{\ewfmP}[1]{\ensuremath{\mathbb{W}_{#1}}}
\newcommand{\ewfm}{\ewfmP{\cP}}
\newcommand{\wfm}{\ewfmP{\cP}}

\newcommand{\lambdac}{\ensuremath{\lambda^{c}}}
\newcommand{\lambdap}{\ensuremath{\lambda^{p}}}
\newcommand{\lambdaq}{\ensuremath{\lambda^{q}}}

\newcommand{\wLb}{\ensuremath{\lb}}
\newcommand\boolAlgebra{\ensuremath{\mathbf{B}_{\wLb}}}

\newcommand{\etpP}[1]{\ensuremath{T_{#1}}}
\newcommand{\etp}{\etpP{\cP}}

\newcommand{\etprP}[2]{{\ensuremath{\etpP{#1}\!\uparrow^{#2}(\botI)}}}
\newcommand{\etpr}[1]{{\ensuremath{\etp\!\uparrow^{#1}(\botI)}}}

\newcommand{\WpP}[1]{\ensuremath{\tilde{\Gamma}_{#1}}}
\newcommand{\Wp}{\WpP{\cP}}

\newcommand{\tpP}[1]{\ensuremath{\tilde{T}_{#1}}}
\newcommand{\tp}{\tpP{P}}

\newcommand{\tprP}[2]{{\ensuremath{\tpP{#1}\!\uparrow^{#2}(\botI)}}}
\newcommand{\tpr}[1]{{\ensuremath{\tp\!\uparrow^{#1}(\botI)}}}

\newcommand{\wtpP}[1]{\ensuremath{\mathfrak{T}_{#1}}}

\newcommand{\wbotI}{\ensuremath{\mathbf{\bot}}}

\newcommand{\wtprPI}[3]{{\ensuremath{\wtpP{#1}\!\uparrow^{#2}({#3})}}}
\newcommand{\wtprP}[2]{\wtprPI{#1}{#2}{\wbotI}}

\newcommand{\wWpP}[1]{\ensuremath{\mathfrak{G}_{#1}}}

\newcommand{\wWprPI}[3]{\ensuremath{\mathfrak{G}^2_{#1}\!\uparrow^{#2}({#3})}}

\newcommand{\wWprP}[2]{\wWprPI{#1}{#2}{\botI}}

\newcommand{\stpP}[1]{\ensuremath{\hat{T}_{#1}}}
\newcommand{\stp}{\stpP{\cP}}

\newcommand{\stprP}[2]{{\ensuremath{\stpP{#1}\!\uparrow^{#2}(\botI)}}}

\newcommand{\sWpP}[1]{\ensuremath{\hat{\Gamma}_{#1}}}
\newcommand{\sWp}{\sWpP{\cP}}

\newcommand{\sWprPI}[3]{{\ensuremath{{\sWpP{#1}^2}\!\uparrow^{#2}({#3})}}}

\newcommand{\sWprP}[2]{\sWprPI{#1}{#2}{\botI}}
\newcommand{\sWpr}{\sWprP{\cP}}

\def\gfp{\hbox{\textnormal{gfp}}}
\def\lfp{\hbox{\textnormal{lfp}}}

\newcounter{programcount}
\newcommand{\newprogram}{\refstepcounter{programcount}\ensuremath{P_{\arabic{programcount}}}}
\newcommand{\programref}[1]{{\ensuremath{\cP_{\ref{#1}}}}}

\begin{document}
\bibliographystyle{acmtrans}

 \submitted{October 21 2015}
 \revised{January 26 2016}
 \accepted{February 22 2016}

\title[Enablers and Inhibitors in Causal Justifications of Logic Programs]{Enablers and Inhibitors in \\ Causal Justifications of Logic Programs\footnote{This is an extended version of a paper presented at the Logic Programming
and Nonmonotonic Reasoning Conference (LPNMR 2015), invited as a rapid
communication in TPLP. The authors acknowledge the assistance of
the conference program chairs Giovambattista Ianni and Miroslaw
Truszczynski.}}

\author[P. Cabalar \& J. Fandinno]
         {Pedro Cabalar and Jorge Fandinno\\
          Department of Computer Science\\
		  University of Corunna, Spain\\
		  \email{\{cabalar, jorge.fandino\}@udc.es}}

\pagerange{\pageref{firstpage}--\pageref{lastpage}}
\volume{\textbf{10} (3):}
\jdate{October 2015}
\setcounter{page}{1}
\pubyear{2015}

\maketitle

\label{firstpage}

\setcounter{page}{1}
\pagestyle{headings}

\begin{abstract}
\emph{To appear in Theory and Practice of Logic Programming (TPLP).}
In this paper we propose an extension of logic programming (LP) where each default literal derived from the well-founded model is associated to a justification represented as an algebraic expression. This expression contains both causal explanations (in the form of proof graphs built with rule labels) and terms under the scope of negation that stand for conditions that enable or disable the application of  causal rules. Using some examples, we discuss how these new conditions, we respectively call \emph{enablers} and \emph{inhibitors}, are intimately related to default negation and have an essentially different nature from regular cause-effect relations. The most important result is a formal comparison to the recent algebraic approaches for justifications in LP: \emph{Why-not Provenance}~(WnP) and \emph{Causal Graphs}~(CG). We show that the current approach extends both WnP and CG justifications under the Well-Founded Semantics and, as a byproduct, we also establish a formal relation between these two approaches.
\end{abstract}


%%%%%%%%%%%%%%%%%%%%%%%%%%%%%%%%%%%%%%%%%%%%%%%%%%%%%%%%%%%%%%%%%%%%%%%%%%%%%%%%%%%%%%%%%%%
%%%%%%%%%%%%%%%%%%%%%%%%%%%%%%%%%%%%%%%%%%%%%%%%%%%%%%%%%%%%%%%%%%%%%%%%%%%%%%%%%%%%%%%%%%%
%%%%%%%%%%%%%%%%%%%%%%%%%%%%%%%%%%%%%%%%%%%%%%%%%%%%%%%%%%%%%%%%%%%%%%%%%%%%%%%%%%%%%%%%%%%

\section{Introduction}\label{sc:introduction}

The strong connection between Non-Monotonic Reasoning~(NMR) and Logic Programming~(LP) semantics for default negation has made possible that LP tools became nowadays an important paradigm for Knowledge Representation~(KR) and problem-solving in Artificial Intelligence~(AI). In particular, \emph{Answer Set Programming}~(ASP)  \cite{niemela1999,MT99} has \review{R1.1}{established} as a preeminent LP paradigm for practical NMR with applications in diverse areas of AI including planning, reasoning about actions, diagnosis, abduction and beyond. The ASP paradigm is based on the \emph{stable models semantics}~\cite{GelfondL88} and is also closely related to the other mainly accepted interpretation for default negation, \emph{well-founded} semantics (WFS)~\cite{van1991well}. One interesting difference between these two LP semantics and classical models (or even other NMR approaches) is that true atoms in LP must be founded or justified by a given derivation. These \emph{justifications} are not provided in the semantics itself, but can be syntactically built in some way in terms of the program rules, as studied in several approaches~\cite{specht1993generating,denecker1993justification,pemmasani2004online,GPST08,pontelli2009justifications,oetsch2010,schultz2013aba}.
%roychoudhury2000justifying

Rather than manipulating justifications as mere syntactic objects, two recent approaches have considered extended multi-valued semantics for LP where justifications are treated as \emph{algebraic} constructions: 
\emph{Why-not Provenance} (WnP)~\cite{damasio2013justifications} and \emph{Causal Graphs} (CG)~\cite{CabalarFF14}. Although these two approaches present formal similarities, they start from different understandings of the idea of justification. On the one hand, WnP answers the query ``why literal $L$ might hold'' by providing conjunctions of \emph{hypothetical modifications} on the program that would allow deriving $L$. These modifications include rule labels, expressions like $not(A)$ with $A$ an atom, or negations `$\neg$' of the two previous cases. As an example, a justification for $L$ like $r_1 \wedge not(p) \wedge \neg r_2 \wedge \neg not(q)$ means that the presence of rule $r_1$ and the absence of atom~$p$ would allow deriving $L$  (hypothetically) if both rule $r_2$ were removed and atom $q$ were added to the program. If we want to explain why $L$ \emph{actually} holds, we have to restrict to justifications without `$\neg$', that is, those without program modifications (which will be the focus of this paper). 

On the other hand, CG-justifications start from identifying program rules as \emph{causal laws} so that, for instance, $(p \leftarrow q)$ can be read as ``event $q$ \emph{causes} effect $p$.'' Under this viewpoint, (positive) rules offer a natural way for capturing the concept of \emph{causal production}, i.e. a continuous chain of events that has helped to cause or produce an effect~\cite{hall2004,hall2007structural}. The explanation of a true atom is made in terms of graphs formed by rule labels that reflect the ordered rule applications required for deriving that atom. These graphs are obtained by algebraic operations exclusively applied on the positive part of the program. Default negation in CG is understood as absence of cause and, consequently, a false atom has \emph{no justification}. 

The explanation of an atom $A$ in CG is more detailed than in WnP, since the former contains graphs that correspond to all relevant proofs of $A$ whereas in WnP we just get conjunctions that do not reflect any particular ordering among rule applications. However, as explained before, CG does not capture the effect of default negation in a given derivation and, sometimes, this information is very valuable, especially if we want to answer questions of the form ``why not.''

\review{R2.1}{
As in the previous paper on CG~\cite{CabalarFF14}, our final goal is to achieve an elaboration tolerant representation of causality that allows reasoning about cause-effect relations.
Under this perspective, although WnP is more oriented to program debugging, its possibility of dealing with hypothetical reasoning of the form ``why not'' would be an interesting feature to deal with counterfactuals, since several approaches to causality (see Section~\ref{sec:contributory.causes}) are based on this concept.} To understand the kind of problems we are interested in, consider the following example. A drug~$d$ in James Bond's drink causes his paralysis $p$ provided that he was not given an antidote $a$ that day. We know that Bond's enemy, Dr. No, poured the drug:
\begin{eqnarray}
p & \leftarrow & d,\, \Not a  \label{r1}\\
d \label{r2}
\end{eqnarray}
\noindent In this case it is obvious that $d$ causes $p$, whereas the absence of $a$ just \emph{enables} the application of the rule. Now, suppose we are said that Bond is 
daily administered an antidote by the MI6, unless it is a holiday $h$:
\begin{eqnarray}
a \leftarrow \Not h \label{r3}
\end{eqnarray}
\review{R1.2, R3.2}{
Adding this rule makes $a$ become an \emph{inhibitor} of $p$, as it prevents $d$ to cause $p$ by rule~\eqref{r1}.} But suppose now that we are in a holiday, that is, fact $h$ is added to the program \eqref{r1}-\eqref{r3}. Then, the inhibitor $a$ is \emph{disabled} and $d$ causes $p$ again. However, we do not consider that the holiday $h$ is a (productive) cause for Bond's paralysis $p$ although, indeed, the latter counterfactually depends on the former: ``had not been a holiday $h$, Bond would have not been paralysed.'' \review{R1.2, R3.2}{We will say that the fact $h$, which disables inhibitor $a$, is an \emph{enabler} of $p$, as it allows applying rule~\eqref{r1}.}
% In the causality literature, the role of $r$ is usually referred as a \emph{preemption}, an event that could inhibit the application of a casual law, but is not the part of the effective cause of some result.
% I do not well understand this sentence but I think is odd
% an event that fulfills some precondition to apply a causal law, but is not the final effective cause of some result.

In this work we propose dealing with these concepts of enablers and inhibitors by augmenting CG justifications with a new negation operator~`$\sneg$' in the CG causal algebra. We show that this new \review{R3.4}{approach, which we call} \emph{Extended Causal Justifications} (ECJ), captures WnP justifications under the Well-founded Semantics, establishing a formal relation between WnP and CG as a byproduct.

The rest of the paper is structured as follows. The next section  defines the new approach. Sections~\ref{sec:CG} and~\ref{sec:WnP} explain the formal relations to CG and WnP through a running example. 
\review{R3.3}{Section~\ref{sec:contributory.causes} studies several examples of causal scenarios from the literature
and finally, Section~\ref{sec:conc} concludes the paper.
Appendix~\ref{sec:figs} contains an auxiliary figure depicting some common algebraic properties and Appendix~\ref{sec:proofs} contains the formal proofs of theorems from the previous sections.}

%%%%%%%%%%%%%%%%%%%%%%%%%%%%%%%%%%%%%%%%%%%%%%%%%%%%%%%%%%%%%%%%%%%%%%%%%%%%%%%%%%%%%%%%%
%%%%%%%%%%%%%%%%%%%%%%%%%%%%%%%%%%%%%%%%%%%%%%%%%%%%%%%%%%%%%%%%%%%%%%%%%%%%%%%%%%%%%%%%%
%%%%%%%%%%%%%%%%%%%%%%%%%%%%%%%%%%%%%%%%%%%%%%%%%%%%%%%%%%%%%%%%%%%%%%%%%%%%%%%%%%%%%%%%%

\section{Extended Causal Justifications (ECJ)}
\label{sec:CP}

% We define next the \emph{Causal-Preemption} (CP) approach as follows.
A \emph{signature} is a pair \signature\ of sets that respectively represent \emph{atoms} (or \emph{propositions}) and \emph{labels}. Intuitively, each atom in $At$ will be assigned justifications built with rule labels from $Lb$. In principle, the intersection $At \cap Lb$ does not need to be empty: we may sometimes find \review{R1.3}{it} convenient to label a rule using an atom name (normally, the head atom). Justifications will be expressions that combine four different algebraic operators: a product `$*$' representing conjunction or joint causation; a sum `$+$' representing alternative causes; a non-commutative product `$\cdot$' that captures the sequential order that follows from rule applications; and a non-classical negation `$\sneg$' which will precede inhibitors (negated labels) and enablers (doubly negated labels).

\begin{definition}[Terms]
Given a set of labels $Lb$, a \emph{term}, $t$ is recursively defined as one of the following expressions \mbox{$t ::= l \ | \ \prod S \ | \ \sum S \ | \ t_1 \cdot t_2 \ | \ \sneg t_1$} where $l \in Lb$, $t_1, t_2$ are in their turn terms and $S$ is a (possibly empty and possibly infinite) set of terms. A term is \emph{elementary} if it has the form $l$, $\sneg l$ or $\sneg\sneg l$ with $l\in Lb$ being a label.\qed
\end{definition}

\noindent When $S=\{t_1,\dots,t_n\}$ is finite we simply write $\prod S$ as $t_1*\dots*t_n$ and $\sum S$ as $t_1+\dots+t_n$. Moreover, when $S=\emptyset$, we denote $\prod S$ by $1$ and $\sum S$ by $0$, as usual, and these will be the identities of the product `$*$' and the addition `$+$', respectively. We assume that `$\cdot$' has higher priority than `$*$' and, \review{R2.5}{in turn}, `$*$' has higher priority than `$+$'.

\begin{definition}[Values]
\label{def:values}
A \emph{(causal) value} is each equivalence class of terms under axioms for a completely distributive (complete) lattice with meet `$*$' and join `$+$' plus the axioms of Figures~\ref{fig:appl} and~\ref{fig:neg}.
The set of (causal) values is denoted by $\evalues$.\qed
\end{definition}

\begin{figure}[htbp]
\begin{center}
\newcommand{\titleSep}{0pt}
\newcommand{\contentSep}{-10pt}
\newcommand{\rowSep}{5pt}
$
\begin{array}{c}
\hbox{\em Associativity}\vspace{\titleSep}\\
\hline\vspace{\contentSep}\\
\begin{array}{r@{\ }c@{\ }r@{}c@{}l c r@{}c@{}l@{\ }c@{\ }l@{\ }}
t & \cdot & (u & \cdot & w) & = & (t & \cdot & u) & \cdot & w\\
\\
\end{array}
\end{array}
$
\ \ \ \
$
\begin{array}{c}
\hbox{\em Absorption}\vspace{\titleSep}\\
\hline\vspace{\contentSep}\\
\begin{array}{r@{\ }c@{\ }c@{\ }c@{\ }l c r@{\ }c@{\ }r@{\ }c@{\ }c@{\ }c@{\ }c@{\ }l@{\ }}
&& t &&& = & t & + & u & \cdot & t & \cdot & w \\
u & \cdot & t & \cdot & w & = & t & * & u & \cdot & t & \cdot & w
\end{array}
\end{array}
$
\ \ \ \
$
\begin{array}{c}
\hbox{\em Identity}\vspace{\titleSep}\\
\hline\vspace{\contentSep}\\
\begin{array}{rc r@{\ }c@{\ }l@{\ }}
t & = & 1 & \cdot & t\\
t & = & t & \cdot & 1
\end{array}
\end{array}
$
\ \ \ \
$
\begin{array}{c}
\hbox{\em Annihilator}\vspace{\titleSep}\\
\hline\vspace{\contentSep}\\
\begin{array}{rc r@{\ }c@{\ }l@{\ }}
0 & = & t & \cdot & 0\\
0 & = & 0 & \cdot & t\\
\end{array}
\end{array}
$
\\
\vspace{\rowSep}
$
\begin{array}{c}
\hbox{\em Idempotency}\vspace{\titleSep}\\
\hline\vspace{\contentSep}\\
\begin{array}{r@{\ }c@{\ }l@{\ }c@{\ }l }
x & \cdot & x  & = & x\\
\\
\\
\end{array}
\end{array}
$
\hspace{.05cm}
$
\begin{array}{c}
\hbox{\em Addition\ distributivity}\vspace{\titleSep}\\
\hline\vspace{\contentSep}\\
\begin{array}{r@{\ }c@{\ }r@{}c@{}l c r@{}c@{}l@{\ }c@{\ }r@{}c@{}l@{}}
t & \cdot & (u & + & w) & = & (t & \cdot & u) & + & (t & \cdot & w)\\
( t & + & u ) & \cdot & w & = & (t & \cdot & w) & + & (u & \cdot & w)\\ \\
\end{array}
\end{array}
$
\hspace{.05cm}
$
\begin{array}{c}
\hbox{\em Product\ distributivity}\vspace{\titleSep}\\
\hline\vspace{\contentSep}\\
\begin{array}{rcl}
c \cdot d \cdot e & = & (c \cdot d) * (d \cdot e) \ \ \ \hbox{with} \ d \neq 1 \\
c \cdot (d*e)     & = & (c \cdot d) * (c \cdot e) \\
(c*d) \cdot e     & = & (c \cdot e) * (d \cdot e)
\end{array}
\end{array}
$
\end{center}
\vspace{-5pt}
\caption{Properties of the `$\cdot$' operator ($c,d,e$ are terms without `$+$' and $x$ is an elementary term).
Distributivity is also satisfied over infinite sums and products.}
\label{fig:appl}
\end{figure}

\begin{figure}[htbp]
\begin{center}
\newcommand{\titleSep}{0pt}
\newcommand{\contentSep}{-10pt}
\newcommand{\rowSep}{5pt}
$
\begin{array}{c}
\hbox{\em Pseudo-complement}\vspace{\titleSep}\\
\hline\vspace{\contentSep}\\
\begin{array}{c@{\ }c@{\ }r@{}  }
t \ * \ \sneg t     & = & 0
\\
\sneg\sneg\sneg t   & = & \sneg t
\end{array}
\end{array}
\hspace{0.6cm}
\begin{array}{c}
\hbox{\em De Morgan}\vspace{\titleSep}\\
\hline\vspace{\contentSep}\\
\begin{array}{r@{}r@{}c@{}l@{\ }c@{\ }l@{}c@{}l@{} }
\sneg & (t & + & u) & = & (\sneg t  & * & \sneg u)\\
\sneg & (t & * & u) & = & (\sneg t  & + & \sneg u)
\end{array}
\end{array}
\hspace{0.6cm}
\begin{array}{c}
\hbox{\em Weak excl. middle}\vspace{\titleSep}\\
\hline\vspace{\contentSep}\\
\begin{array}{c@{\ }c@{\ }r@{}  }
\sneg t \ + \ \sneg\sneg t     & = & 1
\\
\\
\end{array}
\end{array}
\hspace{0.6cm}
\begin{array}{c}
\hbox{\em appl. negation}\vspace{\titleSep}\\
\hline\vspace{\contentSep}\\
\begin{array}{c@{\ }c@{\ }r@{}  }
\sneg (t \ \cdot \ u)   & = & \sneg (t \ * \ u) 
\\
\\
\end{array}
\end{array}
$
\end{center}
\vspace{-5pt}
\caption{Properties of the `$\sneg$' operator.}
\label{fig:neg}
\end{figure}

% \begin{figure}[htbp]
% \begin{center}
% \newcommand{\titleSep}{0pt}
% \newcommand{\contentSep}{-10pt}
% \newcommand{\rowSep}{5pt}
% $
% \begin{array}{c}
% \hbox{\em pseudo-complement}\vspace{\titleSep}\\
% \hline\vspace{\contentSep}\\
% \begin{array}{c@{}c@{}r@{}  }
% t \ * \ \sneg t     & = & 0
% \\
% \sneg\sneg\sneg t   & = & \sneg t
% \end{array}
% \end{array}
% \hspace{0.75cm}
% \begin{array}{c}
% \hbox{\em De Morgan}\vspace{\titleSep}\\
% \hline\vspace{\contentSep}\\
% \begin{array}{r@{}r@{}c@{}l@{}c@{}l@{}c@{}l@{} }
% \sneg & (t & + & u) & = & (\sneg t  & * & \sneg u)\\
% \sneg & (t & * & u) & = & (\sneg t  & + & \sneg u)
% \end{array}
% \end{array}
% \hspace{0.75cm}
% \begin{array}{c}
% \hbox{\em excluded middle}\vspace{\titleSep}\\
% \hline\vspace{\contentSep}\\
% \begin{array}{c@{}c@{}r@{}  }
% \sneg t \ + \ \sneg\sneg t     & = & 1
% \\
% \\
% \end{array}
% \end{array}
% \hspace{0.75cm}
% \begin{array}{c}
% \hbox{\em excluded middle}\vspace{\titleSep}\\
% \hline\vspace{\contentSep}\\
% \begin{array}{c@{}c@{}r@{}  }
% \sneg (t \ \cdot \ u)   & = & \sneg (t \ * \ u) 
% \\
% \\
% \end{array}
% \end{array}
% $
% \end{center}
% \vspace{-5pt}
% \caption{Properties of the `$\sneg$' operator.}
% \label{fig:neg}
% \end{figure}

\noindent
Note that $\tuple{\evalues,+,*,\sneg\ ,0,1 }$ is a
completely distributive Stone algebra (a pseudo-complemented, completely distributive, complete lattice which satisfies the weak excluded middle axiom) whose meet and join are, as usual, the product `$*$' and the addition~`$+$'. Informally speaking, this means that these two operators satisfy the properties of a Boolean algebra but without negation.

Note also that all three operations, `$*$', `$+$' and `$\cdot$' are associative. Product~`$*$' and addition `$+$' are also commutative, and they hold the usual absorption and distributive laws with respect to infinite sums and products of a completely distributive lattice. 

The axioms for `$\cdot$' in Figure~\ref{fig:appl} are directly extracted from the CG algebraic structure. For a more detailed explanation on their induced behaviour see~\cite{CabalarFF14}. The new contribution in this paper with respect to the CG algebra is the introduction of the `$\sneg$' operator whose meaning is captured by the axioms in Figure~\ref{fig:neg}. As we can see, this operator satisfies De Morgan laws and acts as a complement for the product $t * \sneg t = 0$. However, it diverges from a classical Boolean negation in some aspects. In the general case, the axioms $\sneg \sneg t = t$ (double negation) and $t + \sneg t = 1$ (excluded middle) are not valid. Instead\footnote{This behaviour coincides indeed with the properties for default negation obtained in Equilibrium Logic~\cite{Pearce96} or the equivalent General Theory of Stable Models~\cite{FerrarisLL07}.}, we can replace a triple negation $\sneg \sneg \sneg t$ by $\sneg t$, and we have a weak version of the excluded middle axiom $\sneg t + \sneg \sneg t = 1$. 
The negation of an application is defined as the negation of the product $\sneg (t \cdot u) \eqdef \sneg (t * u )$ which, \review{R2.5}{in turn}, is equivalent to $\sneg(u * t)$, since $*$ is commutative. In other words, under negation, the rule application ordering is disregarded. It is not difficult to see that we can apply the axioms of negation to reach an equivalent expression that avoids its application to other operators.
We say that a term is in \emph{negation normal form} (NNF) if no other operator is in the scope of negation `$\sneg$'. Moreover, an NNF term is in \emph{disjuntive normal form} (DNF) if: (1) no sum is in the scope of another operator; (2) only elementary terms are in the scope of application; and (3) every product is transitively closed, that is, of the form of $a \cdotl b * b \cdotl c * a \cdotl c$.
Without loss of generality, we assume from now that all functions defined over causal terms are applied over their DNF form, although, we will usually write them in NNF for short.

The lattice order relation is defined as usual in the following way:
% \vspace{-0.25cm}
\begin{IEEEeqnarray*}{c"C"c"C"c}
t \leq u & \text{ iff } & (t * u = t) & \text{ iff } & (t + u = u)
\end{IEEEeqnarray*}
Consequently $1$ and $0$ are respectively the top and bottom elements with respect to relation $\leq$.

\begin{definition}[Labelled logic program]\label{def:causal.P}
Given a signature $\signature$, a \emph{(labelled logic) program} $P$ is a set of rules of the form:
\begin{eqnarray}
r_i: \ \ \rH \ 
    \leftarrow \ \rB_1, \dotsc, \ \rB_m, \
                    \Not \rC_1, \dotsc, \ \Not \rC_n
    \label{eq:rule} 
\end{eqnarray}

\noindent where \mbox{$r_i\in Lb$} is a label or $r_i=1$, $\rH$ (the \emph{head} of the rule) \review{R3.4}{is an atom}, and  $\rB_i$'s and $\rC_i$'s (the \emph{body} of the rule) are either atoms or terms.\qed
\end{definition}

\noindent
% For any rule $\cRR$ of the form \eqref{eq:rule} we denote by $head(\cRR) \eqdef \rH$ its \emph{head}, by \mbox{$body^+(\cRR) \ \eqdef \ \set{ \rB_1, \dotsc \rB_m }$} its positive body and by \mbox{$body^-(\cRR)\eqdef\set{ \rC_1, \dotsc, \rC_n }$} its negative body. For any set $S$ we denote by $\Not S\eqdef\setm{\Not \rA}{ \rA \in S}$. Hence we denote by $body(\cRR) \eqdef body^+(\cRR) \cup \Not body^-(\cRR)$ the rule body.
When $n=0$ we say that the rule is positive, \review{R2.6}{furthermore, if in addition \mbox{$m=0$} }we say that the rule is a \emph{fact} and omit the symbol `$\leftarrow$.'
When $r_i \in Lb$ we say that the rule is labelled; otherwise $r_i=1$ and we omit both $r_i$ and `$:$'. By these conventions, for instance, an unlabelled fact $\rA$ is actually an abbreviation of $(1: \rA \leftarrow )$.
A program $P$ is \emph{positive} when all its rules are positive, i.e. it contains no default negation. It is \emph{uniquely labelled} when each rule has a different label or no label at all. In this paper, we will assume that programs are uniquely labelled. Furthermore, for the \review{R2.6}{sake of clarity}, we also assume that, for every atom $\rA\in At$, there is an homonymous label $\rA\in Lb$, and that each fact $\rA$ in the program actually stands for the labelled rule $(\rA: \; \rA \leftarrow)$.
For instance, following these conventions, a possible labelled version for the James Bond's program could be program~\newprogram\label{prg:bond}~below:
% \begin{eqnarray}
% r_1 : \ \ p & \leftarrow & d, \Not a  \label{r1.labrelled}\\
% r_2 : \ \ a & \leftarrow & \Not h\\
% d : \ \ d \label{d.labelled}\\
% h : \ \ h
% \end{eqnarray}
\\[-8pt]
\begin{minipage}{\textwidth}
\begin{minipage}[t]{\dimexpr\textwidth-2cm-(0.36\textwidth)\relax}
\begin{eqnarray*}
r_1 : \ \ p & \leftarrow & d, \Not a\\  %\label{r1.labrelled}\\
r_2 : \ \ a & \leftarrow & \Not h
\end{eqnarray*}
\end{minipage}
\hspace{2cm}
\begin{minipage}[t]{0.35\textwidth}
\begin{eqnarray*}
d\\ %\label{d.labelled}\\
h
\end{eqnarray*}
\end{minipage}
\end{minipage}
\\[8pt]
where facts $d$ and $h$ stand for rules $(d: \; d \leftarrow)$ and $(h: \; h\leftarrow)$, respectively.

An \emph{ECJ-interpretation} is a mapping \mbox{$\eI:At\longrightarrow\evalues$} assigning a value to each atom. For interpretations $\eI$ and $\eJ$ we say that $\eI\leq \eJ$ when
\mbox{$\eI(\rA) \leq \eJ(\rA)$} for each atom $\rA \in At$.
Hence, there is a \mbox{$\leq$-bottom} interpretation \botI\ (resp. a $\leq$-top interpretation~\topI) that stands for the interpretation mapping each atom $\rA$ to $0$ (resp. $1$).
The value assigned to a negative literal $\Not \rA$ by an interpretation~$\eI$, denoted as $\eI(\Not \rA)$, is defined as $\eI(\Not \rA) \eqdef \sneg\eI(\rA)$, as expected. Similarly, for a term $t$, $\eI(t)\eqdef [t]$ is the equivalence class of~$t$.

\begin{definition}[Model]
An interpretation $\eI$ satisfies a rule like \eqref{eq:rule}  iff
\begin{IEEEeqnarray}{c+x*}
\big( \ \eI(\rB_1) * \dotsc * \eI(\rB_m)
    * \eI(\Not \rC_1) * \dotsc * \eI(\Not \rC_n) \ \big) 
        \cdot r_i \ \leq \ \eI(\rH) &
     \label{eq:model}
\end{IEEEeqnarray}
and $\eI$ is a (causal) model of $\cP$, written $\eI\models \cP$, iff $\eI$ satisfies all rules in $\cP$.\qed
\end{definition}

As usual in LP, for positive programs, we may define a direct \review{R2.7}{consequence operator}~$\etp$ s.t.
\begin{align*}
\etp(\eI)(\rH)
    \ \ &\eqdef \ \
      \sum \big\{ \ \big( \ \eI(B_1) * \dotsc * \eI(B_n) \ \big) \cdot r_i \
% \\&\hspace{2.5cm}
        \mid \ (r_i: \ \rH \leftarrow B_1, \dotsc, B_n ) \in P \ \big\}
\end{align*}
for any interpretation $\eI$ and atom $\rH \in \at$.
\review{R1.5}{
We also define $\etpr{\alpha} \eqdef \etp(\etpr{\alpha-1})$ for any successor ordinal $\alpha$ and
\begin{gather*}
\etpr{\alpha} \ \ \eqdef \ \ \sum_{\beta < \alpha} \etpr{\beta}
\end{gather*}
for any limit ordinal alpha. 
As usual, $\omega$ denotes the smallest infinite limit ordinal.
Note that $0$ is considered a limit ordinal and, thus,
$\etpr{0} = \sum_{\beta < 0} \etpr{\beta} = \botI$.
}

\begin{theorem}\label{thm:tp.properties}
\label{prop:tp.properties}
Let $\cP$ be a (possibly infinite) positive logic program.
Then,
($i$) \ the least fixpoint of the $\etp$ operator, denoted by $\lfp(\etp)$, satisfies
$\lfp(\etp)=\etpr{\omega}$ and it is the least model of $\cP$, 
\ \ ($ii$) \ furthermore, if $P$ is positive and has $n$ rules, then  $\lfp(\etp)=\etpr{\omega}=\etpr{n}$.\qed
\end{theorem}

Theorem~\ref{thm:tp.properties} asserts that, as usual, positive programs have a $\leq$-least causal model. As we will see later, this least model coincides with the traditional least model (of the program without labels) when one just focuses on the set of true atoms, disregarding the justifications explaining why they are true.
For programs with negation we define the following reduct.

\begin{definition}[Reduct]
\label{def:e.model}
 Given a program $\cP$ and an interpretation $\eI$ we denote by $\ereduct{\cP}{\eI}$ the positive program containing a rule of the form
\begin{gather}
r_i : \ \ \rH \leftarrow \rB_1, \dotsc, \rB_m, \
    \eI(\Not \rC_1), \dotsc, \ \eI(\Not \rC_n)
    \label{eq:e.rule}
\end{gather}
\review{R2.8}{for each rule} of the form $\eqref{eq:rule}$ in $\cP$.\qed
\end{definition}

\noindent
Program $\cP^\eI$ is positive and, from Theorem~\ref{thm:tp.properties}, it has a \emph{least causal model}. 
By $\eWp(\eI)$ we denote the least model of program $\ereduct{\cP}{\eI}$. The operator $\eWp$ is anti-monotonic and, consequently, $\eWp^2$ is monotonic (Proposition~\ref{prop:gamma.antimonotonic} in the appendix) so that, by Knaster-Tarski's theorem, it has a least fixpoint $\elfp$ and a greatest fixpoint $\egfp\eqdef\eWp(\elfp)$. These two fixpoints respectively correspond to the justifications for true and for non-false atoms in the (standard) well-founded model (WFM), we denote as $W_P$.

For instance, in our running example,
$\elfpP{\programref{prg:bond}}(d) = \eWprP{\programref{prg:bond}}{\alpha}(d) = d$ for $1\leq\alpha$ points out that atom $d$ is true because of fact $d$.
Similarly,\ \
$\elfpP{\programref{prg:bond}}(h) = h$
and
$\elfpP{\programref{prg:bond}}(a) = \sneg h \cdotl r_2$
reveals that atom $h$ is true because of fact $h$, and that atom $a$ is not true because fact $h$ has inhibited rule $r_2$.

Furthermore,
\begin{gather*}
\elfpP{\programref{prg:bond}} (p)
	\ \ = \ \ \eWprP{\programref{prg:bond}}{\alpha}(p)
  \ \ = \ \ (\sneg(\sneg h \cdotl r_2) * d ) \cdotl r_1
  \ \ = \ \ (\sneg\sneg h  * d ) \cdotl r_1   +  (\sneg r_2  * d ) \cdotl r_1
\end{gather*}
for $2\leq\alpha$.
That is, Bond has been paralysed because fact $h$ has enabled drug $d$ to cause the paralysis by means of rule $r_1$.
This corresponds to the justification $(\sneg\sneg h * d ) \cdotl r_1$. Notice how the real cause $d$ is a positive label (not in the scope of negation) whereas the enabler $h$ is in the scope of a a double negation $\sneg\sneg h$. Justification~\mbox{$(\sneg r_2 * d ) \cdotl r_1$} means that $d\cdotl r_1$ would have been sufficient to cause $p$, had not been present $r_2$. 
%\comment{Note that, if $d \cdotl r_1$ would have been a justification, then $(\sneg\sneg h \!*\! d ) \cdotl r_1$ would be redundant because the former is stronger (in the sense that uses less rules), formally, the following equivalence holds: $(\sneg\sneg h \!*\! d ) \cdotl r_1 + d \cdot r_1 = d \cdot r_1$. Ad\'onde quieres llegar con este p\'arrafo? Me parece innecesariamente lioso}
This example is also useful for illustrating the importance of axiom \emph{appl. negation}.
By directly evaluating the body of rule $r_1$,
we have seen that
\mbox{$\eWprP{\!\programref{prg:bond}}{2}(p)  \!=\! (\sneg(\sneg h \cdot r_2) * d ) \cdot r_1$}.
Then, axiom \emph{appl. negation} allows us to break the dependence between $\sneg h$ and $r_2$ into enablers and inhibitors:
\mbox{$\sneg(\sneg h \cdot r_2) = \sneg(\sneg h * r_2) = \sneg\sneg h + \sneg r_2$}
and, applying distributivity,
we obtain one enabled justification, $(\sneg\sneg h * d ) \cdotl r_1$, and one disabled one, $ (\sneg r_2 * d ) \cdotl r_1$.

In our previous example, the least and greatest fixpoint coincided $\elfpP{\programref{prg:bond}} = \egfpP{\programref{prg:bond}} = \eWprP{\programref{prg:bond}}{2}$.
To \review{R2.9}{illustrate the case where} this does not hold consider, for instance, the program $\newprogram\label{prg:cycle}$ formed by the following negative cycle:
\begin{gather*}
\begin{IEEEeqnarraybox}[][t]{lC lCl}
r_{1} &: \ \ & a &\leftarrow& \Not b
\end{IEEEeqnarraybox}
\hspace{2cm}
\begin{IEEEeqnarraybox}[][t]{lC lCl}
r_{2} &: \ \ & b &\leftarrow& \Not a
\end{IEEEeqnarraybox}
\end{gather*}
In this case, the least fixpoint of $\eWp^2$ assigns $\elfpP{\programref{prg:cycle}}(a)= \sneg r_2 \cdotl r_1$ and $\elfpP{\programref{prg:cycle}}(b)= \sneg r_1 \cdotl r_2$, while, in its turn, the greatest fixpoint of $\eWp^2$ corresponds to $\egfpP{\programref{prg:cycle}}(a)= r_1$ and $\egfpP{\programref{prg:cycle}}(b)= r_2$. If we focus on atom $a$, we can observe that it is not concluded to be true, since the least fixpoint $\elfp$ has only provided one disabled justification $\sneg r_2 \cdotl r_1$ meaning that $r_2$ is acting as a disabler for $a$. But, on the other hand, $a$ cannot be false either since the greatest fixpoint provides an enabled justification $r_1$ for being non-false (remember that $\egfp$ provides justifications for non-false atoms). As a result, we get that $a$ is left undefined because $r_2$ prevents it to become true while $r_1$ can still be used to conclude that it is not false.

To capture these intuitions, we provide some definitions. A \emph{query literal (q-literal)} $\rL$ is either an atom $\rA$, its default negation `$\Not \rA$' or the expression `$\Undef \rA$' meaning that $\rA$ is undefined.

\begin{definition}[Causal well-founded model]~Given a program $\cP$, its \emph{causal well-founded model} $\ewfm$ is a mapping from q-literals to values~s.t.
\begin{IEEEeqnarray*}{c +x*}
\ewfm(\rA) \eqdef \elfp(\rA) 
\hspace{10pt}
\ewfm(\Not \rA)\eqdef \ \sneg\egfp(\rA)
\hspace{10pt}
\ewfm(\Undef \rA) \eqdef \sneg \ewfm(\rA) * \sneg \ewfm(\Not \rA)  &\qed
\end{IEEEeqnarray*}
\end{definition}

Let $l$ be a label occurrence in a term $t$ in the scope of $n\geq 0$ negations. We say that $l$ is an \emph{odd} or an \emph{even} occurrence if $n$ is odd or even, respectively. We further say that $l$ is \review{R2.10}{a \emph{strictly} even} occurrence if it is even and $n>0$. 

\begin{definition}[Justification]
\label{def:justification}
Given a program $P$ and a q-literal~$\rL$ we say that \review{R3.6}{a term $E$ with no sums} is a \emph{(sufficient causal) justification} for $\rL$ iff \ $E\leq~\ewfm(\rL)$.
Odd (resp. strictly even) labels\footnote{We just mention labels, and not their occurrences because terms are in NNF and $E$ contains no sums. Thus, having odd and even occurrences of a same label at a same time would mean that $E=0$.} in $E$ are called \emph{inhibitors} (resp. \emph{enablers}) of $E$. A justification is said to be \emph{inhibited} if it contains some inhibitor and it is said to be \emph{enabled} otherwise.\qed
\end{definition}

True atoms will have at least one enabled justification, whereas false atoms only contain disabled justifications. As an example of a query for a plain atom $A$, take the already seen explanation for $p$ in Bond's example program \programref{prg:bond}: $\ewfmP{\programref{prg:bond}}(p)  \!=\! \elfpP{\programref{prg:bond}}(p) \!=\! (\sneg\sneg h * d ) \cdotl r_1 + (\sneg r_2 * d ) \cdotl r_1$. We have here two justifications for atom $p$, let us call them \mbox{$E_1 \!=\!(\sneg\sneg h \!*\! d ) \cdotl r_1$} and $E_2 = (\sneg r_2 * d ) \cdotl r_1$. Justification $E_1$ is enabled because it contains no inhibitors (in fact, $E_1$ is the unique real support for $p$). Moreover, $h$ is an enabler in $E_1$ because it is strictly even (it is in the scope of double negation) whereas $d$ is a productive cause, since it is not in the scope of any negation. On the contrary, $E_2$ is disabled because it contains the inhibitor $r_2$ (it occurs in the scope of one negation). Intuitively, $r_2$ has prevented $d \cdotl r_1$ to become a justification of $p$. On the other hand, for atom $a$ we had $\ewfmP{\programref{prg:bond}}(a)  \!=\! \sneg h \cdot r_2$ that only contains an inhibited justification (being $h$ the inhibitor), and so, atom~$a$ is not true.
Now, if we query about the negative q-literal $\Not a$, we obtain $\ewfmP{\programref{prg:bond}}(\Not a)  \!=\! \sneg\egfpP{\programref{prg:bond}}(a)$ which in this case happens to be $\sneg\elfpP{\programref{prg:bond}}(a)  \!=\! \sneg(\sneg h \cdot r_2) = \sneg \sneg h + \sneg r_2$. That is, q-literal $\Not a$ holds, being enabled by $h$. Moreover, $\sneg r_2$ points out that removing $r_2$ would suffice to cause $\Not a$ too. It is easy to see that the explanations we can get for q-literals $\Not \rA$ or $\Undef \rA$ will have all their labels in the scope of negation (either as inhibitors or as enablers).

To illustrate a query for $\Undef A$, let us return to program~$\programref{prg:cycle}$ whose standard well-founded model left both $a$ and $b$ undefined. Given the values we obtained in the least and greatest fixpoints, the causal WFM will assign $\ewfmP{\programref{prg:cycle}}(a) = \sneg r_2 \cdotl r_1$ and
$\ewfmP{\programref{prg:cycle}}(b)= \sneg r_1 \cdotl r_2$, that is, $r_2$ prevents $r_1$ to cause $a$ and $r_1$ prevents $r_2$ to cause~$b$. Furthermore, the values assigned to their respective negations, $\ewfmP{\programref{prg:cycle}}(\Not a) = \sneg r_1$ and $\ewfmP{\programref{prg:cycle}}(\Not b)= \sneg r_2$, point out that atoms~$a$ and~$b$ are not false because rules $r_1$ and $r_2$ have respectively prevented them to be so. Finally, we obtain that $\Undef a$ is true because
\begin{gather*}
\ewfm(\Undef a)
	= \sneg \ewfmP{\programref{prg:cycle}}(a) * \sneg \ewfmP{\programref{prg:cycle}}(\Not a)
	= (\sneg\sneg r_2+\sneg r_1) * \sneg\sneg r_1
	% = \sneg r_1 * \sneg\sneg r_1 + \sneg\sneg r_1 * \sneg\sneg r_2
	% = 0 + \sneg\sneg r_1 * \sneg\sneg r_2
	= \sneg\sneg r_2 * \sneg\sneg r_1
\end{gather*}
that is, rules $r_1$ and $r_2$ together have made $a$ undefined. Similarly, $b$ is also undefined because of rules $r_1$ and $r_2$, $\ewfm(\Undef b) = \sneg\sneg r_1 * \sneg\sneg r_2$.

% For instance, \CHANGE{in our Bond example},
% % in our previous example,
% there are two justifications, \mbox{$E_1 \!=\!(\sneg\sneg h \!*\! d ) \cdotl r_1$} and $E_2 = (\sneg r_2 * d ) \cdotl r_1$, for atom $p$.

The next theorem shows that the literals satisfied by the standard WFM are precisely those ones containing at least one enabled justification in the causal WFM.

\begin{theorem}\label{thm:wellf.justification<-$>$weff.standard}
Let $P$ be a labelled logic program  over a signature~$\signature$ where $\lb$ is a finite set of labels and let $W_P$ its (standard) well-founded model.
A q-literal \mbox{$\rL$} holds with respect to $W_P$ if and only if there is some \mbox{enabled} justification $E$ of $\rL$, that is, \mbox{$E\leq\ewfm(\rL)$} and $E$ does not contain odd negative labels.\qed
\end{theorem}

\noindent Back to our example program $\programref{prg:bond}$, as we had seen, atom $p$ had a unique enabled justification $E_1 = (\sneg\sneg h * d ) \cdotl r_1$. The same happens for atoms $d$ and $h$ whose respective justifications are just their own atom labels. Therefore, these three atoms hold in the standard WFM, $W_{\programref{prg:bond}}$. On the contrary, as we discussed before, the only justification for $a$, $\ewfmP{\programref{prg:bond}}(a) = \sneg h \cdotl r_2$, is inhibited by $h$, and thus, $a$ does not hold in $W_{\programref{prg:bond}}$. 
The interest of an inhibited justification for a literal is to point out ``potential'' causes that have been prevented by some abnormal situation. In our case, the presence of $\sneg h$ in $\ewfmP{\programref{prg:bond}}(a) = \sneg h \cdotl r_2$ points out that an exception $h$ has prevented $r_2$ to cause $a$.
When the exception is removed, the inhibited justification (after removing the inhibitors) becomes an enabled justification.
% $r_2$ for $a$.

In our running example, if we consider a program~$\newprogram\label{prg:bond.remove-h}$ obtained by removing the fact $h$ from \programref{prg:bond}, then $\ewfmP{\programref{prg:bond.remove-h}}(a) = r_2$ points out that $a$ has been caused by rule $r_2$ in this new scenario. This intuition about inhibited justifications is formalized as follows.

\begin{definition}\label{def:varrho}
Given a term $t$ in DNF, by $\varrho_x:\values \longrightarrow\values$, we denote the function that removes  the elementary term $x$ from $t$ as follows:
\begin{align*}
\varrho_x(t) \ &\eqdef \ \begin{cases}
    \varrho_x(u) \otimes \varrho_x(w)
        &\text{if } t=u \otimes v \text{ with } \otimes\in\set{+,*,\cdot}
    \\
    1   &\text{if } \sneg\sneg t \text{ is equivalent to } \sneg\sneg x
    \\
    0   &\text{if } t \text{ is equivalent to } \sneg x 
\end{cases}
\end{align*}
Note that we have assumed that $t$ is in DNF. Otherwise, $\varrho_x(t)\eqdef\varrho_x(u)$ where $u$ is an equivalent term in DNF.\qed
\end{definition}

\begin{theorem}\label{thm:wellf.inhibited.justification-$>$sm.cause}
Let $P$ be a program over a signature~$\signature$ where $\lb$ is a finite set of labels.
Let $\cQ$ be the result of removing from $\cP$ all rules labelled by some $r_i \in\lb$.
Then, the result of removing $r_i$ from the justifications of some atom $\rA$ with respect to program $P$ are justifications of $\rA$ with respect to $Q$, that is,
$\varrho_{\sneg r_i}(\ewfmP{P}(\rA)) \leq \ewfmP{Q}(\rA)$.
\end{theorem}

%%%%%%%%%%%%%%%%%%%%%%%%%%%%%%%%%%%%%%%%%%%%%%%%%%%%%%%%%%%%%%%%%%%%%%%%%%%%%%%%%%%%%%%%%%%
%%%%%%%%%%%%%%%%%%%%%%%%%%%%%%%%%%%%%%%%%%%%%%%%%%%%%%%%%%%%%%%%%%%%%%%%%%%%%%%%%%%%%%%%%%%
%%%%%%%%%%%%%%%%%%%%%%%%%%%%%%%%%%%%%%%%%%%%%%%%%%%%%%%%%%%%%%%%%%%%%%%%%%%%%%%%%%%%%%%%%%%

\section{Relation to Causal Graph Justifications}
\label{sec:CG}

We discuss now the relation between ECJ and CG approaches.
Intuitively, ECJ extends CG causal terms by the introduction of the new negation operator `$\sneg$'.
 % Formally, ECJ extends CG causal terms by the introduction of the new negation operator `$\sneg$'.
Semantically, however, there are more differences than a simple syntactic extension. A first minor difference is that ECJ is defined in terms of a WFM, whereas CG defines (possibly) several causal stable models. In the case of stratified programs, this difference is irrelevant, since the WFM is complete and coincides with the unique stable model. A second, more important difference is that CG exclusively considers productive causes in the justifications, disregarding additional information like the inhibitors or enablers from ECJ. As a result, a false atom in CG has \emph{no justification} -- its causal value is $0$ because there was no way to derive the atom. For instance, in program $\programref{prg:bond}$, the only CG stable model $I$ just makes $I(a)=0$ and we lose the inhibited justification $\sneg h \cdot r_2$ (default $r_2$ could not be applied). True atoms like $p$ also lose any information about enablers: $I(p) = d \cdotl r_1$ and nothing is said about $\sneg\sneg h$. Another consequence of the CG orientation is that negative literals $\Not A$ are never assigned a cause (different from $0$ or~$1$), since they cannot be ``derived'' or produced by rules. In the example, we simply get $I(\Not a)=1$ and $I(\Not p)=0$.

To further illustrate the \review{R2.11}{similarities} and differences between ECJ and CG, consider the following program~\newprogram\label{prg:yale.broken} capturing a variation of the Yale Shooting Scenario.
\begin{gather*}
\begin{IEEEeqnarraybox}[][t]{lClCl}
d_{t+1} &:& dead_{t+1} &\leftarrow& shoot_t, \ loaded_t,\ \Not ab_t
\\
l_{t+1} &:& loaded_{t+1} &\leftarrow& load_t
\\
a_{t+1} &:& ab_{t+1} &\leftarrow& water_t
\end{IEEEeqnarraybox}
\hspace{1.25cm}
\begin{IEEEeqnarraybox}[][t]{l}
\overline{loaded}_0
\\
\overline{dead}_0
\\
\overline{ab}_0
\end{IEEEeqnarraybox}
\hspace{1.25cm}
\begin{IEEEeqnarraybox}[][t]{lCl}
load_1
\\
water_3
\\
shoot_8
\end{IEEEeqnarraybox}
\end{gather*}
plus the following rules corresponding inertia axioms 
\[
F_{t+1} \leftarrow F_{t},\ \Not \overline{F}_{t+1}
\hspace{30pt}
\overline{F}_{t+1} \leftarrow \overline{F}_{t},\ \Not F_{t+1}
\]
%\begin{IEEEeqnarray*}{lClCl+}
%&& F_{t+1} &\leftarrow& F_{t},\ \Not \overline{F}_{t+1}
%\\
%&& \overline{F}_{t+1} &\leftarrow& \overline{F}_{t},\ \Not F_{t+1}
%\vspace{-4pt}
%\end{IEEEeqnarray*}
for $F \in \set{loaded,\,ab,\,dead}$. Atoms of the form $\overline{A}$ represent the strong negation of $A$ and we disregard models satisfying both $A$ and $\overline{A}$.
Atom $dead_9$ does not hold in the standard WFM of~\programref{prg:yale.broken}, and so there is no CG-justification for it.
% However note that having $wet$ gunpowder is an exception and hence an unexpected situation. 
Note here the importance of default reasoning. On the one hand, the default flow of events is that the turkey, Fred, continues to be alive when nothing threats him. Hence, we do not need a cause to explain why Fred is alive. On the other hand, shooting a loaded gun would normally kill Fred, being this a cause of its death. But, in this example, another exceptional situation -- $water$ spilled out -- has \emph{inhibited} this existing threat and allowed the world to flow as if nothing had happened (that is, following its default behaviour). 
% Norm, of default behaviour, has been recognised as a cornerstone in the understanding of causality~\cite{hall2007structural,maudlin2004,halpern2011actual}.
% \cite{hall2007structural,maudlin2004,Halpern08,hitchcock2009cause,halpern2011actual}.

In the CG-approach, $dead_9$ is simply false by default and no justification is provided. However, a gun shooter could be ``disappointed'' since another conflicting default (shooting a loaded gun \emph{normally} kills) has not worked. Thus, an expected answer for the shooter's question ``why $\Not dead_9$?'' is  that $water_3$ broke the default, disabling $d_9$. In fact, ECJ yields the following inhibited justification for $dead_9$:
\begin{gather}
\ewfmP{\programref{prg:yale.broken}}(dead_9)=( \sneg water_3 * shoot_8 * load_1 \cdotl l_2 ) \cdot d_9
    \label{eq:e.just.yale.broken}
\end{gather}
\noindent meaning that $dead_9$ could not be derived because inhibitor $water_3$ prevented the application of rule~$d_9$ to cause the death of Fred. Note that inertia rules are not labelled, which, as mentioned before, is syntactic sugar for rules with label $1$. Since $1$ is the identity of product and application, this has the effect of not being traced in the justifications.
Note also that,
% Moreover,
according to Theorem~\ref{thm:wellf.inhibited.justification-$>$sm.cause}, if we remove fact $water_3$ (the inhibitor) from \programref{prg:yale.broken} leading to a new program \newprogram\label{prg:yale.nobroken}, then we get:
\begin{gather}
\ewfmP{\programref{prg:yale.nobroken}}(dead_9)=(shoot_8 * load_1 \cdotl l_2) \cdot d_9
    \label{eq:e.just.yale.nobroken}
\end{gather}
\noindent which is nothing else but the result of removing $\sneg water_3$ from~\eqref{eq:e.just.yale.broken}. In fact, the only CG stable model of \programref{prg:yale.nobroken} makes this same assignment~\eqref{eq:e.just.yale.nobroken} which also corresponds to the causal graph depicted in Figure~\ref{fig:graphs2}. In the general case, CG-justifications intuitively correspond to enabled justifications after forgetting all the enablers.
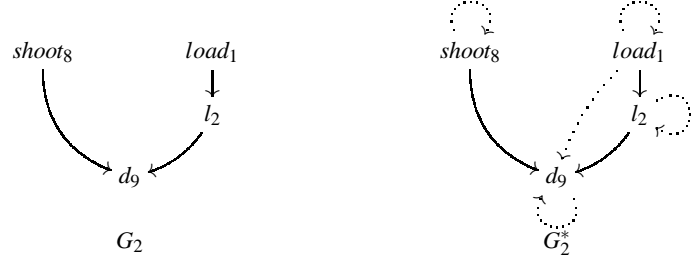
\begin{figure}[htbp]\centering
$$
\xymatrix @-5mm {
{shoot_8} \ar@/_12pt/[ddr]      & & {load_1} \ar[d]         &% & {\sneg\sneg repair_0 } \ar@/^15pt/[ddl]        %& { \sneg\sneg r_3} \ar@/^15pt/[ddll] 
&\hspace{1cm}&
{shoot_8} \ar@/_12pt/[ddr]  \ar@{.>}@(ul,ur)    & & {load_1} \ar[d] \ar@{.>}@/_5pt/[ddl]  \ar@{.>}@(ul,ur)      &
\\
                                & & {l_2} \ar@/^5pt/[dl]            &
&&
                                & & {l_2} \ar@/^5pt/[dl] \ar@{.>}@(ur,dr)           &
\\
                                & {d_9}
&&&&
                                & {d_9} \ar@{.>}@(dr,dl)
\\
                                & {G_2}
&&&&
                                & {G_2^*}
}
$$
\caption{$G_2$ is the cause of $dead_9$ in program~\programref{prg:yale.broken} while $G_2^*$ is its associated causal graphs, that is, its reflexive and transitive closure.}
\label{fig:graphs2}
\end{figure}
Formally, however, there is one more difference in the definition of causal values: CG causal values are defined as ideals for the poset of a type of graphs formed by rule labels.

\begin{definition}[Causal graph]
Given some set $\lb$ of (rule) labels,
a \emph{causal graph} (\emph{c-graph}) $G \subseteq \lb \times \lb$ is a reflexively and transitively closed set of edges.
By $\graphs$, we denote the set of causal graphs.
Given two c-graphs $G$ and $G'$, we write $G \leq G'$ when $G \supseteq G'$.\qed
\end{definition}

Intuitively, causal graphs, like~$G_2$ in Figure~\ref{fig:graphs2}, are directed graphs representing the causal structure that has produced some event.
Furthermore, $G \leq G'$ means that $G$ contains enough information to yield the \review{R2.12}{same effect as} $G'$, but perhaps more than needed (this explains $G \supseteq G'$). For this reason, we sometimes read $G \leq G'$ as ``$G'$ is \emph{stronger} than $G$.'' Causes will be $\leq$-maximal (or $\subseteq$-minimal) causal graphs.
Formally, including reflexive and transitive edges allows to capture this intuitive relation simply by the subgraph relation. Note that, since causal graphs are reflexively closed, every vertex has at least one edge (the reflexive one) and, thus, we can omit the set of vertices. Besides, for the sake of clarity, we only depict the minimum set of edges necessary for defining a causal graph (transitive and reflexive reduction).
\review{R3.7}{For instance, graph~$G_2$ in Figure~\ref{fig:graphs2} is the transitive and reflexive reduction of the causal graph $G_2^*$.}
% actually represents the causal graph obtained after applying the transitive and reflexive closure.

\begin{definition}[CG Values in~\citeNP{CabalarFF14}]
Given a set of labels $\lb$, a \emph{CG causal value} is any ideal (or lower-set) for the poset $\tuple{\graphs,\leq}$.
By~$\cgvalues$, we denote the set of CG causal values.
Product~`$*$', sums~`$+$' and the $\leq$-order relation are defined as the set intersection, union and the subset relation, respectively.
Application is given by
$U\cdotl U' \, \eqdef \, \setm{ G'' \leq G \cdot G' }{ G \in U \text{ and } G' \in U'}$.\qed
\end{definition}

It has been shown in~\cite{fandinno2015thesis} that CG values can be alternatively characterised as a free algebra generated by rule labels under the axioms of a complete distributive lattice plus the axioms of Figure~\ref{fig:appl}.

\begin{definition}[CG Values in~\citeNP{fandinno2015thesis}]
Given a set of labels $\lb$, a CG term is a term without negation~`$\sneg$'.
\emph{CG causal values} are the equivalence classes of CG terms for a completely distributive (complete) lattice with meet `$*$' and join `$+$' plus the axioms of Figure~\ref{fig:appl}.
By $\cgvaluesterms$, we denote the set of CG causal values.\qed
\end{definition}

\begin{theorem}[Causal values isomorphism from~\citeNP{fandinno2015thesis}]\label{thm:algebra.values}
The function \ \ $term: \cgvalues \longrightarrow \cgvaluesterms$ \ \ given by
\begin{gather*}
term(U) \mapsto \sum_{G \in U} \ \prod_{(v_1,v_2) \in G}  v_1 \cdotl v_2 
  \label{eq:algebra.values.isompphism}
\end{gather*}
is an isomorphism between algebras $\tuple{\cgvalues,+,*,\cdot,\graphs,\emptyset}$ and $\tuple{\cgvaluesterms,+,*,\cdot,1,0}$.\qed
\end{theorem}

Theorem~\ref{thm:algebra.values} states that CG causal values can be equivalently described either as ideals of causal graphs or as elements of an algebra of terms. Furthermore, by abuse of notation, by $G$ we also denote the ideal whose maximum element is $G$, corresponding to $term(G)$ as well. For instance, for the causal graph $G_2$ in Figure~\ref{fig:graphs2}, it follows $G_2 = term(G_2) = term(\downarrow\!\!G_2)$ with $\downarrow\!\!G_2$ the ideal whose maximum element is $G_2$. Moreover, from the equivalences in Figure~\ref{fig:appl}, it also follows that
\begin{align*}
G_2 &\ \ = \ \  shoot_8 \cdotl d_9 \ * \
                load_1 \cdotl l_2 \ * \
                l_2 \cdotl d_9 \ * \ \alpha
\\
    &\ \ = \ \  shoot_8 \cdotl d_9 \ * \
                load_1 \cdotl l_2 \ * \
                l_2 \cdotl d_9
\\
&\ \ = \ \  shoot_8 \cdotl d_9 \ * \
                load_1 \cdotl l_2 \cdotl d_9 
\\
&\ \ = \ \  (shoot_8 \ * \ load_1 \cdotl l_2 ) \cdotl d_9
\end{align*}
where \review{R3.8}{$\alpha = load_1 \cdotl d_9 * shoot_8 \cdotl shoot_8 * d_9 \cdotl d_9 * load_1 \cdotl load_1 * l_2 \cdotl l_2 * d_9 \cdotl d_9$}
is a term that, as we can see, can be ruled out and corresponds to the transitive and reflexive doted edges in $G_2^*$.
That is, justification \eqref{eq:e.just.yale.nobroken} associated to atom~$dead_9$ by the causal well-founded model of program~\programref{prg:yale.nobroken} actually corresponds to causal graph~$G_2$.

Theorem~\ref{thm:algebra.values} also formalises the intuition that opens this section: ECJ extends CG causal terms by the introduction of the new negation operator `$\sneg$'.
We formalise next the correspondence between CG and ECJ  justifications.

\begin{definition}[CG mapping]
\label{def:causal.values}
We define a mapping \mbox{$\lambdac:\evalues\longrightarrow\cgvaluesterms$} from ECJ values into CG values in the following recursive way:
\begin{align*}
\lambdac(t) &\eqdef \begin{cases}
    \lambdac(u) \otimes  \lambdac(w)
        &\text{if } t=u \otimes  v \text{ with } \otimes \in\set{+,*,\cdot}
    \\
    1   &\text{if } t=\sneg\sneg l \text{ with } l \in Lb
    \\
    0   &\text{if } t=\hspace{8pt} \sneg l \text{ with } l \in Lb
    \\
    l   &\text{if } t=\hspace{15pt} l \text{ with } l \in Lb
\end{cases}
\end{align*}
Note that we have assumed that $t$ is in DNF. Otherwise, $\lambdac(t)\eqdef\lambdac(u)$ where $u$ is an equivalent term in DNF.\qed
\end{definition}

\noindent Function $\lambdac$ maps every negated label $\sneg l$ to $0$ (which is the annihilator of both product~`$*$' and application~`$\cdot$' and the identity of addition~`$+$'). Hence $\lambdac$ removes all the inhibited justifications. Furthermore $\lambdac$ maps every doubly negated label $\sneg\sneg l$ to $1$ (which is the identity of both product~`$*$' and application~`$\cdot$'). Therefore $\lambdac$ removes all the enablers (i.e. doubly negated labels $\sneg\sneg l$) for the remaining (i.e. enabled) justifications.

A CG interpretation is a mapping \mbox{$\cI:\at\longrightarrow\cgvaluesterms$}.
The value assigned to a negative literal $\Not \rA$ by a CG interpretation~$\cI$, denoted as $\cI(\Not \rA)$, is  defined as: $\cI(\Not \rA) \eqdef 1$ if $\cI(\rA) = 0$; $\cI(\Not \rA) \eqdef 0$ otherwise.
A CG interpretation $\cI$ is a CG model of rule like~\eqref{eq:rule} iff
\begin{IEEEeqnarray}{c+x*}
\big( \ \cI(\rB_1) * \dotsc * \cI(\rB_m)
    * \cI(\Not \rC_1) * \dotsc * \cI(\Not \rC_n) \ \big) 
        \cdot r_i \ \leq \ \cI(\rH) &
     \label{eq:CG.model}
\end{IEEEeqnarray}
Notice that the value assigned to a negative literal by CG and ECJ interpretations is different.
According to \cite{CabalarFF14}, a CG interpretation $\cI$ is a \emph{CG stable model} of a program $\cP$ iff~$\cI$~is the least model of the program $\cP^\cI$. In the following, we provide an ECJ based characterisation of the CG stable models that will allow us to relate both approaches.
% For the sake of clarity, by $\Wp(\cI)$ we will denote the least model of the positive program $\cP^\cI$ according to the CG semantics,
% and
By $\lambdac(\eI)$ we will denote a \review{R1.6}{CG} interpretation $\cI$ s.t.
$\cI(\rA) = \lambdac(\eI(\rA))$ for every atom $\rA$.

\begin{definition}
[CG stable models]
\label{def:wellf.causal.model}
Given a program $\cP$,
a CG interpretation~$\cI$ is a \emph{CG stable model} of $\cP$ iff
there exists a fixpoint $\eI$ of the operator $\eWp^2$,
\review{R3.9}{i.e. $\eWp(\eWp(I))=I$},
such that
\mbox{$\cI=\lambdac(\eI)=\lambdac(\eWp(\eI))$}.\qed
\end{definition}

\begin{theorem}\label{thm:causal-prov.smodels.are.the.causal-smodels}
Let $\cP$ be a program over a signature~$\signature$ where $\lb$ is a finite set of labels.
Then, the CG~stable models (Definition~\ref{def:wellf.causal.model}) are exactly the causal values and causal stable models defined in \cite{CabalarFF14}.\qed
\end{theorem}

\noindent
Theorem~\ref{thm:causal-prov.smodels.are.the.causal-smodels} shows that Definition~\ref{def:wellf.causal.model} is an alternative definition of CG causal stable models. Furthermore, it settles that every causal model corresponds to some fixpoint of the operator $\eWp^2$. Therefore, for every enabled justification there is a corresponding \mbox{CG-justification} common to all stable models. In order to formalise this idea we just take the definition of causal explanation from~\cite{CabalarFF14Jelia}.

\begin{definition}[CG-justification]
\label{def:query}
Given an interpretation $\eI$ we say that a c-graph $G$ is a \emph{(sufficient) CG-justification} for an atom $\rA$ iff \mbox{$term(G)\leq \cI(\rA)$}.\qed
\end{definition}

\noindent
Since $term(\cdot)$ is a one-to-one correspondence, we can define its inverse
$graph(v)\eqdef term^{-1}(v)$ for all $v\in\cgvaluesterms$.

\begin{samepage}
\begin{theorem}\label{thm:wellf.justification-$>$sm.cause}
Let $\cP$ be a program over a signature~$\signature$ where $\lb$ is a finite set of labels. For any enabled justification $E$ of some atom $\rA$ w.r.t. $\wfm$, i.e. \mbox{$E\leq\ewfm(\rA)$}, there is a CG-justification \mbox{$G\eqdef graph(\lambdac(E))$} of~$\rA$ with respect to any stable model $\cI$ of~$\cP$.\qed
\end{theorem}
\end{samepage}

\noindent
% Note that, as a consequence of Theorem~\ref{thm:wellf.inhibited.justification-$>$sm.cause}, inhibited justifications also correspond to CG-justifications when the abnormalities have been removed. 
As happens between the (standard) well-founded and stable model semantics, the converse of Theorem~\ref{thm:wellf.justification-$>$sm.cause} does not hold in general. That is, we may get a justification that is common to all CG-stable models but does not occur in the ECJ well-founded model. For instance, let \newprogram\label{prog:asp.not.implies.wellf} be the program consisting on the following rules:
\[
r_1: \ \ a \leftarrow \Not b
\hspace{20pt}
r_2: \ \ b \leftarrow \Not a, \ \Not c
\hspace{20pt}
c
\hspace{30pt}
r_3: \ \ c \leftarrow a
\hspace{20pt}
r_4: \ \ d \leftarrow b, \ \Not d
\]
%\begin{gather*}
%\begin{IEEEeqnarraybox}[][t]{lCl C l}
%r_1     &:& a &\leftarrow& \Not b
%\\
%r_2     &:& b &\leftarrow& \Not a, \ \Not c
%\end{IEEEeqnarraybox}
%\hspace{0.55cm}
%\begin{IEEEeqnarraybox}[][t]{lCl C l}
%  & & c
%\\
%r_3     &:& c &\leftarrow& a
%\end{IEEEeqnarraybox}
%\hspace{0.55cm}
%\begin{IEEEeqnarraybox}[][t]{lCl C l}
%r_4     &:& d &\leftarrow& b, \ \Not d
%\end{IEEEeqnarraybox}
%\end{gather*}
The (standard) WFM of program~\programref{prog:asp.not.implies.wellf} is two-valued and corresponds to the unique (standard) stable model~$\set{a,c}$. Furthermore, there are two causal explanations of $c$ with respect to this unique stable model: the fact $c$ and the pair of rules $r_1 \cdotl r_3$. 
Note that when $c$ is removed $\set{a,c}$ is still the unique stable model, but all atoms are undefined in the WFM.
Hence, $r_1\cdotl r_3$ is a justification with respect to the unique stable model of the program, but not with respect to its WFM.

%%%%%%%%%%%%%%%%%%%%%%%%%%%%%%%%%%%%%%%%%%%%%%%%%%%%%%%%%%%%%%%%%%%%%%%%%%%%%%%%%%%%%%%%%%%
%%%%%%%%%%%%%%%%%%%%%%%%%%%%%%%%%%%%%%%%%%%%%%%%%%%%%%%%%%%%%%%%%%%%%%%%%%%%%%%%%%%%%%%%%%%
%%%%%%%%%%%%%%%%%%%%%%%%%%%%%%%%%%%%%%%%%%%%%%%%%%%%%%%%%%%%%%%%%%%%%%%%%%%%%%%%%%%%%%%%%%%

\section{Relation to Why-not Provenance}
 \label{sec:WnP}

An evident similarity between ECJ and WnP approaches is the use of an alternating fixpoint operator~\cite{van1989alternating} which has been actually borrowed from WnP.
However, there are some slight differences. A first one is that we have incorporated from CG the non-commutative operator `$\cdot$' which allows capturing not only which rules justify a given atom, but also the dependencies among these rules. The second is the use of a \emph{non-classical} negation `$\sneg$' that is crucial to distinguish between productive causes and enablers. This distinction cannot be represented with the classical negation `$\neg$' in WnP since double negation can always be removed. Apart from the interpretation of negation in both formalisms, there are other differences too. As an example, let us compare the justifications we obtain for $dead_9$ in program~\programref{prg:yale.nobroken}. While for ECJ we obtained \eqref{eq:e.just.yale.nobroken} (or graph $G_2$ in Figure~\ref{fig:graphs2}), the corresponding WnP justification has the form:
\begin{gather}
\begin{aligned}
l_2 \wedge d_9  \wedge load_1 &\wedge shoot_8\\
	&\wedge not(ab_1) \wedge not(ab_2) \wedge \dotsc   \wedge not(ab_7)
	\wedge not(water_0) \wedge  \dotsc   \wedge not(water_6)
\end{aligned}
    \label{eq:prov.just.yale.repair}
\end{gather}
A first observation is that the subexpression $l_2 \wedge d_9  \wedge load_1 \wedge shoot_8$ constitutes, informally speaking, a ``flattening'' of \eqref{eq:e.just.yale.nobroken} (or graph $G_2$) where the ordering among rules has been lost. We get, however, new labels of the form $not(A)$ meaning that atom~$A$ is required not to be a program fact, something that is not present in CG-justifications. For instance, \eqref{eq:prov.just.yale.repair} points out that $water$ can not be spilt on the gun along situations $0,\dotsc,7$. Although this information can be useful for debugging (the original purpose of WnP) its inclusion in a causal explanation is obviously inconvenient from a Knowledge Representation perspective, since it explicitly \emph{enumerates all the defaults} that were applied (no water was spilt at any situation) something that may easily blow up the (causally) irrelevant information in a justification.

An analogous effect happens with the enumeration of exceptions to defaults, like inertia. Take program \newprogram\label{prg:yale.noaction} obtained from \programref{prg:yale.broken} by removing all the performed actions, i.e., facts $load_1$, $water_3$, and $shoot_7$. As expected, Fred will be alive, $\overline{dead}_t$, at any situation $t$ by inertia.
ECJ will assign no cause for $dead_t$, not even any inhibited one, i.e. $\ewfm(\overline{dead}_t)=1$ and $\ewfm(dead_t)=0$ for any~$t$.
\review{R1.7}{
The absence of labels in $\ewfm(\overline{dead}_t)=1$ is, of course, due to the fact that inertia axioms are not labelled, as they naturally represent a default and not a causal law. Still, even if inertia were labelled, say, with $in_k$ per each situation $k$, we would obtain a \emph{unique cause} for $\ewfm(\overline{dead}_t)=in_1 \cdot \ \dots \ \cdot in_t$ for any $t>0$ while maintaining no cause for $\ewfm(dead_t)=0$.
However, the number of minimal WnP justifications of $dead_t$ grows quadratically, as it collects \emph{all the plans} for killing Fred in $t$ steps loading and shooting once. 
}For instance, among others, all the following:
\begin{IEEEeqnarray*}{cCcCcCcCcCcCcCl}
d_9 \wedge \neg not(load_0) &\wedge& r_2 
    &\wedge& \neg not(shoot_1)
    &\wedge& not(water_0)
    &\wedge& not(ab_1)
\\
d_9 \wedge \neg not(load_0) &\wedge& r_2 
    &\wedge& \neg not(shoot_2)
    &\wedge& not(water_0)
    &\wedge& not(water_1)
    &\wedge& not(ab_1)
    &\wedge& not(ab_2)
\\
d_9 \wedge  \neg not(load_1) &\wedge& r_2 
    &\wedge& \neg not(shoot_3)
    &\wedge& not(water_0)
    &\wedge&
\\
&&
	    &\wedge& not(water_1)
	&\wedge& not(water_2)
    &\wedge& not(ab_1)
    &\wedge& not(ab_2)
    &\wedge& not(ab_3)
\\ \dotsc
\end{IEEEeqnarray*}
are WnP-justifications for $dead_9$. The intuitive meaning of expressions of the form $\neg not(A)$ is that $dead_9$ can be justified by adding $A$ as a fact to the program. For instance,
the first conjunction means that it is possible to justify $dead_9$ by adding the facts $load_0$ and $shoot_1$ and not adding the fact $water_0$.
We will call these justifications, which contain a subterm of the form $\neg not(A)$, \emph{hypothetical} in the sense that they involve some hypothetical program modification.

\begin{definition}[Provenance values]
% Given a signature \signature, we denote by \mbox{$\wLb\eqdef Lb \cup\setm{ not(\rA) }{ \rA \in At }$}.
Given a set of labels $Lb$, a \emph{provenance term} $t$ is recursively defined as one of the following expressions \mbox{$t ::= l \ | \ \prod S \ | \ \sum S \ | \ \neg t_1$} where $l \in Lb$, $t_1$ is in its turn a provenance term and $S$ is a (possibly empty and possible infinite) set of provenance terms. \emph{Provenance values} are the equivalence classes of provenance terms under the equivalences of the Boolean algebra. We denote by $\boolAlgebra$ the set of provenance values over~$\wLb$.\qed
\end{definition}

\noindent
Informally speaking, with respect to ECJ, we have removed the application `$\cdot$' operator, whereas product `$*$' and addition~`$+$' hold the same equivalences as in Definition~\ref{def:values} and negation `$\sneg$' has been replaced by `$\neg$' from Boolean algebra. Thus, `$\neg$' is classical and satisfies all the axioms of `$\sneg$' plus \mbox{$\neg\neg t=t$}.
\review{R3.10}{Note also that, in the examples, we have followed the convention from~\cite{damasio2013justifications} of using the symbols `$\wedge$' and $\vee$ to respectively represent meet and join. However, in formal definitions, we will keep respectively using `$*$' and `$+$' for that purpose.}
% when we write provenance formulae.
We define a mapping  $\lambdap:\evalues\longrightarrow\boolAlgebra$ \ in the following recursive way:

\begin{align*}
\lambdap(t) \ &\eqdef \ \begin{cases}
    \lambdap(u) \otimes \lambdap(w)
        &\text{if } t=u \otimes v \text{ with } \otimes\in\set{+,*}
    \\
    \lambdap(u) \, * \, \lambdap(w)
        &\text{if } t=u \cdot v 
    \\
    \neg\lambdap(u)
        &\text{if } t=\sneg u
    \\
    l   &\text{if } t= l \text{ with } l \in Lb
\end{cases}
\end{align*}
% By abuse of notation we also denote by \mbox{$\lambdap(\eI):At\longrightarrow\boolAlgebra$} a mapping $\wI:At\longrightarrow\boolAlgebra$ from atoms to provenance values
% \mbox{$\lambdap(\eI)(\rA)=\lambdap(\eI(\rA))$} for all atom $\rA\in At$. 
% Then we define the provenance of some literal in the following way:

\begin{definition}[Provenance]\label{def:wellf.provenance}
 Given a  program $\cP$, the why-not provenance program $\wP(P)\eqdef \cP \cup \cP'$ where $\cP'$ contains a labelled fact of the form 
% \begin{gather*}
$(\sneg not(\rA) : \rA)$
% \end{gather*}
for each atom $\rA \in At$ not occurring in $\cP$ as a fact. We will write $\wP$ instead of $\wP(\cP)$ when the program $\cP$ is clear by the context. 
We denote by $Why_\cP(\rL)\eqdef\lambdap(\ewfmP{\wP}(\rL))$ the why-not provenance of a q-literal $\rL$.
We also say that a justification is hypothetical when $not(\rA)$ occurs oddly negated in it, non-hypothetical otherwise.\qed
\end{definition}

\begin{theorem}\label{thm:prov.correspondence}
Let $\cP$ be program over a finite signature~$\signature$. Then, the provenance of a literal according to Definition~\ref{def:wellf.provenance} is equivalent to the provenance defined by \cite{damasio2013justifications}.\qed
\end{theorem}

\begin{theorem}
\label{thm:well.non-positive.justifications.modify.program}
Let $\cP$ be program over a finite signature~$\signature$.
$\ewfm$ is the result of removing all non-hypothetical justification from $\ewfmP{\wP}$ and each occurrence of the form $\sneg\sneg not(\rA)$ for the remaining ones, that is,
$\ewfm = \varrho(\ewfmP{\wP})$ where $\varrho$ is the result of removing every label of the form $not(\rA)$, that is $\varrho$ is the composition of
$\varrho_{not(\rA_1)} \circ \varrho_{not(\rA_2)} \circ \dotsc \circ
\varrho_{not(\rA_n)}$ with $\at= \set{\rA_1, \rA_2, \dotsc, \rA_n }$.\qed
\end{theorem}

\noindent
On the one hand, Theorem~\ref{thm:prov.correspondence} shows that the provenance of a literal can be obtained by replacing the negation `$\sneg$' by `$\neg$' and `$\cdot$' by~`$*$' in the causal WFM of the augmented program~$\wP$.
On the other hand, Theorem~\ref{thm:well.non-positive.justifications.modify.program} asserts that non-hypothetical justifications of a program and its augmented one coincide when subterms of the form $\sneg\sneg not(\rA)$ are removed from justifications of the latter.
Consequently, we can \review{R1.8}{establish} the following correspondence between the ECJ justifications and the non-hypothetical WnP justifications.

\begin{theorem}\label{thm:wellf.justification<-$>$weff.provenace}
Let $P$ be program over a finite signature~$\signature$.
Then, the ECJ justifications of some atom $\rA$ (after replacing ``$\cdot$'' by ``$*$'' and ``$\sneg$'' by ``$\neg$'') correspond to the WnP justifications of $\rA$ (after removing every label of the form $not(\rB)$ with $\rB \in \at$), that is,
$\lambdap(\ewfm)(\rA) = \varrho(Why_\cP)(\rA)$
where $\varrho$ is the result of removing every label of the form $not(\rA)$ as in Theorem~\ref{thm:well.non-positive.justifications.modify.program}.\qed
\end{theorem}

\noindent
Theorem~\ref{thm:wellf.justification<-$>$weff.provenace} establishes a correspondence between non-hypothetical WnP-justifications and (flattened) ECJ justifications.
In our running example, \eqref{eq:e.just.yale.broken} is the unique causal justification of $dead_9$, while \eqref{eq:prov.just.yale.broken} (below) is its unique non-hypothetical WnP justification.
\begin{gather}
\begin{aligned}
	\neg water_3 \wedge shoot_8 &\wedge load_1 \wedge l_2 \wedge d_9 \wedge {}\\
		&\wedge not(dead_1) \wedge \dotsc \wedge not(dead_9)
		\wedge not(ab_1) \wedge \dotsc \wedge not(ab_8)
	\label{eq:prov.just.yale.broken}
\end{aligned}
\end{gather}
It is easy to see that, by applying $\lambdap$ to \eqref{eq:e.just.yale.broken} we obtain
\begin{gather}
\lambdap\big( ( \sneg water_3 * shoot_8 * load_1 \cdotl l_2 ) \cdot d_9 \big)
	\ = \
	\neg water_3 \wedge shoot_8 \wedge load_1 \wedge l_2 \wedge d_9
\end{gather}
which is just the result of removing all labels of the form `$not(\rA)$' from \eqref{eq:prov.just.yale.broken}.
The correspondence between the ECJ justification~\eqref{eq:e.just.yale.nobroken} and the WnP justification~\eqref{eq:prov.just.yale.repair} for program~\programref{prg:yale.nobroken} can be easily checked in a similar way. 

Hypothetical justifications are not directly captured by ECJ, but can be obtained using the augmented program~$\wP$ as stated by Theorem~\ref{thm:prov.correspondence}.
As a byproduct we establish a formal relation between WnP and~CG.

\begin{theorem}\label{thm:wellf.why-not-$>$sm.cjustification}
Let $\cP$ be a program  over a finite signature~$\signature$.
Then, every non-hypothetical and enabled WnP-justification $D$ of some atom~$\rA$ (after removing every label of the form $not(\rB)$ with $\rB \in \at$) is a justification with respect to every CG stable model $\cI$ (after replacing ``$\cdot$'' by ``$*$'' and ``$\sneg$'' by ``$\neg$''), that is $D \leq Why_P(\rA)$ implies $\varrho(D) \leq \lambdap(\cI)(\rA)$ where $\varrho$ is the result of removing every label of the form $not(\rB)$ as in Theorem~\ref{thm:well.non-positive.justifications.modify.program}.\qed
\end{theorem}

\noindent
Note that, as happened between the ECJ and CG justifications, the converse of Theorem~\ref{thm:wellf.why-not-$>$sm.cjustification} does not hold in general due to the well-founded vs stable model difference in their definitions. As an example, the explanation for atom $c$ at program~\programref{prog:asp.not.implies.wellf} has a unique WnP justification $c$ as opposed to the two CG~justifications, $c$ and $r_1 \cdotl r_3$.

\section{Contributory causes}
\label{sec:contributory.causes}

Intuitively, a \emph{contributory cause} is an event that has helped to produce some effect. For instance, in program~\programref{prg:yale.nobroken}, it is easy to identify both actions, $load_1$ and $shoot_8$, as events that have helped to produce $dead_9$ and, thus, they are both contributory causes of Fred's death. We may define the above informal concept of contributory cause as: any non-negated label $l$ that occurs in a maximal enabled justification of some atom $\rA$. Similarly, a \emph{contributory enabler} can be defined as a doubly negated label $\sneg\sneg l$ that occurs in a maximal enabled justification of some atom $\rA$. These definitions correctly identify $load_1$ and $shoot_8$ as contributory causes of $dead_9$ in program~\programref{prg:yale.nobroken} and $d$ as a contributory cause of $p$ in program~\programref{prg:bond}. Fact~$h$ is considered a contributory enabler of~$p$. These definitions will also suffice for dealing with what Hall~\citeyear{hall2007structural} calls \emph{trouble cases}: \emph{non-existent threats}, \emph{short-circuits}, \emph{late-preemption} and \emph{switching examples}.

It is worth to mention that, in the philosophic and AI literature, the concept of contributory cause is usually discussed in the broader sense of \emph{actual causation} which tries to provide an \emph{unique everyday-concept} of causation. Pearl~\citeyear{Pearl00} studied actual and contributory causes relying on \emph{causal networks}.
In this approach, it is possible to conclude cause-effect relations like ``$A$ has been an actual (resp. contributory) cause of $B$'' from the behaviour of structural equations by applying, under some \emph{contingency} (an alternative model in which some values are fixed) the \emph{counterfactual dependence} interpretation from~\cite{hume1748}: ``had $A$ not happened, $B$ would not have happened.''
Consider the following example which illustrates the difference between contributory and actual causes under this approach.

\begin{example}[Firing Squad]
Suzy and Billy form a two-man firing squad that responds to the order of the captain. The shot of any of the two riflemen would kill the prisoner. Indeed, the captain gives the order, both riflemen shoot together and the prisoner dies.\qed
\end{example}

On the one hand, the captain is an actual cause of the prisoner's death: ``had the captain not given the order, the riflemen would not have shot and the prisoner would not have died.'' On the other hand, each rifleman alone is not an actual cause: ``had one rifleman not shot, the prisoner would have died anyway because of the other rifleman.'' However, each rifleman's shot is a contributory cause because, under the contingency where the other rifleman does not shoot, the prisoner's death manifests counterfactual dependence on the first rifleman's shot. Later approaches like~\cite{HP01,HP05,hall2004,hall2007structural} have not made this distinction and consider the captain and the two riflemen as actual causes of the prisoner's death, while~\cite{halpern2015modification} considers the captain and the conjunction of both riflemen's shoots, but not each of them alone, as actual causes. We will focus here on representing the above concept of contributory cause and leave to the reader whether this agrees with the concept of cause in the every-day discourse or not.
% Note also that all these approaches does not differenciante between real causes and their enables, considering all of them as casues.

As has been slightly discussed in the introduction, Hall~\citeyear{hall2004,hall2007structural}, has emphasized the difference between two types of causal relations: \emph{dependence} and \emph{production}.
The former relies on the idea that ``counterfactual dependence between wholly distinct events is sufficient for causation.''
The latter is characterised by being \emph{transitive}, \emph{intrinsic}
(two processes following the same laws must be both or neither causal)
and \emph{local} (causes must be connected to their effects via sequences of causal intermediates).

These two concepts can be illustrated in Bond's example by observing the difference between pouring the drug (atom $d$), which is a cause under both understandings, and being a holiday (atom $h$), which is not considered a cause under the production viewpoint, although it is considered a cause under the dependence one.

In this sense,
all the above approaches to actual causation, but \cite{hall2004}, can be classified in the dependence category.
ECJ and CG do not consider $h$ a productive cause of $d$ because the \emph{default} (or \emph{normal}) behaviour of rule~\eqref{r1} is that ``$d$ causes $p$.''
This default criterion is also shared by~\cite{hall2007structural,Halpern08,hitchcock2009cause,halpern2011actual}.
Note that, ECJ (but not CG) captures the fact that $d$ counterfactually depends on $h$, as it considers it an enabler.
% only enables the default behaviour to flow.
In~\cite{hall2004}, the author relies on intrinsicness for rejecting $h$ as a productive cause of $d$: any causal structure (justification) including $h$ and $p$ would have to include the absence of the antidote (atom~$a$), and it would be enough that Bond had taken the antidote by another reason to break the counterfactual dependence between $h$ and~$p$.
By applying the above contributory cause definition to the WnP justification $h \wedge d \wedge r_1$ of Bond's paralysis (atom $p$) in program~\programref{prg:bond}, we can easily identify that $h$ is being considered a cause in WnP, thus, a causal interpretation of WnP clearly follows the dependence-based viewpoint.
On the other hand, the unique CG justification $d \cdotl r_1$ only considers $d$ as a cause, which illustrates the fact that CG is mostly related to the concept of production.
ECJ combines both understandings, and what is a cause under the dependence viewpoint is either an enabler or a cause under the production viewpoint.

In order to illustrate how ECJ can be used for representing the so-called \emph{non-existent threat} scenarios, consider a variation of Bond's example where today is not a holiday and, thus, Bond takes the antidote.
The poured drug $d$ is a threat to Bond's safety, represented as $s$, but that threat is prevented by the antidote. We may represent this scenario by program~\newprogram\label{prg:bond.threat} below:
\begin{gather*}
\begin{IEEEeqnarraybox}[][t]{l C l C l}
r_1 &: \ \ &  p & \leftarrow & d, \Not a\\  
r_2 &: \ \ &  a & \leftarrow & \Not h\\
r_3 &: \ \ &  s & \leftarrow & \Not p
\end{IEEEeqnarraybox}
\hspace{2cm}
\begin{IEEEeqnarraybox}[][t]{l C l C l}
&& a
\\
&& d
\end{IEEEeqnarraybox}
\end{gather*}
The causal WFM of program~\programref{prg:bond.threat} assigns
\begin{gather*}
\ewfmP{\programref{prg:bond.threat}}(s)
  \ = \ \sneg\sneg r_2 \cdotl r_3 \ + \ 
      \sneg d \cdotl r_3 \ + \
      \sneg r_1 \cdotl r_3
\end{gather*}
which recognises rule~$r_2$ (taking the antidote) as a contributory enabler of Bond's safety. The difficulty in this kind of scenarios consists in avoiding the wrong recognition of $r_2$ as an enabler when the threat $d$ does not exist. If we remove fact $d$ from~\programref{prg:bond.threat} to get the new program~\newprogram\label{prg:bond.no-threat}  then we obtain that $\ewfmP{\programref{prg:bond.no-threat}}(d) = 0$ and, consequently,
$\ewfmP{\programref{prg:bond.no-threat}}(s) = r_3$. Intuitively, in the absence of any threat, Bond is just safe because that is his default behaviour as stated by rule $r_3$.

\emph{Short-circuit} examples consist in avoiding the wrong recognition of an event as a contributory enabler that provokes a threat that eventually prevents itself. Consider the program~\newprogram\label{prg:short-circuit} below:
\begin{gather*}
\begin{IEEEeqnarraybox}[][t]{l C l C l}
r_1 &: \ \ &  p & \leftarrow & a,\, \Not f\\
r_2 &: \ \ &  f & \leftarrow & c,\, \Not b\\
r_3 &: \ \ &  b & \leftarrow & c
\end{IEEEeqnarraybox}
\hspace{2cm}
\begin{IEEEeqnarraybox}[][t]{l C l C l}
&&a\\
&&c
\end{IEEEeqnarraybox}
\end{gather*}
Here, $c$ is a threat to $p$, since it may cause $f$ through rule $r_2$. However, $c$ eventually prevents $r_2$, since it also causes $b$ through rule $r_3$. The causal WFM of program~\programref{prg:short-circuit} assigns
\begin{gather*}
\ewfmP{\programref{prg:short-circuit}}(p)
  \ = \   (a * \sneg\sneg r_3) \cdotl r_1 \ + \ 
      (a * \sneg r_2 ) \cdotl r_1 \ + \ 
      (a * \sneg c ) \cdotl r_1
\end{gather*}
which correctly avoids considering $c$ as a contributory enabler of $p$ and recognises $r_3$ as the enabler of $p$.
Note that $c$ is actually considered an inhibitor due to justification $(a * \sneg c ) \cdotl r_1$ pointing out that, had $c$ not happened, then $a \cdotl r_1$ would have been an enabled justification. But then, $(a * \sneg\sneg r_3) \cdotl r_1$ would stop being a justification since $(a * \sneg\sneg r_3) \cdotl r_1 + a \cdotl r_1 = a \cdotl r_1$.

To illustrate \emph{late-preemention} consider the following example from~\cite{lewis2000causation}.

\begin{example}[Rock Throwers]\label{ex:throwers}
Billy and Suzy throw rocks at a bottle. Suzy throws first and her rock arrives first. The bottle shattered. When Billy's rock gets to where the bottle used to be, there is nothing there but flying shards of glass.
Who has caused the bottle to shatter?\qed
\end{example}

The key of this example is to recognise that Suzy, and not Billy, has caused the shattering. The usual way of representing this scenario in the actual causation literature is by introducing two new fluents $hit\_\,{suzy}$ and $hit\_\,{billy}$ in the following way~\cite{hall2007structural,halpern2011actual,halpern2014appropriate,halpern2015modification}:
\begin{IEEEeqnarray}{l L l}
hit\_\,{suzy}  &\ \leftarrow \ throw\_\,{suzy}
  \label{eq:hit.suzy.law}
\\
 hit\_\,{billy}  &\ \leftarrow \ throw\_\,{billy},\; \neg hit\_\,{suzy}
  \label{eq:hit.billy.law}
\\
shattered &\ \leftarrow \ hit\_\,{suzy}
\\
shattered &\ \leftarrow \ hit\_\,{billy}
\end{IEEEeqnarray}
It is easy to see that such a representation mixes, in law~\eqref{eq:hit.billy.law}, both the description of the world and the narrative fact asserting that Suzy threw first. This may easily lead to a problem of elaboration tolerance. For instance, if we have $N$ shooters and they shoot sequentially we would have to modify the equations for all of them in an adequate way, so that the last shooter's equation would have the negation of the preceding $N$-1 and so on. Moreover, all these equations would have to be \emph{reformulated} if we simply change the shooting order.
On the other hand, we may represent this scenario by a program~\newprogram\label{prg:throwers} consisting of the following rules
\begin{gather*}
\begin{IEEEeqnarraybox}[][t]{l C? l L l}
s_{t+1} &:& shattered_{t+1}  &\ \leftarrow \ throw(A)_t,\, \Not \ shattered_t
  % \label{eq:shattered.suzy.law}
\\
&& \overline{shattered}_{0}
\end{IEEEeqnarraybox}
\hspace{1.75cm}
\begin{IEEEeqnarraybox}[][t]{l l rl l l}
throw({suzy})_0
\\
throw({billy})_1
\end{IEEEeqnarraybox}
\end{gather*}
with $A \in \set{suzy,\, billy}$,
plus the following rules corresponding to the inertia axioms 
\begin{align*}
shattered_{t+1} &\ \leftarrow \ shattered_{t},\ \Not \ \overline{shattered}_{t+1}
\\
\overline{shattered}_{t+1}
                &\ \leftarrow \ \overline{shattered}_{t},\ \Not \ shattered_{t+1}
\end{align*}
Atom $shattered_2$ holds in the standard WFM of~\programref{prg:throwers} and its justification corresponds to
\begin{gather*}
% \ewfmP{\programref{prg:throwers}}(shattered_2)
%   \ \ = \ \
throw(suzy)_0 \cdotl s_1 
      \ \ + \ \ (\sneg throw(suzy) * throw(billy)_1) \cdotl s_2
      \ \ + \ \ (\sneg s_1 * throw(billy)_1) \cdotl s_2
\end{gather*}
On the one hand, the first addend points out that fact $throw(suzy)_0$ has caused $shattered_2$ by means of rule $s_1$.
On the other hand, the second addend indicates that $throw(billy)_1$ has not caused it because Suzy's throw has prevented it.
Finally, the third addend means that $throw(billy)_1$ would have caused the shattering if it were not for rule $s_1$.
This example shows how our semantics is able to recognise that it was Suzy, and not Billy, who caused the bottle shattering.
Furthermore, it also explains that Billy did not cause it because Suzy did it first.

Finally, consider the following example from~\cite{hall2000causation}.

\begin{example}[The Engineer]\label{ex:actual.engineer}
An engineer is standing by a switch in the railroad
tracks. A train approaches in the distance. She flips the switch, so that
the train travels down the right-hand track, instead of the left. Since
the tracks reconverge up ahead, the train arrives at its destination all
the same; let us further suppose that the time and manner of its arrival
are exactly as they would have been, had she not flipped the switch.\qed
\end{example}

This has been a controversial example. In \cite{hall2000causation}, the author has argued that the switch should be considered a cause of the arrival because  $switch$ has contributed to the fact that the train has travelled down the right-hand track.
In a similar manner, it seems clear that the train travelling down the right-hand track has contributed to the train arrival.
If causality is considered to be a transitive relation, as \cite{hall2000causation} does, the immediate consequence of the above reasoning is that flipping the $switch$ has contributed to the train $arrival$.
In \cite{hall2007structural} he argues otherwise and points out that commonsense tells  that the $switch$ is not a cause of the $arrival$.
\cite{HP05} had considered $switch$ a cause of $arrival$ depending on whether the train travelling down the tracks is represented by one or two variables in the model.
Although our understanding of causality is closer \review{R.14}{to the one expressed} in \cite{hall2007structural}, it is not the aim of this work to go more in depth in this discussion, but to show instead how both understandings can be represented in ECJ.
Consider the following program~$\newprogram\label{prg:engineer}$
\begin{align*}
\begin{IEEEeqnarraybox}{lC? l C l}
r_1 &:& arrival   &\ \leftarrow \ & right
\\
r_2 &:& arrival   &\ \leftarrow \ & left
% \end{IEEEeqnarraybox}
% \hspace{1.5cm}
% \begin{IEEEeqnarraybox}{lC? l C l}
\\
r_3 &:& right   &\ \leftarrow \ & train,\, \Not \overline{switch}
\\
r_4 &:& left  &\ \leftarrow \ & train,\, \Not switch
\end{IEEEeqnarraybox}
\hspace{1.5cm}
\begin{IEEEeqnarraybox}{lC? l C l}
&& \overline{switch}   &\ \leftarrow \ & \Not switch
\\
&& switch  &\ \leftarrow \ &\Not \overline{switch}
% \end{IEEEeqnarraybox}
% \hspace{1.5cm}
% \begin{IEEEeqnarraybox}{lC? l C l}
\\
&& train
\\
&& switch
\end{IEEEeqnarraybox}
\end{align*}
where $\overline{switch}$ represents the strong negation of $switch$.
The two unlabelled rules capture the idea that the switch behaves classically, that is, it must be activated or not.
The literal $\Not \overline{switch}$ in the body of rule $r_3$ points out that the $switch$ position is an enabler and not a cause of the track taken by the train. This representation can be arguable, but the way in which the  rule has been written would be expressing that if a train is coming, then a train will cross the right track by default unless $\overline{switch}$ prevents it. In that sense, the only productive cause for $right$ (a train in the right track) is $train$ (a train is coming) whereas the switch position just enables the causal rule to be applied. A similar default $r_4$ is built for the left track, flipping the roles of $switch$ and $\overline{switch}$.

The causal WFM of program~$\programref{prg:engineer}$ corresponds to
\begin{gather*}
\begin{IEEEeqnarraybox}[][t]{lCl}
% \ewfm(left)   &\ = \ & (train * \sneg switch) \cdotl r_3
% \\
% \ewfm(right)  &\ = \ & (train * switch) \cdotl r_4
% \\
\ewfmP{\programref{prg:engineer}}(arrival)
  &\ \ = \ \ & 
      (train * \sneg\sneg switch) \cdotl r_3 \cdotl r_1
      \ + \ 
      (train * \sneg switch) \cdotl r_4 \cdotl r_2
\end{IEEEeqnarraybox}
\end{gather*}
It is easy to see that $switch$ is a doubly-negated label occurring in the maximal enabled justification
$E_1 = (train * \sneg\sneg switch) \cdotl r_3 \cdotl r_1$ and, thus, we may identify it as a contributory enabler of $arrival$, but not its productive cause. On the other hand, by looking at the inhibited justification
$E_2 = (train * \sneg switch) \cdotl r_4 \cdotl r_2$, we observe that
$switch$ is also preventing rules $r_4$ and $r_2$ to produce the same effect, $arrival$, that is helping to produce in $E_1$.

If we want to ignore the way in which the train arrives, one natural possibility is using the same label for all the rules for atom $arrival$, reflecting in this way that we do not want to trace whether $r_1$ or $r_2$ has been actually used. Suppose we label $r_2$ with $r_1$ instead, leading to the new program~$\newprogram\label{prg:engineer.same}$
\begin{align*}
\begin{IEEEeqnarraybox}{lC? l C l}
r_1 &:& arrival   &\ \leftarrow \ & right
\\
r_1 &:& arrival   &\ \leftarrow \ & left
% \end{IEEEeqnarraybox}
% \hspace{1.5cm}
% \begin{IEEEeqnarraybox}{lC? l C l}
\\
r_3 &:& right   &\ \leftarrow \ & train,\, \Not \overline{switch}
\\
r_3 &:& left  &\ \leftarrow \ & train,\, \Not switch
\end{IEEEeqnarraybox}
\hspace{1.5cm}
\begin{IEEEeqnarraybox}{lC? l C l}
&& \overline{switch}   &\ \leftarrow \ & \Not switch
\\
&& switch  &\ \leftarrow \ &\Not \overline{switch}
% \end{IEEEeqnarraybox}
% \hspace{1.5cm}
% \begin{IEEEeqnarraybox}{lC? l C l}
\\
&& train
\\
&& switch
\end{IEEEeqnarraybox}
\end{align*}
whose causal WFM corresponds to
\begin{gather*}
\begin{IEEEeqnarraybox}[][t]{lCl C l}
% \ewfm(left)   &\ = \ & (train * \sneg switch) \cdotl r_3
% \\
% \ewfm(right)  &\ = \ & (train * switch) \cdotl r_4
% \\
\ewfmP{\programref{prg:engineer}}(arrival)
  &\ \ = \ \ & 
      (train * \sneg\sneg switch) \cdotl r_3 \cdotl r_1
      \ + \ 
      (train * \sneg switch) \cdotl r_3 \cdotl r_1
  &\ \ = \ \ & 
      train \cdotl r_3 \cdotl r_1
\end{IEEEeqnarraybox}
\end{gather*}
As we can see, this justification does not consider $switch$ at all as a cause of the arrival (nor even a contributory enabler, as before). In other words, $switch$ is \emph{irrelevant} for the train $arrival$, which probably coincides with the most common intuition.

However, we do not find this solution fully convincing yet, because the explanation we obtain for $right$, $\ewfmP{\programref{prg:engineer.same}}(right) = (train * \sneg\sneg switch) \cdotl r_3$ is showing that $switch$ is just acting as an enabler, as we commented before. If we wanted to represent $switch$ as a contributory cause of $right$, we would have more difficulties to simultaneously keep $switch$ irrelevant in the explanation of $arrival$. One possibility we plan to explore in the future is allowing the declaration of a given atom or fluent, like our $switch$, as \emph{classical} so that we include both, the the rule:
\begin{eqnarray*}
switch \leftarrow \Not \Not switch 
\end{eqnarray*}
in the logic program\footnote{This implication actually corresponds to a \emph{choice rule} $0 \{switch\} 1$, commonly used in Answer Set Programming.} and the axiom $\sneg \sneg switch=switch$ in the algebra. The latter immediately implies $switch + \sneg switch = 1$ (due to the weak excluded middle axiom).

Then, \programref{prg:engineer.same} could be simply expressed as
\begin{align*}
\begin{IEEEeqnarraybox}[][t]{lC? l C l}
r_1 &:& arrival   &\ \leftarrow \ & right
\\
r_1 &:& arrival   &\ \leftarrow \ & left
\\
r_3 &:& right   &\ \leftarrow \ & train,\, switch
\\
r_3 &:& left  &\ \leftarrow \ & train,\, \Not switch
\end{IEEEeqnarraybox}
\hspace{1.5cm}
\begin{IEEEeqnarraybox}[][t]{lC? l C l}
&& train
\\
&& switch
\end{IEEEeqnarraybox}
\end{align*}
and the justification of $right$ and $left$ would~become
\begin{gather*}
\begin{IEEEeqnarraybox}[][t]{lCl  " C " lCl}
\ewfm(right)
  &\ \ = \ \ & 
      (train * switch) \cdotl r_3
&&
\ewfm(left)
  &\ \ = \ \ & 
      (train * \sneg switch) \cdotl r_3
\end{IEEEeqnarraybox}
\end{gather*}
pointing out that $switch$ is a cause (resp. an inhibitor) of the train travelling down the right (resp. left) track.
Then, the justification of arrival would be
\begin{gather*}
\begin{IEEEeqnarraybox}[][t]{lCl C l}
% \ewfm(left)   &\ = \ & (train * \sneg switch) \cdotl r_3
% \\
% \ewfm(right)  &\ = \ & (train * switch) \cdotl r_4
% \\
\ewfm(arrival)
  &\ \ = \ \ & 
      (train * switch) \cdotl r_3 \cdotl r_1
      \ + \ 
      (train * \sneg switch) \cdotl r_3 \cdotl r_1
  &\ \ = \ \ & 
      train \cdotl r_3 \cdotl r_1
\end{IEEEeqnarraybox}
\end{gather*}
We leave the study of this possibility for a future deeper analysis.

%%%%%%%%%%%%%%%%%%%%%%%%%%%%%%%%%%%%%%%%%%%%%%%%%%%%%%%%%%%%%%%%%%%%%%%%%%%%%%%%%%%%%%%%%
%%%%%%%%%%%%%%%%%%%%%%%%%%%%%%%%%%%%%%%%%%%%%%%%%%%%%%%%%%%%%%%%%%%%%%%%%%%%%%%%%%%%%%%%%
%%%%%%%%%%%%%%%%%%%%%%%%%%%%%%%%%%%%%%%%%%%%%%%%%%%%%%%%%%%%%%%%%%%%%%%%%%%%%%%%%%%%%%%%%

% \section{Related Work}\label{sec:rel}

%%%%%%%%%%%%%%%%%%%%%%%%%%%%%%%%%%%%%%%%%%%%%%%%%%%%%%%%%%%%%%%%%%%%%%%%%%%%%%%%%%%%%%%%%
%%%%%%%%%%%%%%%%%%%%%%%%%%%%%%%%%%%%%%%%%%%%%%%%%%%%%%%%%%%%%%%%%%%%%%%%%%%%%%%%%%%%%%%%%
%%%%%%%%%%%%%%%%%%%%%%%%%%%%%%%%%%%%%%%%%%%%%%%%%%%%%%%%%%%%%%%%%%%%%%%%%%%%%%%%%%%%%%%%%

\section{Conclusions and other related work}
\label{sec:conc} 

In this paper we have introduced a unifying approach that combines causal production with enablers and inhibitors. We formally capture inhibited justifications by introducing a ``non-classical'' negation `$\sneg$' in the algebra of causal graphs (CG). An inhibited justification is nothing else but an expression containing some negated label. We have also distinguished productive causes from enabling conditions (counterfactual dependences that are not productive causes) by using a double negation `$\sneg\sneg$' for the latter. 
%These justifications can be obtained by mimicking the alternating fixpoint operator $\Gamma^2$. 
The existence of enabled justifications is a sufficient and necessary condition for the truth of a literal. Furthermore, our justifications capture, under the Well-Founded Semantics, both Causal Graph and Why-not Provenance justifications. As a byproduct we established a formal relation between these two approaches.

We have also shown how several standard examples from the literature on actual causation can be represented in our formalism and illustrated how this representation is suitable for domains which include dynamic defaults -- those whose behaviour are not predetermined, but rely on some program condition -- as for instance the inertia axioms.
As pointed out by~\cite{maudlin2004}, causal knowledge can be structured by a combination of \emph{inertial laws} -- how the world would evolve if nothing intervened
%\footnote{Although not necessarily in the traditional Newtonian sense of things remaining the same.}
-- and \emph{deviations} from these inertial laws.
% The importance of default knowledge has been widely recognised as a cornerstone of the problem of actual causation in \cite{Halpern08,hitchcock2009cause} 
% %halpern2011actual,
% among others.

In addition to the literature on actual causes cited in Section~\ref{sec:contributory.causes}, our work also relates to papers on reasoning about actions and change~\cite{Lin95,McCain97,Thielscher97}.
These works
have been traditionally focused on using causal inference to solve representational problems (such as, the frame, ramification and qualification problems) without paying much attention to the derivation of cause-effect relations.
Focusing on LP, our work obviously relates to explanations obtained from ASP debugging approaches~ \cite{specht1993generating,denecker1993justification,pemmasani2004online,GPST08,pontelli2009justifications,oetsch2010,schultz2013aba}.
%roychoudhury2000justifying,
The most important difference of these works with respect to ECJ, and also WnP and CG, is that the last three provide fully algebraic semantics in which justifications are embedded into program models.
A formal relation between \cite{pontelli2009justifications} and WnP was established in \cite{damasio2013justifications} and so, using Theorems~\ref{thm:prov.correspondence} and~\ref{thm:wellf.justification<-$>$weff.provenace}, it can be directly extended to ECJ, but at the cost of flattening the graph information (i.e. losing the order among rules).

Interesting issues for future study are incorporating enabled and inhibited justifications to the stable model semantics and replacing the syntactic definition in favour of a logical treatment of default negation, as done for instance with the Equilibrium Logic~\cite{Pearce96} characterisation of stable models.
Other natural steps would be the consideration of syntactic operators, for capturing more specific knowledge about causal information as done in~\cite{fandinno2015aspocp} capturing sufficient causes in the CG approach,
and also the representation of non-deterministic causal laws,
by means of disjunctive programs or the incorporation of probabilistic knowledge.
% From a KR point of view, another interesting future line of study is to apply our semantics to other traditional examples of the actual causation literature like \emph{short-circuits} and \emph{switches}~\cite{hall2007structural}.

\paragraph{Acknowledgements} We are thankful to Carlos Dam\'asio  for his suggestions and comments on earlier versions of this work. We also thank the anonymous reviewers for their help to improve the paper.
This research was partially supported by Spanish Project TIN2013-42149-P.

\bibliographystyle{acmtrans}
\bibliography{refs}

\newpage

\appendix

% \makeatletter
% \renewcommand\section{\@startsection {section}{1}{\z@}%
%                                    {-3.5ex \@plus -1ex \@minus -.2ex}%
%                                    {2.3ex \@plus.2ex}%
%                                    {\normalfont\scshape}}
% \makeatletter

%%%%%%%%%%%%%%%%%%%%%%%%%%%%%%%%%%%%%%%%%%%%%%%%%%%%%%%%%%%%%%%%%%%%%%%%%%%%%%%%%%%%%%
%%%%%%%%%%%%%%%%%%%%%%%%%%%%%%%%%%%%%%%%%%%%%%%%%%%%%%%%%%%%%%%%%%%%%%%%%%%%%%%%%%%%%%
%%%%%%%%%%%%%%%%%%%%%%%%%%%%%%%%%%%%%%%%%%%%%%%%%%%%%%%%%%%%%%%%%%%%%%%%%%%%%%%%%%%%%%

%%%%%%%%%%%%%%%%%%%%%%%%%%%%%%%%%%%%%%%%%%%%%%%%%%%%%%%%%%%%%%%%%%%%%%%%%%%%%%
\section{Auxiliary figures}\label{sec:figs}
%%%%%%%%%%%%%%%%%%%%%%%%%%%%%%%%%%%%%%%%%%%%%%%%%%%%%%%%%%%%%%%%%%%%%%%%%%%%%%

\vspace{1cm}

\begin{figure}[htbp]
\begin{center}
$
\begin{array}{c}
\hbox{\em Associativity} \\
\hline
%$
\begin{array}{r@{\ }c@{\ }r@{}c@{}l c r@{}c@{}l@{\ }c@{\ }l@{\ }}
t & + & (u & + & w) & = & (t & + & u) & + & w\\
t & * & (u & * & w) & = & (t & * & u) & * & w
\end{array}
%$
\end{array}
$
\ \
$
\begin{array}{c}
\ \ \ \ \hbox{\em Commutativity}\ \ \ \ \\
\hline
%$
\begin{array}{r@{\ }c@{\ }l c r@{\ }c@{\ }l@{\ }}
t & + & u & = & u & + & t\\ 
t & * & u & = & u & * & t
\end{array}
%$
\end{array}
$
\ \
$
\begin{array}{c}
\hbox{\em Absorption} \\
\hline
%$
\begin{array}{c c r@{\ }c@{\ }r@{}c@{}l@{\ }}
t & = & t & + & (t & * & u)\\
t & = & t & * & (t & + & u)
\end{array}
%$
\end{array}
$
\ \
\\
\vspace{10pt}
$
\begin{array}{c}
\hbox{\em Distributive} \\
\hline
%$
\begin{array}{r@{\ }c@{\ }r@{}c@{}l c r@{}c@{}l@{\ }c@{\ }r@{}c@{}l@{}}
t & + & (u & * & w) & = & (t & + & u) & * & (t & + & w)\\
t & * & (u & + & w) & = & (t & * & u) & + & (t & * & w)
\end{array}
%$
\end{array}
$
\ \
$
\begin{array}{c}
Identity \\
\hline
%$
\begin{array}{rcr@{\ }c@{\ }l@{\ }}
t & = & t & + & 0\\
t & = & t & * & 1
\end{array}
%$
\end{array}
$
\ \
$
\begin{array}{c}
\hbox{\em Idempotency} \\
\hline
%$
\begin{array}{rcr@{\ }c@{\ }l@{\ }}
t & = & t & + & t\\
t & = & t & * & t
\end{array}
%$
\end{array}
$
\ \ 
$
\begin{array}{c}
\hbox{\em Annihilator} \\
\hline
%$
\begin{array}{rcr@{\ }c@{\ }l@{\ }}
1 & = & 1 & + & t\\
0 & = & 0 & * & t
\end{array}
%$
\end{array}
$
\end{center}
\caption{Sum and product satisfy the properties of a completely distributive lattice.}
\label{fig:DBLattice}
\end{figure}

%%%%%%%%%%%%%%%%%%%%%%%%%%%%%%%%%%%%%%%%%%%%%%%%%%%%%%%%%%%%%%%%%%%%%%%%%%%%%%%%%%%%
%%%%%%%%%%%%%%%%%%%%%%%%%%%%%%%%%%%%%%%%%%%%%%%%%%%%%%%%%%%%%%%%%%%%%%%%%%%%%%%%%%%%
%%%%%%%%%%%%%%%%%%%%%%%%%%%%%%%%%%%%%%%%%%%%%%%%%%%%%%%%%%%%%%%%%%%%%%%%%%%%%%%%%%%%

\vspace{3cm}

\section{Proofs of Theorems and Implicit Results}
\label{sec:proofs}

In the following, by abuse of notation, for every function $f:\evalues \longrightarrow \evalues$, we will also denote by $f$ a function over the set of interpretations such that $f(\eI)(\rA) = f(\eI(\rA))$ for every atom $\rA \in \at$.
We have organized the proofs into different subsections.

\subsection{Proofs of Propositions~\ref{prop:term.neg.leq} to~\ref{prop:negation.normal.form}}

% \begin{definition}
% A term $t$ is said to be in disjunctive normal form (DNF) iff $t$ is in NNF and, additionally, sums~``$+$'' are not in the scope of any other operation and every application~``$\cdot$'' subterm is elementary.
% \end{definition}

\begin{proposition}\label{prop:term.neg.leq}
Negation `$\sneg$' is anti-monotonic. That is $t\leq u$ holds if and only if $\sneg t \geq \sneg u$ for any given two causal terms $t$ and $u$.\qed
\end{proposition}
\begin{proof}
By definition $t\leq u$ iff $t*u=t$.
Furthermore, by De Morgan laws,
$\sneg(t*u) = \sneg t + \sneg u$
and, thus, $\sneg(t*u)=\sneg t$ iff $\sneg t + \sneg u = \sneg t$.
Finally, just note that $\sneg t + \sneg u = \sneg t$ iff $\sneg t \geq \sneg u$.
Hence, $t\leq u$ holds iff $\sneg t \geq \sneg u$. \qed
\end{proof}

\begin{proposition}\label{prop:closure}
The map $t \mapsto \sneg\sneg t$ is a closure.
That is, it is monotonic, idempotent and it holds that $t \leq \sneg\sneg t$ for any given causal term $t$.\qed
\end{proposition}

% \comment{the application of distributivity in the proof that t * ~~t = t is not clear, and there are other inferences implicit in that step.}

\begin{proof}
To show that $t \mapsto \sneg\sneg t$ is monotonic just note that $t \mapsto \sneg t$ is antimonotonic (Proposition~\ref{prop:term.neg.leq})
and then $t \leq u$ iff $\sneg t \geq \sneg u$ iff $\sneg\sneg t \leq \sneg\sneg u$.
Furthermore, $\sneg\sneg( \sneg\sneg t) = \sneg(\sneg\sneg\sneg t) = \sneg\sneg t$, that is,
$t \mapsto \sneg\sneg t$ is idempotent.
Finally, note that, by definition, $t \leq \sneg\sneg t$ iff $t * \sneg\sneg t = t$ and
\review{R3.16}{\begin{align*}
t * \sneg\sneg t
  &\ = \ t * \sneg\sneg t \ + \ 0
  &&(\text{identity})\\
  &\ = \ t * \sneg\sneg t \ + \ t * \sneg t
  &&(\text{pseudo-complement})\\
  &\ = \ (t * \sneg\sneg t + t ) * (t * \sneg\sneg t + \sneg t )
  &&(\text{distributivity})\\
  &\ = \ (t + t) \ * \ (\sneg\sneg t + t) \ * \ (t + \sneg t) \ * \ (\sneg\sneg t+ \sneg t)
  &&(\text{distributivity})\\
  &\ = \ t \ * \ (\sneg\sneg t + t) \ * \ (t + \sneg t) \ * \ (\sneg\sneg t+ \sneg t)
  &&(\text{idempotency})\\
  &\ = \ t \ * \ (t + \sneg t) \ * \ (\sneg\sneg t + t) \ * \ 1
  &&(\text{w.\ excluded\ middle})\\
  &\ = \ t \ * \ (t + \sneg t) \ * \ (\sneg\sneg t + t)
  &&(\text{identity})\\
  &\ = \ t \ * \ (\sneg\sneg t + t)
  &&(\text{absorption})\\
  &\ = \ t 
  &&(\text{absorption})
\end{align*}}
Hence, $t \mapsto \sneg\sneg t$ is a closure.\qed
\end{proof}

\begin{proposition}\label{prop:negation.normal.form}
Given any term $t$, it can be rewritten as an equivalent term $u$ in negation and disjuntive normal forms.\qed
\end{proposition}
\begin{proof}
This is a trivial proof by structural induction using the DeMorgan laws and negation of application axiom.
Furthermore, using the axiom $\sneg\sneg\sneg t=t$ no more than two nested negations are required.
Furthermore, it is easy to see that by applying distributivity of ``$\cdot$'' and ``$*$'' over ``$+$,'' every term can be equivalently represented as a term ``$+$'' is not in the scope of any other operation.
Moreover, applying distributivity of ``$\cdot$'' over ``$*$'' every such term can be represented as one in every application subterm is elementary.\qed
\end{proof}

\begin{lemma}
Let $t$ be a join irreducible causal value.
Then, either $t * \sneg\sneg u = 0$ or $t * \sneg\sneg u$ is join irreducible for every causal value $u \in \evalues$.\qed
\end{lemma}

\begin{proof}
Suppose that $t*u$ is not join irreducible and let $W \subseteq \evalues$ a set of causal values such that $w \neq t* \sneg\sneg u$ for every $w \in W$ and $t* \sneg\sneg u = \sum_{w \in W} w$.
Since $t* \sneg\sneg u = \sum_{w \in W} w$, it follows that $w \leq t* \sneg\sneg u$ for every $w \in W$ and, since $w \neq t* \sneg\sneg u$, it follows that $w < t* \sneg\sneg u$ for every $w \in W$.
Furthermore,
% $t * u + t * \sneg\sneg u + t * \sneg u = t * (u + \sneg\sneg u) + t * \sneg u = t * \sneg\sneg u + t * \sneg u$.
% Also,
$t * \sneg\sneg u + t * \sneg u = t * (\sneg\sneg u + \sneg u) = t$.

Since $t$ is join irreducible, it follows that
either $t = t*\sneg\sneg u$ or $t=t*\sneg u$.
If $t=t*\sneg u$, then $t* \sneg\sneg u=(t*\sneg u)* \sneg\sneg u=0$.
Otherwise, $t=t*\sneg\sneg u$ and $t$ is join irreducible by hypothesis.\qed
\end{proof}

\begin{lemma}\label{lem:lambdap.neg.homomorphism}
Let $t$ be a term. Then $\lambdap(\sneg t) = \neg\lambdap(t)$.\qed
\end{lemma}
\begin{proof}
We proceed by structural induction assuming that $t$ is in negated normal form. In case that $t=a$ is elementary, it follows that \mbox{$\lambdap(\sneg a)=\neg a=\neg \lambdap(a)$}. In case that $t=\sneg a$ with $a$ elementary, \mbox{$\lambdap(\sneg t)=\lambdap(\sneg\sneg a)$} and \mbox{$\lambdap(\sneg\sneg a) = a  =  \neg\neg a =\neg \lambdap(\sneg a)=\lambdap(t)$}.
In case that \mbox{$t=\sneg\sneg a$}, with $a$ elementary, \mbox{$\lambdap(\sneg t)=\lambdap(\sneg\sneg\sneg a)$} and
$$\lambdap(\sneg\sneg\sneg a) = \lambdap(\sneg a) = \neg a = \neg \lambdap(\sneg\sneg a) = \neg\lambdap(t)$$

In case that $t=u + v$. Then
$$\lambdap(\sneg t)=\lambdap(\sneg u * \sneg v)= \lambdap(\sneg u) \wedge \lambdap (\sneg v)$$
By induction hypothesis $\lambdap(\sneg u) = \neg\lambdap(u)$ and $\lambdap(\sneg v) = \neg\lambdap(v)$ and, therefore, it holds that \mbox{$\lambdap(\sneg t)  =  \neg\lambdap(u)\wedge \neg\lambdap(v)$}.
Thus, $\neg\lambdap(t)=\neg(\lambdap(u)\vee\lambdap(v)) = \neg\lambdap(u)\wedge\neg\lambdap(v) = \lambdap(\sneg t)$.

In case that $t=u \otimes v$ with $\otimes\in\set{*,\cdot}$. Then $\lambdap(\sneg t)=\lambdap(\sneg u + \sneg v)= \lambdap(\sneg u) \vee \lambdap (\sneg v)$ and by induction hypothesis $\lambdap(\sneg u) = \neg\lambdap(u)$ and $\lambdap(\sneg v) = \neg\lambdap(v)$. Consequently it holds that \mbox{$\lambdap(\sneg t) = \neg\lambdap(t)$}.\qed
\end{proof}

\begin{lemma}\label{lem:lambdaprov.neg.homomorphism.formula}
Let $t$ be a term and $\phi$ a provenance term. If \mbox{$\phi\leq\lambdap(t)$}, then $\lambdap(\sneg t)\leq \neg\phi$ and if $\lambdap(t)\leq\phi$, then \mbox{$\neg\phi\leq\lambdap(\sneg t)$}.\qed
\end{lemma}
\begin{proof}
If $\phi\leq\lambdap(t)$, then $\phi = \lambdap(t) * \phi$ and then $\neg\phi =\neg\lambdap(t)+\neg\phi$ and, by Lemma~\ref{lem:lambdap.neg.homomorphism}, it follows that $\neg\phi = \lambdap(\sneg t) + \neg\phi$. Hence $\lambdap(\sneg t)\leq \neg\phi$.
Furthermore if $\lambdap(t)\leq\phi$, then $\phi = \lambdap(t) + \phi$ and then $\neg\phi =\neg\lambdap(t)*\neg\phi$ and, by Lemma~\ref{lem:lambdap.neg.homomorphism}, it follows that $\neg\phi = \lambdap(\sneg t) * \neg\phi$. Hence $\neg\phi\leq\lambdap(\sneg t)$.\qed
\end{proof}

\subsection{Proof of Theorem~\ref{thm:tp.properties}}

The proof of Theorem~\ref{thm:tp.properties} will relay on the definition of the following direct consequence operator 
\begin{align*}
\tp(\cI)(\rH)
    \ \ &\eqdef \ \
      \sum \big\{ \ \big( \ \cI(B_1) * \dotsc * \cI(B_n) \ \big) \cdot r_i \
% \\&\hspace{2.5cm}
        \mid \ (r_i: \ \rH \leftarrow B_1, \dotsc, B_n ) \in P \ \big\}
\end{align*}
for any CG interpretation $\cI$ and atom $\rH \in \at$.
Note that the definition of this direct consequence operator $\tp$ is analogous to the $\etp$ operator, but the domain and image of $\tp$ are the set of CG interpretations while the domain and image of $\etp$ are the set of ECJ interpretations.

\begin{theorem}[Theorem~2 from \citeNP{CabalarFF14}]\label{theorem:tp.properties}
Let $P$ be a (possibly infinite) positive logic program with $n$ causal rules.
Then,
($i$) $\mathit{lfp}(\tp)$ is the least model of $P$, and
($ii$)  $\mathit{lfp}(\tp)=\tpr{\omega}=\tpr{n}$.\qed
\end{theorem}

\begin{proofof}{Theorem~\ref{thm:tp.properties}}
Assume that every term occurring in $\cP$ is NNF and let $\cQ$ be the program obtained by renaming in $\cP$ each occurrence of $\sneg l$ as $l'$ and each occurrence of $\sneg\sneg l$ as $l''$ with $l'$ and $l''$ new symbols.
Note that this renaming implies that $\sneg l$ and $\sneg\sneg l$ are treated as completely independent symbols from $l$ and, thus, all equalities among terms derived from program $\cQ$ are also satisfied by $\cP$, although the converse does not hold. Note also that, since $\sneg$ does not occur in~$\cQ$, this is also a CG program.
From Theorem~\ref{theorem:tp.properties},
$\lfp(\tpP{\cQ})=\tprP{\cQ}{\omega}$ is the least model of $\cQ$.
By renaming back $l'$ and $l''$ as $\sneg l$ and $\sneg\sneg l$ in
$\tprP{\cQ}{k}$
we obtain
$\etpr{k}$ for any $k$.
Hence, $\lfp(\etp)=\etpr{\omega}$ is the least model of $\cP$.
Statement (ii) is proved in the same manner.\qed
\end{proofof}

\newpage

\subsection{Proof of Proposition~\ref{prop:gamma.antimonotonic}}

\begin{lemma}\label{lem:least.model.monotonic}
Let $P_1$ and $P_2$ be two programs and let \review{R1.18}{$\eU_1$ and $\eU_2$} be two interpretations such that $P_1 \supseteq P_2$ and $\eU_1 \leq \eU_2$.
Let also $\eI_1$ and $\eI_2$ be the least  models of~$\ereduct{\cP_1}{\eU_1}$~and~$\ereduct{\cP_2}{\eU_2}$, respectively. Then $\eI_1 \geq \eI_2$.\qed
\end{lemma}

\begin{proof}
First, for any rule $r_i$ and pair of \review{R1.19, R3.17}{}interpretations $\eJ_1$ and
$\eJ_2$ such that $\eJ_1\geq \eJ_2$,
\begin{gather*}
\eJ_1(body^+(r_i^{U_1})) \ \ \geq \ \ \eJ_2(body^+(r_i^{U_2}))
\end{gather*}
Furthermore, since $U_1\leq U_2$, by Proposition~\ref{prop:term.neg.leq}, it follows
\begin{gather*}
\eU_1(body^-(r_i^{U_1}))  \ \ \geq \ \ \eU_2(body^-(r_i^{U_2}))
\end{gather*}
and, since by Definition~\ref{def:e.model}
$\eJ_j(body^-(r_i^{U_1})) \ \ \eqdef \ \ \eU_j(body^-(r_i^{U_1}))$, it follows that
\begin{gather*}
\eJ_1(body^-(r_i^{U_1})) \ \ \geq \ \ \eJ_2(body^-(r_i^{U_2}))
\end{gather*}
Hence, we obtain that $\eJ_1(body(r_i^{U_1}))\geq \eJ_2(body(r_i^{U_2}))$.

Since $\cP_1 \supseteq P_2$, it follows that every rule $r_i\in\cP_2$ is in $\cP_1$ as well.
Thus, 
\review{R3.17}{$\etpP{\ereduct{\cP_1}{\eU_1}}(\eJ_1)(\rH) \geq \etpP{\ereduct{\cP_2}{\eU_2}}(\eJ_2)(\rH)$}
for every atom~$\rH$. Furthermore, since
\begin{gather*}
\etprP{\ereduct{\cP_1}{\eU_1}}{0}(\rH)
    \ \ = \ \  \etprP{\ereduct{\cP_2}{\eU_2}}{0}(\rH)
    \ \ = \ \ 0
\end{gather*}
it follows
\review{R1.20}{\mbox{$\etprP{\ereduct{\cP_1}{\eU_1}}{i}(\rH) \ \geq \ \etprP{\ereduct{\cP_2}{\eU_2}}{i}(\rH)$}}
for all $0\leq i$. Finally,
\begin{gather*}
\etprP{\ereduct{\cP_j}{\eU_j}}{\omega}(\rH)
    \ \ \eqdef \ \ \sum_{i\leq\omega}\etprP{\ereduct{\cP_j}{\eU_j}}{i}(\rH)
    \ \ = \ \ 0
\end{gather*}
and hence \review{R1.21, R3.17}{\mbox{$\etprP{\ereduct{\cP_1}{\eU_1}}{\omega}(\rH) \ \geq \ \etprP{\ereduct{\cP_2}{\eU_2}}{\omega}(\rH)$}}. By Theorem~\ref{thm:tp.properties}, these are respectively the least models of $\ereduct{\cP_1}{\eU_1}$ and $\ereduct{\cP_2}{\eU_2}$.
That is $\eI_1 \geq \eI_2$.\qed
\end{proof}

\begin{proposition}\label{prop:gamma.antimonotonic}
$\eWp$ operator is anti-monotonic and operator $\eWp^2$ is monotonic. That is, $\eWp(\eU_1) \geq \eWp(\eU_2)$ and \review{R1.22}{$\eWp^2(\eU_1) \leq \eWp^2(\eU_2)$} for any pair of interpretations $\eU_1$ and~$\eU_2$ such that $\eU_1 \leq \eU_2$.\qed
\end{proposition}
\begin{proof}
Since $\eU_1\leq\eU_2$, by Lemma~\ref{lem:least.model.monotonic}, it follows $\eI_1\geq \eI_2$ with $\eI_1$ and $\eI_2$ being respectively the least models of $\ereduct{\cP}{\eU_1}$ and $\ereduct{\cP}{\eU_2}$.
Then, \mbox{$\eWp(\eU_1)=\eI_1$} and $\eWp(\eU_2)=\eI_2$ and, thus, $\eWp(\eU_1) \geq \eWp(\eU_2)$. Since $\eWp$ is anti-monotonic it follows that $\eWp^2$ is monotonic. 
%See~\cite{khamsi1997fixed} or elsewhere.
\qed
\end{proof}

\subsection{Proof of Theorem~\ref{thm:wellf.justification<-$>$weff.standard}}

\review{R3.11}{The proof of Theorem~\ref{thm:wellf.justification<-$>$weff.standard} will rely on the relation between ECJ justifications and non-hypothetical WnP justifications established by Theorem~\ref{thm:wellf.justification<-$>$weff.provenace} and it can be found below the proof of that theorem in page~\pageref{proof:thm:wellf.justification<-$>$weff.standard}.}

\newpage

\subsection{Proof of Theorem~\ref{thm:wellf.inhibited.justification-$>$sm.cause}}

\begin{definition}
A term $t \in \evalues$ is \emph{join irreducible} iff $t = \sum_{u \in U} u$ implies that $u = t$ for some $u \in U$ and it is \emph{join prime} iff $t \leq \sum_{u \in U} u $ implies that $u \leq t$ for some $u \in U$.\qed
\end{definition}

\review{R1.12}{\begin{proposition}\label{prop:join.irreducible}
The following results hold:
\begin{enumerate}
\item A term is join irreducible iff is join prime.
\item If $\lb$ is finite, then every term $t$ can be represented as a unique finite sum of pairwise incomparable join irreducible terms.\qed
\end{enumerate}
\end{proposition}
}

\begin{proof}
The first result directly follows from Theorem 1 in~\cite[page~65]{balbes1975distributive}.
Furthermore, from Theorem~2 in~\cite[page~66]{balbes1975distributive}, in every distributive lattice satisfying the descending chain condition, any element can be represented as a unique finite sum of pairwise incomparable join irreducible elements and it is clear that every finite lattice satisfies the descending chain condition.\qed
\end{proof}

% \begin{definition}
% We said that a term is in \emph{join irreducible normal form} (JNF) iff it is a sum of pairwise incomparable join irreducible elements in NNF.\qed
% \end{definition}

\begin{lemma}\label{lem:eWp.label.removing-rules}
Let $P$ be a positive program over a signature~$\signature$ where $\lb$ is a finite set of labels and $Q$ be the result of removing all rules labelled by some label~$l \in \lb$.
Let $I$ and $J$ be two interpretations such that $J$ 
such that $\varrho_{\sneg l}(I) \geq J$.
Then, $\varrho_{\sneg l}(\eWp(I)) \leq \eWpP{\cQ}(J)$.\qed
\end{lemma}

\begin{proof}
% {Proposition~\ref{prop:eWp.label.removing-rules}}
By definition $\eWp(\eI)$ and $\eWpP{\cQ}(\eJ)$ are the least models of programs
$\ereduct{\cP}{\eI}$ and $\ereduct{\cQ}{\eJ}$, respectively.
Furthermore, from Theorem~\ref{thm:tp.properties}, the least model of any program $P$ is the least fixpoint of the $T_P$ operator, that is, $\eWpP{X}(Y) = \etprP{X^Y}{\omega}$ with $X \in \set{P,Q}$ and $X^Y\in\set{\cP^\eI,\cP^\eJ}$.
Then, the proof follows by induction assuming that 
$u \leq \etprP{Q^\eJ}{\beta}(\rH)$
implies
$\varrho_{\sneg l}(u) \leq \etprP{\cQ^\eI}{\beta}(\rH)$ for any join irreducible $u$, atom~$\rH$ and every ordinal $\beta < \alpha$.

Note that $\etprP{Q^\eJ}{0}(\rH)= 0 = \varrho_{\sneg l}(0) = \etprP{P^\eI}{0}(\rA)$ for any atom~$\rH$ \review{R1.11}{and, thus, the statement holds vacuous.}

If $\alpha$ is a successor ordinal, since $u \leq\etprP{P^\eI}{\alpha}(\rH)$, there is a rule in $P$ of the form~\eqref{eq:rule} such that
\begin{gather*}
u\ \ \leq \ \ (u_{B_1}*\dotsc*u_{B_m}*u_{C_1}*\dotsc*u_{C_n}) \cdot r_i
\end{gather*}
where 
\mbox{$u_{B_j}\leq\etprP{P^\eI}{\alpha-1}(B_j)$} and \mbox{$u_{C_j}\leq\sneg\eI(C_j)$} for each positive literal $B_j$ and each negative literal $\Not \ C_j$ in the body of rule $r_i$.
Then,\begin{enumerate}
\item By induction hypothesis, it follows that $\varrho_{\sneg l}(u_{B_j}) \leq\etprP{Q^J}{\alpha-1}(\rB_j)$, and

\item from $\varrho_{\sneg l}(I(\rH)) \geq J(\rH)$, it follows that $u_{C_j}\leq\sneg\eI(C_j)$ implies $\varrho_{\sneg l}(u_{C_j})\leq \sneg\eJ(\rC_j)$.
\end{enumerate}
Furthermore, if $r_i \neq l$, then $r_i \in Q$ and, thus,
\begin{gather*}
\varrho_{\sneg l}(u)
    \ \ \leq \ \ 
     (\varrho_{\sneg l}(u_{B_1})*\dotsc*\varrho_{\sneg l}(u_{B_m})*\varrho_{\sneg l}(u_{C_1})*\dotsc*\varrho_{\sneg l}(u_{C_n})) \cdot r_i
     \ \ \leq \ \ \etprP{Q^J}{\alpha}(\rH)
\end{gather*}
If otherwise $r_i = l$, then $\varrho_{\sneg l}(u)=0 \leq\etprP{Q^J}{\alpha}(\rH)$.

In case that $\alpha$ is a limit ordinal,
$u \leq \etprP{\cP^\eI}{\alpha}$ iff $u \leq \etprP{\cP^\eI}{\beta}$ for some $\beta < \alpha$ and any join irreducible $u$.
Hence, by induction hypothesis, it follows that
$\varrho_{\sneg l}(u) \leq \etprP{\cQ^\eJ}{\beta} \leq \etprP{\cQ^\eJ}{\alpha}$
and, thus, 
$\varrho_{\sneg l}(\etprP{\ereduct{\cP}{\eI}}{\alpha}) \leq \etprP{\ereduct{\cQ}{\eJ}}{\alpha}$.
\qed
\end{proof}

\begin{proofof}
{Theorem~\ref{thm:wellf.inhibited.justification-$>$sm.cause}}
In the sake of simplicity, we just write $\varrho$ instead of $\varrho_{\sneg r_i}$.
Note that, by definition, for any atom~$\rH$, it follows that $\ewfmP{X}(\rH) = \elfpP{X}(\rH)$ with $X \in\set{\cP,\cQ}$.
The proof follows by induction in the number of steps of the $\Gamma^2$ operator assuming as induction hypothesis that 
$\eWprP{Q}{\beta} \leq \varrho(\eWprP{\cP}{\beta})$ for every $\beta < \alpha$.
Note that \mbox{$\eWprP{Q}{0}(\rH)=0\leq\varrho(\eWprP{\cP}{0})(\rH)$} and, thus, the statement trivially holds for $\alpha=0$ .

In case that $\alpha$ is a successor ordinal, by induction hypothesis, it follows that
\begin{gather*}
\eWprP{Q}{\alpha-1}\ \ \leq \ \ \varrho(\eWprP{\cP}{\alpha-1})
\end{gather*}
and, from Lemma~\ref{lem:eWp.label.removing-rules},
it follows that
\begin{IEEEeqnarray*}{lCl}
\eWpP{\cQ}(\eWprP{Q}{\alpha-1}) &\ \ \geq \ \ &\varrho(\eWpP{P}(\eWprP{\cP}{\alpha-1}))
\\
\eWpP{Q}^2(\eWprP{Q}{\alpha-1}(\rH))) &\ \ \leq \ \ &\varrho(\eWpP{P}^2(\eWprP{\cP}{\alpha-1}))
\end{IEEEeqnarray*}
That is, $\eWprP{Q}{\alpha} \ \ \leq \ \ \varrho(\eWprP{\cP}{\alpha}$.

Finally, in case that $\alpha$ is a limit ordinal,
every join irreducible $u$ satisfies
$u \leq \eWprP{Q}{\alpha}=\sum_{\beta<\alpha}\eWprP{Q}{\beta}$
iff
$u \leq \eWprP{Q}{\beta}$
for some $\beta < \alpha$ and, thus, by induction hypothesis
$\varrho(u) \leq \eWprP{P}{\beta} \leq \eWprP{P}{\alpha}$.
Consequently, 
$\eWprP{Q}{\infty}  \leq  \varrho(\eWprP{\cP}{\infty}$ and
$\ewfmP{Q}(\rA)  \leq  \varrho(\ewfmP{\cP}(\rA)$ for any atom~$\rA$.
\qed
\end{proofof}

\subsection{Proof of Theorem~\ref{thm:causal-prov.smodels.are.the.causal-smodels}}

By $\Wp(\cI)$ we denote the least model of a program $\cP^\cI$.
Note that the relation between $\Wp$ and $\eWp$ is similar to the relation between $\tp$ and $\etp$: the $\Wp$ operator is a function in the set of CG interpretations while $\eWp$ is a function in the set of ECJ interpretations.
Note also that the evaluation of negated literals with respect to CG and ECJ interpretations and, thus, the reducts $\cP^\cI$ and $\cP^\eI$ may be different even if $\cI(\rA) = \eI(\rA)$ for every atom~$\rA$.

% \subsection{Proof of Proposition~\ref{prop:aux:thm:causal-prov.smodels.are.the.causal-smodels}}

\begin{lemma}\label{lem:gamma.causal<-explanations}
Let $\cP$ be a labelled logic program, $\cI$ and $\eJ$ be respectively an CG and a ECJ interpretation such that $\cI\geq\lambdac(\eJ)$.
Then $\Wp(\cI) \leq \lambdac(\,\eWp(\eJ))$.\qed
\end{lemma}

\begin{proof}
By definition $\Wp(\cI)$ and $\eWpP{\cP}(\eJ)$ are respectively the least model of the programs
$\ereduct{\cP}{\cI}$ and $\ereduct{\cP}{\eJ}$.
Furthermore, from Theorem~\ref{thm:tp.properties} the least model of any program $P$ is the least fixpoint of the $T_P$ operator,
that is,
$\Wp(\cI)=\tprP{\ereduct{\cP}{\cI}}{\omega}$
and
$\eWp(\eJ)=\etprP{\ereduct{\cP}{\eJ}}{\omega}$.
In case that $\alpha=0$, it follows that 
$\tprP{\cP^\cI}{0}(\rH) \ = \ 0 \ \leq \ \lambdac(\etprP{\ereduct{\cP}{\eJ}}{0})(\rH)$ for every atom $\rH$.
We assume as induction hypothesis that
$\tprP{\ereduct{\cP}{\cI}}{\beta} \leq \lambdac(\etprP{\ereduct{\cP}{\eJ}}{\beta})$ for all $\beta<\alpha$.

In case that $\alpha$ is a successor ordinal,
$E\leq\tprP{\ereduct{\cP}{\cI}}{\alpha}(\rH)=\tpP{\ereduct{\cP}{\cI}}(\tprP{\ereduct{\cP}{\cI}}{\alpha-1})(\rH)$ if and only if there is a rule $\rR^I$ in $\cP^\cI$ of the form
\begin{gather*}
r_i : \rH \leftarrow \rB_1, \dotsc, \rB_m,
\end{gather*}
which is the reduct of a rule~$\rR$ of the form \eqref{eq:rule} in $\cP$ and that satisfies $E \leq ( E_{B_1} * \dotsc * E_{B_m} ) \cdot r_i$ with each
$E_{B_j} \leq \tprP{\ereduct{\cP}{\cI}}{\alpha-1}(B_j)$ and $\cI(C_j)=0$ for all $B_j$ and $C_j$ in $body(\rR)$.
Hence there is a rule in $\ereduct{\cP}{\eJ}$ of the form 
\begin{gather*}
r_i : \rH \leftarrow \rB_1, \dotsc, \rB_m, \
    \eJ(\Not \rC_1), \dotsc, \ \eJ(\Not \rC_n)
\end{gather*}
and, by induction hypothesis,
$E_{B_j} \leq \lambdac\big( \etprP{\ereduct{\cP}{\eJ}}{\alpha-1}(\rB_j)\big)$ for all $\rB_j$.
Furthermore, by definition
\begin{gather*}
\big(\etprP{\ereduct{\cP}{\eJ}}{\alpha-1}(\rB_1) * \dotsc * \etprP{\ereduct{\cP}{\eJ}}{\alpha-1}(\rB_m) * \eJ(\Not C_1) * \dotsc * \eJ(\Not C_m) \big) \cdot r_i \ \ \leq \ \ \etprP{\ereduct{\cP}{\eJ}}{\alpha}(\rH)
\end{gather*}
From the fact that $\cI(\rC_j) = 0$ and the lemma's hypothesis $\cI\geq\lambdac(\eJ)$, it follows that $0 \geq \lambdac(\eJ(\rC_j))$ and, thus, $1 \leq \lambdac(\sneg\eJ(\rC_j)) = \lambdac(\eJ(\Not \rC_j))$.
Hence,
\begin{align*}
\lambdac&\big(( \etprP{\ereduct{\cP}{\eJ}}{\alpha-1}(B_1) * \dotsc * \etprP{\ereduct{\cP}{\eJ}}{\alpha-1}(B_m) * \eJ(\Not C_1) * \dotsc * \eJ(\Not C_m)) \cdot r_i \big) \ \ =
\\
&= \ \ 
\lambdac\big(( \etprP{\ereduct{\cP}{\eJ}}{\alpha-1}(B_1) * \dotsc * \etprP{\ereduct{\cP}{\eJ}}{\alpha-1}(B_m) \big) *
\lambdac\big(\eJ(\Not C_1)\big) * \dotsc *\lambdac\big(\eJ(\Not C_m)\big) \big) \cdot r_i
\\
&= \ \ 
\lambdac\big(( \etprP{\ereduct{\cP}{\eJ}}{\alpha-1}(B_1) * \dotsc * \etprP{\ereduct{\cP}{\eJ}}{\alpha-1}(B_m) \big) *
1 * \dotsc * 1 \big) \cdot r_i
\\
&= \ \ 
\lambdac\big(( \etprP{\ereduct{\cP}{\eJ}}{\alpha-1}(B_1) * \dotsc * \etprP{\ereduct{\cP}{\eJ}}{\alpha-1}(B_m) \big) \big) \cdot r_i
\end{align*}
and, thus,
\begin{align*}
\lambdac\big(( \etprP{\ereduct{\cP}{\eJ}}{\alpha-1}(B_1) * \dotsc * \etprP{\ereduct{\cP}{\eJ}}{\alpha-1}(B_m) \big) \big) \cdot r_i
\ \ \leq \ \ \lambdac\big( \etprP{\ereduct{\cP}{\eJ}}{\alpha}(\rH) \big)
\end{align*}
Since $E_{B_j} \leq \lambdac\big( \etprP{\ereduct{\cP}{\eJ}}{\alpha-1}(\rB_j)\big)$ for all $\rB_j$, it follows that
\begin{gather*}
E \ \ \leq \ \ ( E_{B_1} * \dotsc * E_{B_m} ) \cdot r_i \ \ \leq \ \ \lambdac(\etprP{\ereduct{\cP}{\eJ}}{\alpha})(\rH)
\end{gather*}

Finally, in case that $\alpha$ is a limit ordinal, it follows from Theorem~\ref{thm:tp.properties} that $\alpha=\omega$.
Furthermore, since $\cI$ is a CG interpretation, it follows that $P^\cI$ is a CG program and, thus, $E \leq \etprP{\cP^\cI}{\omega}$ iff $E \leq \etprP{\cP^\cI}{n}$ for some $n < \omega$ (see~\citeNP{CabalarFF14}).
Hence, by induction hypothesis, it follows that $E \leq \etprP{\cP^\eJ}{n} \leq \etprP{\cP^\eJ}{\omega}$.\qed
\end{proof}

\begin{lemma}\label{lem:gamma.causal-$>$explanations}
Let $\cP$ be a labelled logic program over a signature~$\signature$ where $\lb$ is a finite set of labels, $\cI$ and $\eJ$ respectively be a CG and a ECJ interpretation such that
$\cI \leq \lambdac(\eJ)$. Then $\Wp(\cI) \geq \lambdac(\eWp(\eJ))$.\qed
\end{lemma}

\begin{proof}
Since $\lb$ is finite, it follows that $\evalues$ is also finite.
Furthermore, since $\evalues$ is a finite distributive lattice,
every element $t \in \evalues$ can be represented as a unique sum of join irreducible elements (Proposition~\ref{prop:join.irreducible}).

Assume as induction hypothesis that 
$u \leq \etprP{\ereduct{\cP}{\eJ}}{\beta}(\rH)$
implies
$\lambdac(u) \leq \tprP{\ereduct{\cP}{\cI}}{\beta}(\rH)$
for every
join irreducible $u$, atom $\rH \in \at$ and ordinal $\beta < \alpha$.

In case that $\alpha$ is a successor ordinal.
For any join irreducible justification $u\leq\etprP{\ereduct{\cP}{\eJ}}{\alpha}(\rH)$
there is a rule $\rR^\eJ$ in $\cP^\eJ$ of the form \eqref{eq:e.rule} and there are join irreducible terms
$u_{B_j} \leq \etprP{\ereduct{\cP}{\eJ}}{\alpha-1}(B_j)$
and
\mbox{$u_{C_j} \leq \sneg\eJ(C_j)$} for all $B_j$ and $C_j$ such that
\begin{gather*}
u \ \ \leq \ \ ( u_{B_1} * \dotsc * u_{B_m} * u_{C_1} * \dotsc * u_{C_n} ) \cdot r_i
\end{gather*}
% Suppose that there is some $u_{B_j}$ with $B_j \in \set{B_1, \dotsc B_m}$ which is not join irreducible.
% Then $u_{B_j} = w_1 + w_2$ for some $w_1$ and $w_2$ and, thus,
% \begin{align*}
% ( u_{B_1} &* \dotsc * u_{B_j} * \dotsc  * u_{B_m} * u_{C_1} * \dotsc * u_{C_n} ) \cdot r_
% \\
%     \ \ &=\ \
%     ( u_{B_1} * \dotsc * (w_1 + w_2) * \dotsc  * u_{B_m} * u_{C_1} * \dotsc * u_{C_n} ) \cdot r_i
% \\
%     \ \ &=\ \
%     ( u_{B_1} * \dotsc * w_1  * \dotsc  * u_{B_m} * u_{C_1} * \dotsc * u_{C_n} ) \cdot r_i
%     +
%     ( u_{B_1} * \dotsc * w_2 * \dotsc  * u_{B_m} * u_{C_1} * \dotsc * u_{C_n} ) \cdot r_i
% \end{align*}

If $u_{C_j}$ contains an oddly negated label for some $C_j$, then $\lambdac(u_{C_j})=0$ and it consequently follows that \mbox{$\lambdac(u)=0 \ \leq \ \tprP{\cP^\cI}{\alpha}(\rH)$}. Thus, we assume that $u_{C_j}$ only contains evenly negated labels for any $C_j$.
Note that, since $u_{C_j} \leq \sneg \eJ(C_j)$, then $u_{C_j}$ cannot contain any non-negated label, that is, all occurrences of labels in $u_{C_j}$ are strictly evenly negated and, thus, every term $u'_{C_j} \leq \eJ(C_j)$ must contain some oddly negated label.
Hence, $\cI(C_j) \leq \lambdac(\eJ(C_j)) = 0$ for any $C_j$ and there is a rule
$\rR^\cI$ in $\cQ^\cI$ of the form
\begin{gather*}
r_i : \ \ \rH \leftarrow  B_1, \dotsc, B_m
\end{gather*}
By induction hypothesis,
$u_{B_j} \leq \etprP{\ereduct{\cP}{\eJ}}{\alpha-1}(B_j)$
implies $\lambdac(u_{B_j}) \leq \tprP{\ereduct{\cP}{\cI}}{\alpha-1}(B_j)$
and, consequently, $\lambdac(u) \leq \tprP{\ereduct{\cP}{\cI}}{\alpha}(\rH)$.

Since $\etprP{\ereduct{\cP}{\eJ}}{\alpha}(\rH) = \sum_{ u \in U_\rH} u$ where every $u \in U_\rH$ is join irreducible and every $u \in U_\rH$ satisfies
$u \leq \etprP{\ereduct{\cP}{\eJ}}{\alpha}(\rH)$, it follows that
$\lambdac(u) \leq \tprP{\ereduct{\cP}{\cI}}{\alpha}(\rH)$
and, thus,
$\sum_{u \in U_\rH} \lambdac(u) \leq \tprP{\ereduct{\cP}{\cI}}{\alpha}(\rH)$.
Note that, by definition,
$\lambdac(\sum_{u \in U_\rH} u) = \sum_{u \in U_\rH} \lambdac(u)$
and, thus,
\begin{gather*}
\lambdac(\etprP{\ereduct{\cP}{\eJ}}{\alpha}(\rH))
  \ \ = \ \ \lambdac(\sum_{u \in U_\rH} u) 
  \ \ \leq \ \ \tprP{\ereduct{\cP}{\cI}}{\alpha}(\rH)
\end{gather*}
In case that $\alpha$ is a limit ordinal, it follows 
$u \leq \etprP{\cP^\eJ}{\alpha}(\rH)$ iff $u \leq \etprP{\cP^\eJ}{\beta}(\rH)$ for some $\beta < \omega$ and,
by induction hypothesis, it follows that
$\lambdac(u) \leq \tprP{\cP^\cI}{\beta}(\rH) \leq \tprP{\cP^\cI}{\alpha}(\rH)$
and, thus, 
$\tprP{\ereduct{\cP}{\cI}}{\alpha} \geq \lambdac(\etprP{\ereduct{\cP}{\eJ}}{\alpha})$.

Finally, by definition $\Wp(\cI)$ and $\eWp(\eJ)$
are respectively the least models of $P^\cI$ and $P^\eJ$ and, from Theorem~\ref{theorem:tp.properties}, these are precisely 
$\tprP{\ereduct{\cP}{\cI}}{\omega}$
and
$\etprP{\ereduct{\cP}{\eJ}}{\omega}$.
Hence, 
$\tprP{\ereduct{\cP}{\cI}}{\omega} \geq \lambdac(\etprP{\ereduct{\cP}{\eJ}}{\omega})$
implies
$\Wp(\cI) \geq \lambdac(\eWp(\eJ))$.
\qed
\end{proof}

\begin{proposition}\label{prop:aux:thm:causal-prov.smodels.are.the.causal-smodels}
Given a program $P$ over a signature~$\signature$ where $\lb$ is a finite set of labels,
any ECJ interpretation $\eI$ satisfies $\Wp(\lambdac(\eI))=\lambdac(\eWp(\eI)))$.\qed
% the following facts hold: 
% {\renewcommand{\theenumi}{(\roman{enumi})}%
% \begin{enumerate}
% \item $\Wp(\cI)=\lambdac(\eWp(\eI))$ \ for any CG interpretation $\cI$ where $\eI$ is a ECJ interpretation such that $\eI(\rH) = \cI(\rH)$ for every atom $\rH \in \at$, and
% \item $\Wp(\lambdac(\eI))=\lambdac(\eWp(\eI)))$ \ for any ECJ interpretation $\eI$.\qed
% \end{enumerate}}
\end{proposition}

\begin{proofof}
{Proposition~\ref{prop:aux:thm:causal-prov.smodels.are.the.causal-smodels}}
% (i) 
Let $\cI$ be a CG interpretation such that $I(\rH)=\cI(\rH)$ for every atom~$\rH$.
Then, it follows that $\cI = \lambdac(\eI)$.
Hence, from Lemmas~\ref{lem:gamma.causal<-explanations} and~\ref{lem:gamma.causal-$>$explanations}, it respectively follows that
$\Wp(\cI)\leq \lambdac(\,\eWp(\eI))$ and
$\Wp(\cI)\geq \lambdac(\,\eWp(\eI))$.
Then, $\Wp(\cI) = \Wp(\lambdac(\eI)) = \lambdac(\,\eWp(\eI))$.
% (ii) As above, any ECJ interpretation $\eI$ satisfies
% $\Wp(\lambdac(\eI)) = \lambdac(\,\eWp(\eI))$
% and, consequently, we have that
% $\Wp(\lambdac(\lambdac(\eI))) = \lambdac(\Wp(\lambdac(\eI)))$.
% Since $\lambdac$ is clearly idempotent $\lambdac(\lambdac(\eI)) = \lambdac(\eI)$ it follows that
% $\Wp(\lambdac(\eI)) = \lambdac(\Wp(\lambdac(\eI))$.
% From this last equivalence plus
% $\Wp(\lambdac(\eI)) = \lambdac(\eWp(\eI))$.
% % it follows that
% % $\lambdac(\eWp(\eI)) = \lambdac(\Wp(\lambdac(\eI)))$.
% Finally, just reversing the equality, it follows that $\Wp(\lambdac(\eI)) = \lambdac(\eWp(\eI))$.
\qed
\end{proofof}

\begin{proofof}
{Theorem~\ref{thm:causal-prov.smodels.are.the.causal-smodels}}
According to~\cite{CabalarFF14}, a CG interpretation~$\cI$ is a CG stable model of $\cP$ iff $\cI$ is the least model of the program $\cP^\cI$. Then, the CG stable models are just the fixpoints of the $\Wp$ operator.

Let $\cI$ be a CG stable model according to~\cite{CabalarFF14},
let $\eI$ be a ECJ interpretation such that
$\eI(\rH) = \cI(\rH)$ for every atom $\rH \in \at$ and
 let $\eJ\eqdef\eWprI{\infty}{\eI}$ be the least fixpoint of $\eWp^2$ iterating from~$\eI$.
Since $\eI(\rH) = \cI(\rH)$ for every atom $\rH \in \at$, it follows that $\cI = \lambdac(\eI)$
and, by definition of CG stable model, it follows that $\cI = \Wp(\cI)$.
Thus, from Proposition~\ref{prop:aux:thm:causal-prov.smodels.are.the.causal-smodels}, it follows that $\cI=\lambdac(\eWp(\eI))$.
Applying $\Wp$ to both sides of this equality, we obtain that
$\Wp(\cI)=\Wp(\lambdac(\eWp(\eI)))$.
From Proposition~\ref{prop:aux:thm:causal-prov.smodels.are.the.causal-smodels} again,
it follows that
$\Wp(\lambdac(\eWp(\eI))) = \lambdac(\eWp(\eWp(\eI))) = \lambdac(\eWp^2(\eI))$
and, thus, 
$\Wp(\cI)=\lambdac(\eWp^2(\eI))$.
Furthermore, since $\cI = \Wp(\cI)$,
it follows that $\cI=\lambdac(\eWp^2(\eI))$.
Inductively applying this argument, it follows that
\mbox{$\cI=\lambdac(\eWprI{\alpha}{\eI})$}
for any successor ordinal $\alpha$.
Moreover, for a limit ordinal $\alpha$,
\begin{gather*}
\lambdac\big(\, \eWprI{\alpha}{\eI} \,\big)
  \ = \ \lambdac\Big(\, \sum_{\beta < \alpha} \eWprI{\beta}{\eI} \,\Big)
  \ = \ \sum_{\beta < \alpha} \lambdac\big(\eWprI{\beta}{\eI} \big)
  \ = \ \cI
\end{gather*}
Then, since we have defined $\eJ=\eWprI{\infty}{\eI}$, it follows that
$\cI = \lambdac(\eJ) = \lambdac(\eI)$
and, since we also have that $\cI=\lambdac(\eWp(\eI))$, we obtain that
$\lambdac(\eI)=\lambdac(\eWp(\eI))$.

The other way around. Let $\eI$ be a fixpoint of $\eWp^2$ such that $\lambdac(\eI)=\lambdac(\eWp(\eI))$ and let $\cI\eqdef\lambdac(\eI)$.
In the same way as above, it follows that $\Wp(\cI)=\lambdac(\eWp(\eI))=\lambdac(\eI)=\cI$. That is, $\Wp(\cI)=\cI$ and so that $\cI$ is a causal stable model of $\cP$ according to~\cite{CabalarFF14}.\qed
\end{proofof}

\subsection{Proof of Theorem~\ref{thm:wellf.justification-$>$sm.cause}}

\begin{proofof}
{Theorem~\ref{thm:wellf.justification-$>$sm.cause} }
Let $\cI$ be a causal stable model of $\cP$ and $\eI$ be the correspondent fixpoint of $\eWp^2$ with $\cI = \lambdac(\eI)$.
Since $E$ is a enabled justification of $\rA$, i.e. $E\leq\ewfm(\rA)$, then $E\leq\elfp(\rA)$ with $\elfp$ the least fixpoint of $\eWp^2$.
Since, $\eI$ is a fixpoint of $\eWp^2$,
if follows that $E \leq \elfp(\rA) \leq \eI(\rA)$
and, thus, $\lambdac(E) \leq \lambdac(\eI(\rA)) = \cI(\rA)$.
Then \mbox{$G\eqdef graph(\lambdac(E))$} is, by definition, a causal explanation of the atom~$\rA$.
\end{proofof}

\subsection{Proof of Theorem~\ref{thm:prov.correspondence}}

The proof of \review{R1.13, R3.12}{Theorem~\ref{thm:prov.correspondence}} will need the following definition.

\begin{definition}\label{def:WnP-interpretation}
Given a program $\cP$, a \emph{WnP interpretation} is a mapping \mbox{$\wI:At \longrightarrow \boolAlgebra$} assigning a Boolean formula to each atom. The evaluation of a negated literal $\Not \rA$ \review{R1.15}{with respect to} a WnP interpretation is given by $\wI(\Not \rA) = \neg \wI(\rA)$.
An interpretation $\wI$ is a WnP model of rule like \eqref{eq:rule}  iff
\begin{IEEEeqnarray*}{c+x*}
\wI(\rB_1) * \dotsc * \wI(\rB_m)
    * \wI(\Not \rC_1) * \dotsc * \wI(\Not \rC_n) 
        * r_i \ \leq \ \wI(\rH) &
     % \label{eq:model}
\end{IEEEeqnarray*}
The operator $\wWpP{P}(\wI)$ maps a WnP interpretation~$\wI$ to the least model of the program $\cP^\wI$.\qed
\end{definition}

Note that the only differences in the model evaluation between ECJ and WnP comes from the valuation of negative literals and the use of `$*$' instead of `$\cdot$' for keeping track of rule application.
Besides, we will also use the following facts whose proof is addressed in an appendix.

\begin{definition}
Given a positive program $\cP$, we define a direct consequence operator~$\wtpP{\cP}$ such that
\begin{align*}
\wtpP{\cP}(\wI)(\rH)
    \ \ &\eqdef \ \
      \sum \big\{ \ \wI(B_1) * \dotsc * \wI(B_n) * r_i \
% \\&\hspace{2.5cm}
        \mid \ (r_i : \ \rH \leftarrow B_1, \dotsc, B_n ) \in P \ \big\}
\end{align*}
for any WnP interpretation $\wI$ and atom $\rH \in \at$.\qed
\end{definition}

\begin{definition}[From \citeNP{damasio2013justifications}]\label{def:WnPsyntax}
Given a  program $\cP$, \review{R1.16}{its why-not program is given by} $\wwP\eqdef \cP \cup \cP'$ here $\cP'$ contains a labelled fact of the form 
\begin{gather*}
\neg not(\rA) : \rA
\end{gather*}
for each atom $\rA \in At$ not occurring in $\cP$ as a fact.
The why-not provenance information
under the well-founded semantics is defined as follows:
$Why_\wwP (\rH) = [\wlfpP{\wwP}(\rH)]$;
$Why_\wwP (\rH) = [\neg\wgfpP{\wwP}(\rH)]$; and
$Why_\wwP (\Undef A) = [\neg\wlfpP{\wwP}(\rH) \wedge \wgfpP{Q}(\rH)]$
where
$\wlfpP{\wwP}$ and $\wgfpP{\wwP}=\wWpP{\wwP}(\wlfpP{\wwP})$ be the least and greates fixpoints of $\wWpP{\wwP}^2$, respectively.
\qed
\end{definition}

\begin{lemma}\label{lem:gamma.causal-$>$provenance}
Let $P$ be a labelled logic program over a signature~$\signature$ where $\lb$ is a finite set of labels and let $\eI$ and $\wI$ be respectively a ECJ and a WnP interpretation such that
$\lambdap(\eI)\geq \wI$.
Then,
$\lambdap(\eWpP{\wP}(\eI))\leq \wWpP{\wwP}(\wI)$.
\end{lemma}

\begin{proof}
By definition $\eWpP{\wP}(\eI)$ and $\wWpP{\wwP}(\wI)$ are the least model of the programs
$\ereduct{\wP}{\eI}$ and $\ereduct{\wwP}{\wI}$, respectively.
Furthermore, the least model of programs $\wP^\eI$ and $\wwP^\wI$ are the least fixpoint of the $T_{\wP^\eI}$ and $\wtpP{\wwP^\wJ}$ operators,
that is,
$\eWpP{\wP}(\eI)=\etprP{\ereduct{\wP}{\eI}}{\omega}$ 
and
$\wWpP{\wwP}(\eJ)=\wtprP{\ereduct{\wwP}{\wI}}{\omega}$.

In case that $\alpha=0$, it follows that 
$\lambdap(\etprP{\wP^\eI}{0}(\rH)) = \wtprP{\ereduct{\wwP}{\wI}}{0}(\rH) = 0$ for every atom $\rH$.
We assume as induction hypothesis that
$\lambdap(\etprP{\ereduct{\wP}{\eI}}{\beta}) \leq \wtprP{\ereduct{\wwP}{\wI}}{\beta}$
for all $\beta<\alpha$.

In case that $\alpha$ is a successor ordinal.
Assume that $u\leq\etprP{\ereduct{\wP}{\eI}}{\alpha-1}(\rH)$ for some join irreducible~$u$ and atom~$\rH$.
Then there is a rule $r_i\in\cP$ of the form \eqref{eq:rule}
% \begin{gather*}
% r_i : p \leftarrow B_1, \dotsc B_m,\ \Not C_1, \dotsc, \ \Not C_n
% \end{gather*}
and
\begin{gather*}
u \ \ \leq \ \ ( u_{B_1} * \dotsc * u_{B_1} * u_{C_1} * \dotsc * u_{C_1} ) \cdot r_i
\end{gather*}
where $u_{\rB_j}\leq\etprP{\ereduct{\wP}{\eI}}{\alpha-1}(\rB_j)$ and
$u_{\rC_j}\leq\sneg\eI(\rC_j)$.
Hence, by induction hypothesis, it follows that
$\lambdap(u_{\rB_j})\leq\wtprP{\ereduct{\wwP}{\wI}}{\alpha-1}(\rB_j)$ and, since \mbox{$u_{\rC_j}\leq\sneg\eI(\rC_j)$},
it also follows that \mbox{$\lambdap(u_{\rC_j})\leq\neg\wI(\rC_j)$} for all~$\rC_j$.
Consequently, we have that $\lambdap(u) \leq\wtprP{\wwP^\wI}{\alpha}(\rH)$.

In case that $\alpha$ is a limit ordinal,
$u \leq \etprP{\wP^\eI}{\alpha}$ iff $u \leq \etprP{\wP^\eI}{\beta}$ for some $\beta < \alpha$ and all join irreducible $u$.
Hence, by induction hypothesis, it follows that
$\lambdap(u) \leq \etprP{\wwP^\wJ}{\beta} \leq \etprP{\wwP^\wJ}{\alpha}$
and, thus, 
$\lambdap(\etprP{\ereduct{\wP}{\eI}}{\alpha}) \leq \wtprP{\ereduct{\wwP}{\wJ}}{\alpha}$.
\qed
\end{proof}

\begin{lemma}\label{lem:gamma.causal<-provenance}
Let $P$ be a labelled logic program over a signature~$\signature$ where $\lb$ is a finite set of labels and let $\eI$ and $\wI$ be respectively a ECJ and a WnP interpretation such that
$\lambdap(\eI)\leq\wI$.
Therefore,
$\lambdap(\eWpP{\wP}(\eI))\geq\wWpP{\wwP}(\wI)$.\qed
\end{lemma}

\begin{proof}
The proof is similar to the proof of Lemma~\ref{lem:gamma.causal-$>$provenance} and we just show the case in which $\alpha$ is a successor ordinal.

Assume that $u\leq\wtprP{\ereduct{\wwP}{\wI}}{\alpha}(\rH)$ for some join irreducible~$u$ and atom~$\rH$.
Hence, there is some rule $r_i\in\cP$ of the form~\eqref{eq:rule} and
\begin{gather*}
u \ \leq \ u_{\rB_1} * \dotsc * u_{\rB_m} * u_{\rC_1} * \dotsc * u_{\rC_n} * r_i
\end{gather*}
where $u_{\rB_j} \leq \wtprP{\ereduct{\wwP}{\wI}}{\alpha-1}(\rB_j)$ for each $\rB_j$ and \mbox{$u_{\rC_j} \leq \neg\wI(\rC_j)$} for each $\rC_j$.
By induction hypothesis,
$u_{\rB_j} \leq \lambdap(\etprP{\ereduct{\wP}{\eI}}{\alpha-1})(\rB_j)$
for all $\rB_j$.
Furthermore, since $\lambdap(\eI)\leq\wI$ it follows, from Lemma~\ref{lem:lambdaprov.neg.homomorphism.formula}, that $\lambdap(\sneg\eI)\geq\neg\wI$
and, since \mbox{$u_{\rC_j} \leq \neg\wI(\rC_j)$}, it also follows that
\mbox{$u_{\rC_j} \leq \lambdap(\sneg\eI(\rC_j))$}.
Hence,
\begin{gather*}
\lambda(u) \ \leq \ ( \lambdap(u_{B_1}) * \dotsc * \lambdap(u_{B_1}) * \lambdap(u_{C_1}) * \dotsc * \lambdap(u_{C_1}) ) * r_i
  \ \ \leq \ \ \lambdap(\etprP{\ereduct{\wP}{\eI}}{\alpha}(\rH))
\end{gather*}
Thus,
$\wtprP{\ereduct{\wwP}{\wI}}{\alpha}(\rB_j)\leq \lambdap(\etprP{\ereduct{\wP}{\eI}}{\alpha}(\rB_j))$.\qed
\end{proof}

Note that the image of $\lambdap$ is a boolean algebra and the set of causal values corresponding to negated terms $\setm{ \sneg t}{ t \in \evalues }$ are also a boolean algebra.
Consequently, \review{R1.16}{we define a function $\lambdaq(t) = \sneg\sneg t$}
% \begin{align*}
% \lambdaq(t) \ &\eqdef \ \begin{cases}
%     \lambdaq(u) \odot \lambdaq(w)
%         &\text{if } \lambdap(t)=u \odot v \text{ with } \odot\in\set{+,*}
%     \\
%     \sneg\lambdaq(u)
%         &\text{if } \lambdap(t)=\neg u
%     \\
%     \sneg\sneg l   &\text{if } \lambdap(t)= l \text{ with } l \in Lb
% \end{cases}
% \end{align*}
which is analogous to $\lambdap$ but whose image is in $\evalues$.

\begin{lemma}\label{lem:gamma.lmabdaq}
Let $P$ be a labelled logic program and let $\eI$ be an ECJ interpretation.
Then,
$\eWpP{\wP}(\eI) = \eWpP{\wP}(\lambdaq(\eI))$ and
$\lambdap(t)=\lambdap(\lambdaq(t))$.
\qed
\end{lemma}

\begin{proof}
For $\eWpP{\wP}(\eI) = \eWpP{\wP}(\lambdaq(\eI))$.
Since $\lambdaq(t) = \sneg\sneg t$ and
$\sneg\sneg\sneg t = \sneg t$, it follows that
$\lambdaq(\sneg I) = \sneg\sneg\sneg I = \sneg I$ and, thus,
$\wP^I = \wP^{\lambdaq(I)}$.
Since by definition
$\eWpP{\wP}(\eI)$ and $\eWpP{\wP}(\lambdaq(\eI))$ are respectively the least models of
programs $\wP^I$ and $\wP^{\lambdaq(I)}$ it is clear that
$\eWpP{\wP}(\eI) = \eWpP{\wP}(\lambdaq(\eI))$.

For $\lambdap(t)=\lambdap(\lambdaq(t))$, just note
$\lambdap(\lambdaq(t)) =  \lambdap(\sneg\sneg t)= \neg\neg \lambdap(t) = \lambdap(t)$.\qed
\end{proof}

\begin{proposition}\label{prop:aux:thm:prov.correspondence}
Let $P$ be a program over a signature~$\signature$ where $\lb$ is a finite set of labels.
Then, any causal \review{R1.14, R3.13}{interpretation}~$\eI$ satisfies: 
{\renewcommand{\theenumi}{(\roman{enumi})}%
\begin{enumerate}
\item $\wWpP{\wwP}(\lambdap(\eI)) \ = \ \lambdap(\eWpP{\wP}(\eI))$,

\item $\eWpP{\wP}(\lambdaq(\eI)) \ = \ \eWpP{\wP}(\eI)$ and

\item $\lambdap(t)=\lambdap(\lambdaq(t))$.\qed
\end{enumerate}}
\end{proposition}

\begin{proof}
% {Proposition~\ref{prop:aux:thm:prov.correspondence}}
(i) From Lemmas~\ref{lem:gamma.causal-$>$provenance} and~\ref{lem:gamma.causal<-provenance}, it respectively follows that
$\lambdap(\,\eWpP{\wP}(\eI)) \leq \wWpP{\wwP}(\lambdap(\eI))$ and that
$\lambdap(\,\eWpP{\wP}(\eI))\geq \wWpP{\wwP}(\lambdap(\eI))$.
Then, $\wWpP{\cP}(\lambdap(\eI)) = \lambdap(\,\eWpP{\wwP}(\eI))$.
(ii) and (iii) follow from Lemma~\ref{lem:gamma.lmabdaq}.\qed
\end{proof}

% \begin{definition}[\citeNP{damasio2013justifications}]\label{def:provenance}
% Let $P$ be a logic program and let $\wlfp$ and $\wgfp=\wWp(\wlfp)$ be the least and greates fixpoints of $\wWp^2$, respectively.
% Let $\rH$ be an atom.
% The why-not provenance information
% under the well-founded semantics is defined as follows:
% $Why_P (\rH) = [\wlfp(\rH)]$;
% $Why_P (\rH) = [\neg\wgfp(\rH)]$; and
% $Why_P (\Undef A) = [\neg\wlfp(\rH) ^ \wgfp(\rH)]$.\qed
% \end{definition}

\begin{proofof}
{Theorem~\ref{thm:prov.correspondence}}
Note that \mbox{$Why_\wwP(\rA)=\wlfpP{\wwP}(\rA)$} and that, by $\lambdap$ definition, it follows that $\lambdap(\botI) = \botI$ and thus, from Proposition~\ref{prop:aux:thm:prov.correspondence} (i),
it follows that
$\wWpP{\wwP}(\bot) = \wWpP{\wwP}(\lambdap(\botI)) = \lambdap(\eWpP{\wP}(\botI))$
and
\begin{gather*}
\wWpP{\wwP}(\bot)
		\ \ = \ \ \wWpP{\wwP}(\lambdap(\botI))
		\ \ = \ \ \lambdap(\eWpP{\wwP}(\botI))
		\ \ = \ \ \lambdap(\lambdaq(\eWpP{\wwP}(\botI)))
\end{gather*}
Hence, from Proposition~\ref{prop:aux:thm:prov.correspondence}, it follows that
% \vspace{-0.25cm}
\begin{align*}
\wWpP{\wwP}^2(\bot)
		\ \ = \ \ \wWpP{\wwP}(\wWpP{\wwP}(\bot))
		\ \ &= \ \ \wWpP{\wwP}(\lambdap(\lambdaq(\eWpP{\wwP}(\botI))))
		\\
		\ \ &= \ \ \lambdap(\eWpP{\wP}(\lambdaq(\eWpP{\wP}(\botI))))
		\ \ = \ \ \lambdap(\eWpP{\wP}(\eWpP{\wP}(\botI)))
		\ \ = \ \ \lambdap(\eWpP{\wP}^2(\botI))
\end{align*}
Inductively applying this reasoning it follows that
$\wWprP{\wwP}{\infty} = \lambdap(\eWprP{\wP}{\infty})$ which, by Knaster-Tarski theorem are the least fixpoints of the operators, that is, $\wlfpP{\wwP}=\lambdap(\elfpP{\wP})$
and, consequently,
\mbox{$Why_\wwP(\rA) = \wlfpP{\wwP}(\rA) = \lambdap(\elfpP{\wP}(\rA))= \lambdap(\ewfmP{\wP}(\rA)) = Why_\cP(\rA) $}.
Similarly, by definition, it follows that \mbox{$Why_\wwP(\Not \rA) = \neg\wgfpP{\wwP}(\rA)$} where $\wgfpP{\wwP}$ is the greatest fixpoint of the operator $\wWpP{\wwP}^2$. Thus,
\begin{align*}
Why_\wwP(\Not \rA) \ = \ \neg\wWpP{\wwP}(\wlfpP{\wwP})
	\ &= \ \lambdap(\sneg\eWpP{\wP}(\elfpP{\wP}))
	\ = \ \lambdap(\sneg\egfpP{\wP}(\rA))
	\ = \ \lambdap(\ewfmP{\wP}(\Not \rA))
\end{align*}
Finally, $Why_\wwP(\Undef \rA)\ = \ \neg\wlfpP{\wwP}(\rA) * \wgfpP{\wwP}(\rA)$ and, thus
\begin{align*}
Why_\wwP(\Undef \rA)
  \ &= \ \lambdap(\sneg\elfpP{\wP}(\rA)) * \lambdap(\sneg\sneg\egfpP{\wP}(\rA))
  \\
  \ &= \ \lambdap(\sneg \elfpP{\wP}(\rA) * \sneg\sneg\egfpP{\wP}(\rA))
  \\
  \ &= \ \lambdap(\sneg\ewfmP{\wP}(\rA) * \sneg\ewfmP{\wP}(\Not \rA))
  \ = \ \lambdap(\ewfmP{\wP}(\Undef \rA))
\end{align*}
and, thus, 
$Why_\wwP(\Undef \rA)
  % = \lambdap(\sneg\ewfmP{\wP}(\rA) * \sneg\ewfmP{\wP}(\Not \rA))
  = \lambdap(\ewfmP{\wP}(\Undef \rA))
  =  Why_\cP(\Not \rA)$.\qed
\end{proofof}

\subsection{Proof of Theorem~\ref{thm:well.non-positive.justifications.modify.program}}

\begin{lemma}\label{lem:least.label.removing-rules}
Let $P$ be a labelled logic program  over a signature~$\signature$ where $\lb$ is a finite set of labels and no rule is a labelled by $not(\rA)$ nor $\sneg\sneg not(\rA)$.
Let $Q$ be the result of removing all rules labelled by $\sneg not(\rA)$ for some atom $\rA$.
Let $\eI$ and $\eJ$ be two interpretations such that
$\eJ = \varrho_{not(\rA)}(\eI)$.
Then, $\eWpP{Q}(\eJ) = \varrho_{not(\rA)}(\eWp(\eI))$.\qed
\end{lemma}

\begin{proof}
In the sake of simplicity, we just write $\varrho$ instead of $\varrho_{not(\rA)}$.
By definition $\eWp(\eI)$ and $\eWpP{\cQ}(\eJ)$ are respectively the least model of $\ereduct{\cP}{\eI}$ and $\ereduct{\cQ}{\eJ}$.
The proof follows then by induction on the steps of the $T_P$ operator assuming that
$\varrho(\etprP{P^\eI}{\beta}) = \etprP{Q^\eJ}{\beta}$ for all $\beta < \alpha$.

Note that, $\etprP{X}{0}(\rH)=0$ for any program $X$ and atom~$\rH$ and, thus, the statement trivially holds.

In case that $\alpha$ is a successor ordinal.
Let $u \in \evalues$ be a join irreducible causal value such that $u \leq\etprP{P^\eI}{\alpha}(\rH)$.
Then, there is a rule in $P$ of the form~\eqref{eq:rule} such that
% \vspace{-0.25cm}
\begin{gather*}
u\ \ \leq \ \ (u_{B_1}*\dotsc*u_{B_m}*u_{C_1}*\dotsc*u_{C_n}) \cdot r_i
\end{gather*}
where 
\mbox{$u_{B_j}\leq\etprP{\cP^\eI}{\alpha-1}(B_j)$} and \mbox{$u_{C_j}\leq\sneg\eI(C_j)$} for each positive literal $B_j$ and each negative literal $\Not C_j$ in the body of rule $r_i$.

If $r_i = \sneg not(\rA)$, then
$\varrho(u) = 0 \leq \etprP{Q}{\alpha-1}(\rH)$.
Otherwise,
\begin{enumerate}
\item By induction hypothesis, it follows that $\varrho(u_{B_j}) \leq\etprP{Q}{\alpha-1}(\rB_j)$, and

\item from $J(\rH) = \varrho(I(\rH))$ and \mbox{$u_{C_j}\leq\sneg\eI(C_j)$}, it follows that $\varrho(u_{C_j}) \leq\sneg\eJ(\rC_j)$.
\end{enumerate}
Furthermore, no rule in the program $P$ is labelled with $not(\rA)$ nor $\sneg\sneg not(\rA)$ and, thus,
$r_i \neq not(\rA)$ and $r_i \neq \sneg\sneg not(\rA)$.
Hence, $\varrho(u) \leq\etprP{Q}{\alpha-1}(\rH)$.

The other way around is similar.
Since $u \leq\etprP{\cQ^\eJ}{\alpha}(\rH)$ there is a rule in $Q$ of the form~\eqref{eq:rule} such that
\begin{gather*}
u \ \ \leq \ \ (u_{B_1}*\dotsc*u_{B_m}*u_{C_1}*\dotsc*u_{C_n}) \cdot r_i
\end{gather*}
and \mbox{$u_{B_j}\leq\etprP{\cQ^\eJ}{\alpha-1}(B_j)$} and \mbox{$u_{C_j}\leq\sneg\eJ(C_j)$} for each positive literal $B_j$ and each negative literal $\Not C_j$ in the body of rule $r_i$.
By induction hypothesis,
\mbox{$u_{B_j}\leq\varrho(\etprP{\cP^\eI}{\alpha-1}(B_j))$}
for each $B_j$ with $1 \leq j \leq m$
and, since $J(\rH) = \varrho(I(\rH))$ and \mbox{$u_{C_j}\leq\sneg\eJ(C_j)$}, it follows that
$u_{C_j} \leq\varrho(\sneg\eI(\rC_j))$.
Then, $u \leq \varrho(\etprP{P^\eI}{\alpha}(\rH))$.

In case that $\alpha$ is a limit ordinal
$\etprP{X}{\alpha}=\sum_{\beta<\alpha}\etprP{X}{\beta}(\rH)$ and, thus,
$u \leq \etprP{X}{\alpha}$
if and only if
$u \leq \etprP{X}{\beta}(\rH)$ with $\beta < \alpha$.
By induction hypothesis,
$\varrho(\etprP{\cP^\eI}{\beta}(\rH)) = \etprP{\cQ^\eJ}{\beta}(\rH)$
and, thus,
$u \leq \varrho(\etprP{\cP^\eI}{\alpha})$
if and only if
$u \leq \etprP{\cQ^\eJ}{\alpha}$.
Hence, 
$\varrho(\etprP{\cP^\eI}{\alpha}) = \etprP{\cQ^\eJ}{\alpha}$
and, consequently,
$\eWpP{Q}(\eJ) = \varrho(\eWp(\eI))$.
\qed
\end{proof}

\begin{proposition}
\label{prop:eWp.least.fixpoint.label.removing-rules}
Let $P$ be a labelled logic program  over a signature~$\signature$ where $\lb$ is a finite set of labels where no rule is a labelled by $not(\rA)$ nor $\sneg\sneg not(\rA)$.
Let $Q$ be the result of removing all rules labelled by $\sneg not(\rA)$ for some atom $\rA$.
Then, $\elfpP{\cQ}=\varrho_{not(\rA)}(\elfpP{\cP})$ and
$\egfpP{\cQ}=\varrho_{not(\rA)}(\egfpP{\cP})$.\qed
\end{proposition}

\begin{proof}
% {Proposition~\ref{prop:eWp.least.fixpoint.label.removing-rules}}
Note that $\elfpP{X} = \eWprP{X}{\infty}$ with $X \in \set{\cP,\cQ}$.
Furthermore, by definition, it follows that
$\eWprP{\cP}{0}=\eWprP{\cQ}{0}=0$.
Then, assume as induction hypothesis that
$\eWprP{\cQ}{\beta} = \varrho(\eWprP{\cP}{\beta})$ for all $\beta<\alpha$.
When $\alpha$ is a successor ordinal, by definition
$\eWprP{X}{\alpha} = \eWWpP{X}(\eWprP{X}{\alpha-1}) = \eWpP{X}(\eWpP{X}(\eWprP{X}{\alpha-1}))$ with $X \in \set{\cP,\, \cQ}$ and, thus, the statement follows from Lemma~\ref{lem:least.label.removing-rules}.

In case that $\alpha$ is a limit ordinal
$\eWprP{X}{\alpha} = \sum_{\beta < \alpha} \eWprP{X}{\beta}$.
Then, for every join irreducible $u$ it follows that
$u\leq \eWprP{P}{\alpha}$ if and only if
$u\leq \eWprP{P}{\beta}$ for some $\beta < \alpha$
(by induction hypothesis) iff
$\varrho(u)\leq \eWprP{\cP}{\beta}$ iff
$\varrho(u)\leq \eWprP{\cP}{\alpha}$.
Hence, $\eWprP{\cQ}{\alpha}=\varrho(\eWprP{\cP}{\alpha})$ and, conseuqntly,
$\elfpP{\cQ}= \varrho(\elfp)$

Finally, note that $\egfpP{X}= \eWpP{X}(\elfpP{X})$ with $X \in \set{\cP,\,\cQ}$ and, thus, the statement follows directly from Lemma~\ref{lem:least.label.removing-rules}.\qed
\end{proof}

\begin{proofof}{Theorem~\ref{thm:well.non-positive.justifications.modify.program}}
By definition, program~$\cP$ is the result of removing all rules labelled with $\sneg not(A)$ in $\wP$.
In case that $\rL$ is \review{R1.17}{some atom $\rH$}, by definition, it follows that
$\ewfm(\rH) = \elfp(\rH)$
and
$\ewfmP{\wP}(\rH) = \elfpP{\wP}(\rH)$
and, from Proposition~\ref{prop:eWp.least.fixpoint.label.removing-rules},
it follows that
$\elfpP{\cP}=\varrho(\elfpP{\wP})$
and, thus
$\ewfmP{\cP}=\varrho(\ewfmP{\wP})$.

Similarly, in case that $\rL$ is a negative literal ($\rL=\Not \rH$), then
$\ewfm(\rH) = \sneg\egfp(\rH)$
and
$\ewfmP{\wP}(\rH) = \sneg\egfpP{\wP}(\rH)$
and, from Proposition~\ref{prop:eWp.least.fixpoint.label.removing-rules},
it follows that
$\egfpP{\cP}=\varrho(\egfpP{\wP})$.
Just note tha $\varrho_{x}(\sneg u) = \sneg\varrho_{x}(u)$ for any elementary term $x$ and any value $u$.
Hence,
$\egfpP{\cP}=\varrho(\egfpP{\wP})$
implies that
$\sneg\egfpP{\cP}=\varrho(\sneg\egfpP{\wP})$
and, consequently,
$\ewfmP{\cP}=\varrho(\ewfmP{\wP})$.

In case that $\rL$ is an undefined literal ($\rL=\Undef \rH$), by definition, it follows that
$\ewfm(\rH) = \sneg\ewfm(\rH)*\sneg\ewfm(\Not \rH)= \sneg\elfp(\rH)*\sneg\sneg\egfp(\rH)$
and
$\ewfmP{\wP}(\rH) = \sneg\elfpP{\wP}(\rH)*\sneg\sneg\egfpP{\wP}(\rH)$
and the result follows as before from  Proposition~\ref{prop:eWp.least.fixpoint.label.removing-rules}.
\qed
\end{proofof}

\subsection{Proof of Theorem~\ref{thm:wellf.justification<-$>$weff.provenace}}

\begin{proofof}{Theorem~\ref{thm:wellf.justification<-$>$weff.provenace}}
% {Corollary~\ref{thm:wellf.justification<-$>$weff.provenace} }
Note that $\varrho(\lambdap(u))=\lambdap(\varrho(u))$ for any causal value $u \in \evalues$.
By definition $Why_\cP(\rL) = \lambdap(\ewfmP{\wP})(\rL)$ and, thus
\begin{gather*}
\varrho(Why_\cP(\rL))
  \ \ =\ \ \varrho(\lambdap(\ewfmP{\wP})(\rL))
  \ \ =\ \ \lambdap(\varrho(\ewfmP{\wP}))(\rL)
\end{gather*}
From Theorem~\ref{thm:well.non-positive.justifications.modify.program},
it follows that $\ewfm=\varrho(\ewfmP{\wP})$ and, thus,
$\varrho(Why_\cP(\rL)) = \lambdap(\ewfm)(\rL)$.\qed
\end{proofof}

\subsection{Proof of Theorem~\ref{thm:wellf.justification<-$>$weff.standard}}
\label{proof:thm:wellf.justification<-$>$weff.standard}

The proof of Theorem~\ref{thm:wellf.justification<-$>$weff.standard} will rely on the relation between ECJ justifications and non-hypothetical WnP justifications established by Theorem~\ref{thm:wellf.justification<-$>$weff.provenace} plus the following result from~\cite{damasio2013justifications}.
First, we need some notation.
Given a conjuntion of labels $D$, by $Remove(D)$ we denote the set of negated labels in $D$, by $Keep(D)$ the set of positive labels, by $AddFacts(D)$ the set of facts $\rA$ such that $\neg not(\rA)$ occurs in $D$ and by $NoFacts(D)$ the set of facts $\rA$ such that $not(\rA)$ occurs in $D$.

\begin{theorem}[Theorem~3 from \citeNP{damasio2013justifications}]
\label{thm:prov.models}
Given a labelled logic program $\cP$, let $N$ be a set of facts not in program $\cP$ and $R$ be a subset of rules of $\cP$. A literal $L$ belongs to the $WFM$ of $(\cP \backslash R) \cup N$ iff there is a conjunction of literals \mbox{$D \models Why_\cP(L)$}, such
that
\mbox{$Remove(D) \subseteq R$},
\mbox{$Keep(D) \cap R = \emptyset$},
\mbox{$AddFacts(D) \subseteq N$}, and
\mbox{$NoFacts(D) \cap N = \emptyset$}.
\qed
\end{theorem}

% \begin{lemma}
% Let $D$ be positive conjunction of labels and $t$ be a term such that $D \leq \lambdap(t)$.
% Let $E = \prod_{x \cdotl x' \in X} x \cdotl x'$ be a justification where
% \begin{gather*}
% X \ \ = \ \ \setm{ x \cdotl x' }{ x \text{ and } x' \text{ are labels or negated labels in } D \text{ after replacing $\sneg$ by $\neg$}}
% \end{gather*}
% Then, $E \leq t$.\qed
% \end{lemma}

% \begin{proof}
% Let
% $E' = \prod_{y \cdotl y' \in Y} y \cdotl y'$ be a justification in DNF such that $\lambdap(E') = D$ and
% suppose that $E \not\leq E'$.
% Then, there is some $y \cdotl y' \in Y$ such that $x \cdotl x' \not\leq y \cdotl y'$ for any 
% $x \cdotl x' \in X$.
% Thus, either $y$ or $y'$ do not occur in $D$ and, thus, $\lambdap(E') \neq D$.
% Hence, $E$ is the minimum justification that $\lambdap(E) = D$.

% Suppose that $E \not\leq t$.
% Then, there

% Since $D \leq \lambdap(t)$, then $\lambdap(t) = \lambdap(t) + D$.
% Furthermore, $D = \lambdap(E)$ and, thus, $\lambdap(t) = \lambdap(t) + \lambdap(E)$

% Suppose not and let $t$ be term in DNF such that $\varrho(D) \leq \lambdap(t)$ and $E \not\leq t$.
% Since $E \not\leq t$,
% there is a term $u$ containing all labels such that $u \leq E$, but $u \not\leq t$.

% it follows that $E \not\leq t'$
% for every term with no-sums $t' \in T$ with $t = \sum_{t' \in T} t'$.
% \end{proof}

\begin{definition}
Given a positive program $\cP$, we define a direct consequence operator~$\stpP{\cP}$ such that
\begin{align*}
\stpP{\cP}(\sI)(\rH)
    \ \ &\eqdef \ \
      \sum \big\{ \ \sI(B_1) * \dotsc * \sI(B_n)  \
% \\&\hspace{2.5cm}
        \mid \ (r_i : \ \rH \leftarrow B_1, \dotsc, B_n ) \in P \ \big\}
\end{align*}
for any standard interpretation interpretation $\sI$ and atom $\rH \in \at$.\qed
\end{definition}

\begin{lemma}\label{lem:gamma.causal-$>$standard}
Let $P$ be a labelled logic program  over a signature~$\signature$ where $\lb$ is a finite set of labels and let $\eI$ and $\sI$ be respectively a ECJ and a standard interpretation satisfying that there is some enable justification $E \leq \sneg\eI(\rH)$ for every atom $\rH$ such that $\sI(\rH) = 0$.
Then, every atom $\rH$ satisfies $\sWpP{\cP}(\sI)(\rH)=1$ iff there is some enabled justification $E \leq \eWpP{\cP}(\eI)(\rH)$.\qed
\end{lemma}

\begin{proof}
By definition $\eWpP{\cP}(\eI)$ and $\sWpP{\cP}(\sI)$ are the least model of the programs
$\ereduct{\cP}{\eI}$ and $\ereduct{\cP}{\sI}$, respectively.
Furthermore, the least model of programs $\cP^\eI$ and $\cP^\sI$ are the least fixpoint of the $T_P$ and $\stp$ operators,
that is,
$\eWp(\eI)=\etprP{\ereduct{\cP}{\eI}}{\omega}$
and
$\sWp(\eJ)=\stprP{\ereduct{\cP}{\sI}}{\omega}$.
In case that $\alpha=0$, it follows that 
$\stprP{\cP^\sI}{0}(\rH)$ for every atom $\rH$ and, thus, the statement holds vacuous.
We assume as induction hypothesis that
for every atom $\rH$ and ordinal $\beta<\alpha$ such that $\stprP{\ereduct{\cP}{\sI}}{\beta}(\rH)=1$, there is some enabled justification $E \leq \etprP{\ereduct{\cP}{\eI}}{\beta}(\rH)$.

In case that $\alpha$ is a successor ordinal.
If $\stprP{\ereduct{\cP}{\eI}}{\alpha-1}(\rH)=1$, then
there is a rule $r_i\in\cP$ of the form \eqref{eq:rule}
such that $\stprP{\ereduct{\cP}{\eI}}{\alpha-1}(\rB_j)=1$ and
$\eI(\rC_j)=0$. 
On the one hand, by induction hypothesis,
it follows that 
there is some enabled justification $E_{B_j} \leq \etprP{\ereduct{\cP}{\eI}}{\alpha-1}(\rB_j)$
and, by hypothesis,
there is some enabled justification $E_{C_j} \leq \sneg I(C_j)$.
Hence, 
\begin{gather*}
E \ \ \eqdef \ \ (E_{B_1} * \dotsc E_{B_m} * E_{C_1} * \dotsc * E_{C_n} ) \cdotl r_i
\end{gather*}
is an enabled justification $E \leq\etprP{\ereduct{\cP}{\eI}}{\alpha}(\rH)$.

The other way around, let $E$ be some join irreducible justification.
If $E \leq  \etprP{\ereduct{\cP}{\eI}}{\alpha}(\rH)$,
then
there is a rule $r_i\in\cP$ of the form \eqref{eq:rule}
such that 
\begin{gather*}
E \ \ \leq \ \ (E_{B_1} * \dotsc E_{B_m} * E_{C_1} * \dotsc * E_{C_n} ) \cdotl r_i
\end{gather*}
where
$E_{B_j} \leq  \etprP{\ereduct{\cP}{\eI}}{\alpha}(\rB_j)$ and
$E_{C_j} \leq  \sneg\eI(\rC_j)$ are enabled justifications.
Hence, it follows that \mbox{$\stprP{\ereduct{\cP}{\sI}}{\alpha}(\rB_j)=1$} and $\sI(\rC_j)=0$.

In case that $\alpha$ is a limit ordinal,
$\stprP{\cP^\cI}{\alpha} = 1$ iff $\stprP{\cP^\cI}{\beta}=1$ for some $\beta < \alpha$
iff there is a join irreducible enabled justification
$E \leq \etprP{\cP^\eI}{\beta}) \leq \lambdap(\etprP{\cP^\eI}{\alpha}$.
\qed
\end{proof}

\begin{proofof}
{Theorem~\ref{thm:wellf.justification<-$>$weff.standard}}
Let $E\leq\ewfm(\rL)$ be an enabled justification of $\rL \in \set{ \rA,\, \Not \rA,\, \Undef \rA}$.
From Theorem~\ref{thm:wellf.justification<-$>$weff.provenace}, it follows that
$\lambdap(E) \leq \lambdap(\ewfm(\rL)) = \varrho(Why_\cP(\rL))$,
that is, 
$\lambdap(E) \leq \varrho(Why_\cP(\rL))$.
Note that the minimum causal value $t$ such that $\varrho(t) = \varrho(Why_\cP(\rL))$ is
$Why_\cP(\rL) \wedge \bigwedge_{\rA \in \at} not(\rA)$
and, thus, $D \leq Why_\cP(\rL)$ 
where $D$ is defined by
$D = \lambdap(E) \wedge \bigwedge_{\rA \in \at} not(\rA)$.
Furthermore, since $E$ is an enabled justification,
$\lambdap(E)$ is a positive conjunction and, thus,
so it is $D$.
Hence, there is a positive conjunction $D$ such that $D \leq Why_\cP(\rL)$ and, from Theorem~\ref{thm:prov.models}, it follows that $\rL$ holds with respect to the standard WFM of $P$.

The other way around. 
If $\rL=\rA$ is an atom, then $\rL$ holds with respect to the standard WFM iff $\lfp(\sWp^2)(\rL)=1$.
Furthermore,
$\sWpr{0}(\rH)=\eWpr{0}=0$ for any atom $\rH$ and, thus, there is an enabled justification $E \leq \sneg \eWpr{0}= \sneg 0 = 1$ for any atom $\rH$.
Then, from Lemma~\ref{lem:gamma.causal-$>$standard}, 
for any atom $\rH$ ,
there is an enabled justification $E \leq \eWp(\eWpr{0})(\rH)$ iff
\mbox{$\sWp(\sWpr{0})(\rH)=1$}.
Applying this result again, it follows that
$E \leq \eWWpr{1}(\rH)=\eWp^2(\eWpr{0})(\rH)$ if and only if
$\sWpr{1})(\rH)=\sWp^2(\sWpr{0})(\rH)=1$.
Inductively applying this reasoning it follows that
$\sWpr{\infty}(\rH) = 1$ iff there is an enabled justification
$E \leq \eWpr{\infty}(\rH)$ which, by Knaster-Tarski theorem are the least fixpoints respectively of the $\sWp$ and $\eWp$ operators.

Similarly, if $\rL = \Not \rA$, then $\rL$ holds with respect to the standard WFM if and only if $\gfp(\sWp^2)(\rL)=\sWp(\lfp(\sWp^2))(\rL)=0$ iff there is not any an enabled justification $E \leq \eWp(\lfp(\eWp^2))(\rL)=\gfp(\eWp^2)(\rL)$ iff there is an enabled justification
$E \leq \wfm(L) = \sneg \gfp(\eWp^2)(\rL)$.

Finally, if $\rL = \Undef \rA$, then $\rL$ holds with respect to the standard WFM iff $\lfp(\sWp^2)(\rL)=0$ and $\gfp(\sWp^2)(\rL)=1$
if and only if
there is not any enabled justification
$E \leq \wfm(\rL)$
and
there is not any enabled justification
$E \leq \wfm(\Not \rL)$
iff
there is some enabled justification
$E \leq \sneg\wfm(\rL)$
and
there is some enabled justification
$E \leq \sneg\wfm(\Not \rL)$
iff
there is some enabled justification
$\ewfm(\Undef \rA)=\sneg\ewfm(\rA) * \sneg\ewfm(\Not \rA)$.\qed
\end{proofof}

\subsection{Proof of Theorem~\ref{thm:wellf.why-not-$>$sm.cjustification}}

\begin{lemma}\label{lem:removing.monotonicity}
Let $t$ and $u$ be two causal terms such that no-sums occur in $t$ ant $t \leq u$.
Then, $\varrho_{x}(t) \leq \varrho_x(u)$.\qed
\end{lemma}

\begin{proof}
By definition $t \leq u$ if and only if $t = t * u$.
Then, $\varrho_x(t) = \varrho_x(t*u) = \varrho_x(t) * \varrho_x(u)$
and, thus if follows that $\varrho_x(t) \leq \varrho_x(u)$.
\qed
\end{proof}

\begin{lemma}\label{lem:lambdac.lambdap.leq}
Let $t$ be a causal term.
Then, $\lambdac(\lambdap(t)) \leq \lambdap(\lambdac(t))$.
\qed
\end{lemma}

\begin{proof}
If $t \in \lb$ is a label, then $\lambdac(t) = t$ and $\lambdap(t) = t$ and, thus,
$\lambdac(\lambdap(t)) = t \leq t = \lambdap(\lambdac(t))$.
If $t  = \sneg l$ with $l \in \lb$ a label, then
$\lambdac(t) = 0$ and $\lambdap(t) = \neg l$ and, thus,
$\lambdac(\lambdap(t)) = 0 \leq 0 = \lambdap(\lambdac(t))$.
If $t  = \sneg\sneg l$ with $l \in \lb$ a label, then
$\lambdac(t) = 1$ and $\lambdap(t) = l$ and, thus,
$\lambdac(\lambdap(t)) = l \leq 1 = \lambdap(\lambdac(t))$.

Assume as induction hypothesis that
$\lambdac(\lambdap(u)) \leq \lambdap(\lambdac(u))$ for every subterm $u$ of $t$.
If $t = u_1 \cdotl u_2$,
then
\begin{align*}
\lambdac(\lambdap(u_1 \cdotl u_2))
    \ \ = \ \ \lambdac(\lambdap(u_1) * \lambdac(\lambdap(u_2)
    \ \ \leq \ \ \lambdap(\lambdac(u_1) * \lambdap(\lambdac(u_2)
    \ \ = \ \ \lambdap(\lambdac(u_1 \cdotl u_2))
\end{align*}
Similarly, if $t = \sum_{u \in U} u$, then
\begin{align*}
\lambdac(\lambdap(\sum_{u \in U} u)
    \ \ = \ \ \sum_{u \in U}\lambdac(\lambdap(u)
    \ \ \leq \ \ \sum_{u \in U}\lambdap(\lambdac((u))
    \ \ = \ \ \lambdap(\lambdac(\sum_{u \in U} u))
\end{align*}
and
if $t = \prod_{u \in U} u$, then
\begin{align*}
\lambdac(\lambdap(\prod_{u \in U} u)
    \ \ = \ \ \prod_{u \in U}\lambdac(\lambdap(u)
    \ \ \leq \ \ \prod_{u \in U}\lambdap(\lambdac((u))
    \ \ = \ \ \lambdap(\lambdac(\prod_{u \in U} u))
\end{align*}
\qed
\end{proof}

\begin{proofof}
{Theorem~\ref{thm:wellf.why-not-$>$sm.cjustification}}
From Theorem~\ref{thm:wellf.justification<-$>$weff.provenace},
it follows that
$\varrho(Why_P(\rA)) = \lambdap(\wfm)(\rA)$.
Furthermore, 
since $D \leq Why_P(\rA)$,
from Lemma~\ref{lem:removing.monotonicity}, it follows that
\begin{gather*}
\varrho(D) \ \ \leq \ \ \varrho(Why_P(\rA)) \ \ = \ \ \lambdap(\wfm)(\rA) \ \ = \ \ \lambdap(\elfp)(\rA)
\end{gather*}
and, thus, $\lambdac(\varrho(D)) \leq \lambdac(\lambdap(\elfp))(\rA)$.
Let $\cI$ be any CG stable model.
Then, since $\cI=\lambdac(\eI)$ for some fixpoint $\eI$ of $\eWp^2$,
it follows that $\lambdac(\elfp) \leq \cI$
and, thus, $\lambdap(\lambdac(\elfp)) \leq \lambdap(\cI)$.
Furthermore, from Lemma~\ref{lem:lambdac.lambdap.leq}, it follows that
$\lambdac(\lambdap(\elfp)) \leq \lambdap(\lambdac(\elfp))$
and, thus
\begin{gather*}
\lambdac(\varrho(D))
	 \ \ \leq \ \ \lambdac(\lambdap(\elfp))(\rA)
	 \ \ \leq \ \ \lambdap(\lambdac(\elfp))(\rA)
	 \ \ \leq \ \ \lambdap(\cI)(\rA)
\end{gather*}
Note that, since $D$ is non-hypothetical and enabled, it does not contain negated labels and, thus, $\lambdac(\varrho(D)) = \varrho(D)$.
Consequently, $\varrho(D) \leq \lambdap(\cI)(\rA)$.\qed
\end{proofof}

\newpage

%%%%%%%%%%%%%%%%%%%%%%%%%%%%%%%%%%%%%%%%%%%%%%%%%%%%%%%%%%%%%%%%%%%%%%%%%%%%%%%

\end{document}